\documentclass[11pt,a4paper]{amsart}
\usepackage[foot]{amsaddr}
\usepackage{ifxetex}
\ifxetex
  \usepackage[no-math]{fontspec}
\else
\fi
\usepackage{amsmath}
\usepackage{amsfonts}
\usepackage{amssymb}
\usepackage{amsthm}
\usepackage{fullpage}
\usepackage{microtype}
\ifxetex 
  \usepackage[libertine]{newtxmath}
\else
  \usepackage{newtxmath}
\fi
\usepackage[tt=false]{libertine} 
\usepackage{caption}
\usepackage{tcolorbox}
\usepackage{bbm}
\usepackage{hyperref, color}
\hypersetup{colorlinks=true,citecolor=blue, linkcolor=blue, urlcolor=blue}
\usepackage[linesnumbered,boxed,ruled,vlined]{algorithm2e}
\usepackage{bm}
\usepackage{bbm}
\usepackage[numbers]{natbib}
\usepackage{xcolor}
\usepackage{enumerate} 
\usepackage{enumitem}
\usepackage{ragged2e}
\usepackage{tabularx}
\usepackage{array}
\usepackage{longtable}
\usepackage{hhline}
\usepackage{multirow}
\usepackage{cleveref}
\newcolumntype{L}[1]{>{\raggedright\arraybackslash}p{#1}}
\newcolumntype{C}[1]{>{\centering\arraybackslash}m{#1}}
\newcolumntype{R}[1]{>{\raggedleft\arraybackslash}p{#1}}

\usepackage{makecell}
\usepackage{footnote}
\usepackage{float}
\makesavenoteenv{tabular}

\renewcommand{\epsilon}{\varepsilon}

\newtheorem{theorem}{Theorem}[section]
\newtheorem{observation}[theorem]{Observation}

\newtheorem*{claim*}{Claim}
\newtheorem{condition}[theorem]{Condition}

\newtheorem{lemma}[theorem]{Lemma}
\newtheorem{proposition}[theorem]{Proposition}
\newtheorem{corollary}[theorem]{Corollary}
\theoremstyle{definition}

\newtheorem{definition}[theorem]{Definition}
\newtheorem{remark}[theorem]{Remark}
\newtheorem*{remark*}{Remark}

\def\Ex{\mathop{\mathbf{E}}\nolimits}

\renewcommand{\emptyset}{\varnothing}
\renewcommand{\Pr}[2][]{ \ifthenelse{\isempty{#1}}
  {\mathop{\mathbf{Pr}}\left[#2\right]} {\mathop{\mathbf{Pr}}_{#1}\left[#2\right]} }

\newcommand{\abs}[1]{\left\vert#1\right\vert}
\newcommand{\set}[1]{\left\{#1\right\}}
\newcommand{\tuple}[1]{\left(#1\right)} 
\newcommand{\ceil}[1]{\lceil #1\rceil}

\newcommand{\eps}{\varepsilon}
\newcommand{\tp}{\tuple}

\newcommand{\numP}{\ensuremath{\#\mathbf{P}}}
\renewcommand{\P}{\ensuremath{\mathbf{P}}}
\newcommand{\NP}{\ensuremath{\mathbf{NP}}}

\newcommand{\defeq}{:=}

\newcommand{\DTV}[2]{d_{\mathrm{TV}}\left({#1},{#2}\right)}

\def\*#1{\bm{#1}} 
\def\+#1{\mathcal{#1}} 
\def\-#1{\mathrm{#1}} 
\def\=#1{\mathbb{#1}} 

\usepackage[textsize=tiny]{todonotes}

\usepackage{xifthen}

\makeatletter
\def\prob#1#2#3{\goodbreak\begin{list}{}{\labelwidth\z@ \itemindent-\leftmargin
                        \itemsep\z@  \topsep6\p@\@plus6\p@
                        \let\makelabel\descriptionlabel}
                \item[\it Name]#1
               \item[\it Instance]                #2
                \item[\it Output]#3
                \end{list}}
\makeatother

\newcommand{\poly}{{\rm poly}}  
\newcommand{\basicsample}{\textnormal{\textsf{LB-Sample}}}

\newcommand{\appmghinds}{\textnormal{\textsf{ApproxMarginIndSet}}}

\newcommand{\appmghcol}{\textnormal{\textsf{ApproxMarginColouring}}}
\newcommand{\resolve}{\textnormal{\textsf{Resolve}}}
\newcommand{\appresolve}{\textnormal{\textsf{ApproxResolve}}}
\newcommand{\appresolverandom}{\textnormal{\textsf{ApproxResolve-Random}}}

\newcommand{\config}{\mathsf{Boundary}}
\newcommand{\appmghcoledge}{\textnormal{\textsf{ApproxMarginColouring-Set}}}
\newcommand{\upd}{\textnormal{\textsf{pred}}}

\newcommand{\ts}{\textnormal{\textsf{TS}}}

\newcommand{\vbl}{\mathsf{vbl}}
\newcommand{\var}[1]{{{\vbl}}\left({#1}\right)}  
\usepackage{enumitem}

\newcommand{\ssmstr}{\textnormal{\textsf{SSMS-Truncated}}}
\newcommand{\bdsplittr}{\textnormal{\textsf{BD-SPLIT-Truncated}}}

\newcommand{\dist}{\mathrm{dist}}

\newcommand{\KL}[2]{\ensuremath{D_{\textnormal{KL}}\left(#1\parallel#2\right)}}

\newcommand{\randombit}{\basicsample}

\newbool{doubleblind}
\setbool{doubleblind}{false}

\title{Towards derandomising Markov chain Monte Carlo}

\ifdoubleblind
  \author{Author(s)}
\else
\author{Weiming Feng, Heng Guo, Chunyang Wang, Jiaheng Wang, Yitong Yin}

\address[Weiming Feng, Heng Guo, Jiaheng Wang]{School of Informatics, University of Edinburgh, Informatics Forum, Edinburgh, EH8 9AB, United Kingdom. \textnormal{E-mail: \url{wfeng@ed.ac.uk}, \url{hguo@inf.ed.ac.uk}, \url{jiaheng.wang@ed.ac.uk}}. }
\address[Chunyang Wang, Yitong Yin]{State Key Laboratory for Novel Software Technology, Nanjing University, 163 Xianlin Avenue, Nanjing, Jiangsu Province, 210023, China. \textnormal{E-mail: \url{wcysai@smail.nju.edu.cn}, \url{yinyt@nju.edu.cn}}}
\fi

\begin{document}


\begin{abstract}
  We present a new framework to derandomise certain Markov chain Monte Carlo (MCMC) algorithms.
  As in MCMC, we first reduce counting problems to sampling from a sequence of marginal distributions.
  For the latter task,
  we introduce a method called \emph{coupling towards the past} that can, in logarithmic time, 
  evaluate one or a constant number of variables from a stationary Markov chain state.
  Since there are at most logarithmic random choices, this leads to very simple derandomisation.
  We provide two applications of this framework, namely efficient deterministic approximate counting algorithms for hypergraph independent sets and hypergraph colourings, 
  under local lemma type conditions matching, up to lower order factors, their state-of-the-art randomised counterparts.
\end{abstract}	

\maketitle

\setcounter{tocdepth}{1}
\tableofcontents

\section{Introduction}

It is a central question in the theory of computing to understand the power of randomness.
Indeed, randomisation has been shown to be extremely useful in designing efficient algorithms.
One early and surprising illustration of its power is through estimating the volume of a convex body.
Deterministic algorithms cannot achieve good approximation ratios through membership queries \cite{elekes1986geometric, BF87},
and yet Dyer, Frieze, and Kannan \cite{DFK91} discovered an efficient randomised approximation algorithm under the same model.
Their driving force is the celebrated Markov chain Monte Carlo (MCMC) method,
which has been studied since the origin of electronic computers.
In MCMC, one reduces the counting problem (such as computing the volume) to (usually a sequence of) related sampling problems \cite{JVV86},
and the latter is solved using Markov chains.
This powerful method has lead to many great achievements,
ranging from the early results of approximating the partition function of ferromagnetic Ising models \cite{JS93} and the permanent of non-negative matrices \cite{jerrum2004polynomial},
to more recent developments such as counting the number of bases in matroids \cite{ALOV19,CGM19},
and estimating partition functions of spin systems up to critical thresholds \cite{ALO20, chen2020rapid, CLV21, chen2021rapid, anari2022entropic, CE22,CFYZ22}.

While randomness remains an indispensable ingredient to the MCMC method,
the belief that randomness is essential to efficient approximate counting is seriously challenged over the past two decades.
This was initiated by a highly influential result of Weitz \cite{Wei06}, which gave the first deterministic fully polynomial-time approximation scheme (FPTAS) for a $\numP$-hard problem.
Since then, deterministic algorithms are gradually catching up with their random counterparts on many fronts.
A plethora of techniques have been introduced and developed for deterministic approximate counting, including:
decay of correlation \cite{Wei06,bayati2007simple,gamarnik2007correlation}, 
zero-freeness of polynomials \cite{barvinok2016combinatorics,patel2017deterministic, LSS22}, 
linear programming based methods \cite{Moi19,GLLZ19,JPV21a},
and various statistical physics related techniques \cite{HPR19,JPP22,JPSS22}.
Curiously, although these algorithms are usually guided by probabilistic intuitions related to the problem,
they work very differently from typical randomised algorithms.

All these developments beg one question: can we derandomise MCMC more directly?
The benefit of a positive answer would be two folds.
It can help us to better understand the role of randomisation in MCMC and in approximate counting.
It might also lead to easier ways of designing deterministic counting algorithms thanks to the plentifulness of MCMC algorithms.
However, this task is not easy, as Markov chains require at least a linear amount of random bits to approach their stationary distribution.
Short of a breakthrough in pseudo-random number generators,
it seems impossible to fully derandomise Markov chains.

In this paper we make a substantial step towards turning rapid mixing Markov chains into deterministic approximate counting algorithms.
The saving grace to the issue mentioned above is that the Monte Carlo step of MCMC does not require fully simulating the Markov chains anyways!
Usually, only the marginal probability of a single variable needs to be evaluated, rather than the whole state.
We show that for certain MCMC algorithms,
one or a constant number of variables can be evaluated in their stationary state within only logarithmic time.
This leads to some very simple brute force enumeration derandomisation.
To illustrate the power of this new method, we obtain efficient deterministic approximate counting algorithms for hypergraph independent sets and hypergraph colourings.
Our algorithms match their currently best randomised counterparts under local lemma type conditions up to log factors.
We describe our results in more details next.

\subsection{Our contributions}
We give deterministic approximate counting algorithms by derandomising certain MCMC algorithms. 
In fact, we venture the idea that efficient deterministic approximate counting follows from ``highly efficient'' randomised algorithms.
We first explain what these are and how we obtain them. 
%
 
\subsubsection{Coupling towards the past}
We consider a class of generic Markov chains known as \emph{systematic scan Glauber dynamics},\footnote{Instead of random update Markov chains, we choose systematic scan chains to illustrate our ideas, as the analysis is much cleaner and the result is no worse.
Our method applies to random update Markov chains as well, but with some additional effort. See \Cref{sec:intro-random} for a brief discussion and \Cref{sec:derandomise-random-scan} for the details.}
which are single-site {Glauber dynamics} (a.k.a.~\emph{Gibbs sampling}, \emph{heat-bath dynamics}) 
with a fixed scan order.
In classical MCMC, 
a Markov chain $(X_t)_{t\ge 0}$ is simulated chronologically for long enough to draw a sample according to the stationary distribution $X_\infty\sim\mu$.
Due to a lower bound of Hayes and Sinclair \cite{hayes2007general}, 
which applies to a wide range of models,
this will cost $\Omega(n\log n)$ random bits if $X_t=(X_t(v))_{v\in[n]}$ consists of $n$ variables.
However, this simulation seems very wasteful when we are only interested in evaluating $X_\infty(v)$ for some particular $v\in[n]$.
This reflects a more general question: can we calculate the fixed point of a dynamical system without simulating the dynamical system until convergence?
%

We introduce a new method for evaluating the states of variables in a stationary Markov chain without simulating the entire chain. 
We call this method \emph{coupling towards the past} (CTTP).\footnote{The name resembles the coupling \emph{from} the past (CFTP) method by Propp and Wilson~\cite{PW96}. However, our method has several key differences from CFTP. See \Cref{sec:related-work} for a comparison.}
Consider a mixed chain $(X_t)_{t\le 0}$ that runs from the infinite past to now.
Our goal is to evaluate $X_0(v)$ for $v\in[n]$, which follows the marginal distribution, denoted $\mu_v$.
%
%
It suffices to simulate the last update for $v$, 
and the key observation here is that updates of Glauber dynamics may depend only on a small amount of information.
In particular, when the marginal distributions are properly lower bounded,
there is a positive probability of determining the update without any information on the current state.
Thus, we can deduce $X_0(v)$ by recursively revealing only the necessary randomness backwards in time.
Alternatively, this can be viewed as a grand coupling constructed towards the past.
Each random bit revealed is used for all possible chains,
and every information successfully deduced is the same for all chains as well.
When this process terminates at time $-t$ for some $t\ge 0$,
no matter what the state before $-t$ is, it evolves into the same $X_0(v)$.
This implies that $X_0(v)$ follows $\mu_v$ as desired.

%
When the CTTP process terminates in logarithmic steps with high probability, 
it yields a \emph{marginal sampler} that draws, approximately, from the marginal distribution with $O(\log n)$ cost in both time and random bits.
The error in total variation distance is the failure probability of CTTP.
Such marginal samplers can be straightforwardly derandomised by enumerating all possible random choices in polynomial time to deterministically estimate the marginal probabilities,
which implies FPTASes via standard self-reductions \cite{JVV86}.

We will apply CTTP to the uniform distribution of hypergraph independent sets and a projected distribution induced by uniformly at random hypergraph colourings.
In both applications, the conditions we obtain to guarantee the $O(\log n)$ run-time of CTTP match those for $O(n\log n)$ mixing time bounds of Glauber dynamics.
From a complexity-theory point of view,
our construction of marginal samplers from the original Markov chain algorithm has a certain \emph{direct-sum} flavour: 
we transform a protocol for sampling $n$ variables,
to a new protocol for sampling single variables, using only $O(1/n)$ fraction of the original cost each.

We remark that these low-cost marginal samplers have significance beyond deterministic approximate counting.
For example, in probabilistic inferences,
it is often the end goal to estimate marginal probabilities.
Thus, our marginal samplers are substantially faster in such contexts than standard MCMC.
As an instance of such, the perfect marginal sampler we obtain for hypergraph independent sets terminates in $O(\log 1/\gamma)$ time with probability at least $1-\gamma$ in the regime of parameters we consider (see \Cref{rem:perfect-marginal-HIS}). 
In comparison, other standard methods for such purpose require $\Omega(n)$ running time where $n$ is the number of vertices. 
See \Cref{sec:related-work} for a more detailed elaboration on this. 

\subsubsection{Counting hypergraph independent sets}
The first testing field of our framework is to approximately count the number of hypergraph independent sets (equivalent to satisfying assignments of monotone CNF formulas).
Let $H=(V,\+E)$ be a hypergraph.
A set $S\subseteq V$ is a (weak) \emph{independent set} if $S\cap e\neq e$ for all $e\in \+E$.
This problem is naturally parameterised by $k$ and $\Delta$, which denote the (uniform) hyperedge size and the maximum vertex degree of $H$, respectively.
There was an exponential gap between the parameters for which efficient randomised and deterministic algorithms exist,
and our work closes this gap.

Estimating the number of hypergraph independent sets was first studied by Borderwich, Dyer, and Karpinski \cite{BDK06,BDK08} using Markov chains.
They used path coupling to show that the straightforward Glauber dynamics mixes in $O(n\log n)$ time if $\Delta\le k-2$.
The mixing time analysis was later improved by Hermon, Sly, and Zhang \cite{HSZ19} to extend the condition exponentially to $\Delta\le c2^{k/2}$ for some absolute constant $c>0$ via an information percolation argument.
This threshold is tight up to constants due to a hardness result in \cite{BGGGS19}.
Very recently, Qiu, Wang, and Zhang \cite{qiu2022perfect} gave a perfect sampler under similar conditions.

On the other hand, deterministic counting algorithms have been lagging behind for this problem.
Even with significant adjustment, the correlation decay method works only if $\Delta\le k$ \cite{BGGGS19}.
One may also put it under the Lov\'asz local lemma framework,
and apply a more general algorithm by He, Wang, and Yin~\cite{HWY22c} to obtain an efficient algorithm assuming $\Delta\lesssim 2^{k/5}$.
The symbol $\lesssim$ suppresses lower-order items such as $\poly(k)$.\footnote{We make sure to suppress same order terms when comparing to other works in \Cref{tab:his} and \Cref{tab:hc}.}
The conditions of the latter algorithm still has a gap exponential in~$k$ compared to the randomised algorithm or the hardness threshold.

We close this $\exp(k)$ gap using our new framework.
The result is summarised in \Cref{thm-h-ind},
matching the previously mentioned randomised algorithm and hardness threshold up to a factor of $O(k^2)$.
We note that one caveat of our algorithm, similar to all other deterministic approximate counting algorithms we are aware of,
is that the exponent of its running time depends on $\Delta$ and $k$, 
rather than an absolute constant as in the case of randomised algorithms.
A detailed running time bound of our algorithm is given in \Cref{section:hyper-indset}.
Comparisons with previous works are summarised in \Cref{tab:his}.



\begin{theorem}\label{thm-h-ind}
Let $k \geq 2$ and $\Delta \geq 2$ be two constants satisfying 
$\Delta\le \frac{1}{\sqrt{8\mathrm{e}}k^2}\cdot2^{\frac{k}{2}}$.
There is an FPTAS for the number of independent sets in $k$-uniform hypergraphs with maximum degree $\Delta$. 
\end{theorem}

\begin{table}[htbp]
\centering
\begin{tabular}{|>{\centering\hspace{0pt}}m{0.320\linewidth}|>{\centering\hspace{0pt}}m{0.176\linewidth}|>{\centering\hspace{0pt}}m{0.200\linewidth}|>{\centering\arraybackslash\hspace{0pt}}m{0.200\linewidth}|} 
\hline
\centering Hypergraph independent sets                                                                             & Reference & Bound & Running time  \\ 
\hhline{|====|}
\multirow{2}{0.855\linewidth}{\centering Randomised \\  counting / sampling} &   ~\cite{BDK08,BDK06}    &  $ \Delta\leq k-2$  & $\tilde{O}(n^2)$ / $O(n\log{n})$               \\ 
\cline{2-4}
    & ~\cite{HSZ19,qiu2022perfect}     & $ \boldsymbol{\Delta\lesssim 2^{k/2}}$      &      $\tilde{O}(n^2)$ / $O(n\log{n})$         \\ 
\hhline{|====|}
\multirow{4}{0.855\linewidth}{ \centering Deterministic\\ counting}       &  ~\cite{BGGGS19}   &  $ \Delta\leq k$       &   $n^{O(\log (k\Delta))}$            \\ 
\cline{2-4}
  & ~\cite{JPV21a}   &   $ \Delta\lesssim 2^{k/7}$        &   $n^{\poly(k,\Delta)}$             \\ 
\cline{2-4}
          &  ~\cite{HWY22c}    & $ \Delta\lesssim 2^{k/5}$   &     $n^{\poly(k,\Delta)}$           \\ 
\cline{2-4}
        &\textbf{This work}      & $ \boldsymbol{\Delta\lesssim 2^{k/2}}$       &   $n^{\poly(k,\Delta)}$             \\ 
\hhline{|====|}
Hardness                                                                         &  ~\cite{BGGGS19}     &   $ \Delta\ge 5\cdot2^{k/2}$ assuming $\P\neq\NP$    &         \\
\hline
\end{tabular}
\caption{Algorithms and hardness results for hypergraph independent sets}
\label{tab:his}
\end{table}

This is a relatively straightforward application, since we just apply CTTP to the uniform distribution over hypergraph independent sets.
Our run-time analysis incorporates various techniques developed in the local lemma context.
The crucial part is to show that when CTTP runs for too long,
there must be many independent unlikely events happening.

Comparing to the previous best deterministic algorithm by He, Wang, and Yin~\cite{HWY22c},
our algorithm is both simpler and stronger, thanks to the new CTTP technique.
The previous algorithm is a derandomisation of a recursive marginal sampler by Anand and Jerrum~\cite{AJ22},
which is also a major inspiration for our algorithm.
We will discuss it in more detail in \Cref{sec:intro-AJ}.
Derandomising the Anand-Jerrum algorithm requires sophisticated dynamic programming,
and one needs to carefully control how conditioning goes in the recursive calls,
leading to some complicated analysis and worse conditions.
In contrast, as CTTP always draws fresh random variable, there is no such issue.

The regime of parameters in which our technique applies go further if the hypergraph is \emph{linear}; namely, any two hyperedges intersect on at most one vertex. 
In this case our result also almost matches not only the state-of-the-art randomised algorithms \cite{HSZ19},
where they require $\Delta\le c2^k/k^2$ for some absolute constant $c>0$, 
but also the hardness result that approximate counting is intractable when $\Delta>2.5\cdot 2^k$ \cite{QW22}.
The improvement for linear hypergraphs is achieved by adapting a technique introduced in \cite{FGW22a}.

\begin{theorem}\label{thm-h-ind-simple}
  For any real number $\delta > 0$,
  let $k \geq \frac{25(1+\delta)^2}{\delta^2}$ and $\Delta \geq 2$ be two integers such that $\Delta \le \frac{1}{100k^3} 2^{k/(1+\delta)}$.
  There is an FPTAS for the number of independent sets in $k$-uniform linear hypergraphs with maximum degree $\Delta$.
\end{theorem}    

\subsubsection{Counting hypergraph colourings}

The other application we give is to approximately count the number of hypergraph (proper) colourings.
Again let $H=(V,\+E)$ be a hypergraph, and a $q$-colouring $\sigma\in [q]^V$ is called \emph{proper} if no hyperedge is monochromatic under $\sigma$. 
Similar to hypergraph independent sets,
our work also closes the exponential gap that previously existed between the range of parameters for which efficient randomised and deterministic algorithms exist.

This problem was also studied first in \cite{BDK08},
where they show that Glauber dynamics mixes in $O(n\log n)$ time if $\Delta<q-1$ and $k>4$.
However, to go beyond $\Delta=\Theta(q)$, one starts to encounter the so-called frozen barrier.
As observed in \cite{FM11},
Glauber dynamics can no longer be connected when $\Delta > cq$ for sufficiently large $c$.

The key to bypassing the frozen barrier lies in a classical combinatorial result, the Lov\'asz local lemma.
Introduced by Erd\H{o}s and Lov\'asz \cite{EL75},
the local lemma was originally used to show that hypergraph colourings exist when $\Delta\le q^{k-1}/\mathrm{e}$.
In fact, it gives more information about the distribution of uniform colourings than their mere existence \cite{haeupler2011new}.
Especially if we assume similar but stronger conditions than the local lemma condition,
the distribution enjoys many nice properties.
This enables a projection approach \cite{FGYZ21b,FHY21} where one finds a properly projected distribution to run Glauber dynamics on.
The projected distribution has better connectivity than the original state space, yet transitions between projected states can be efficiently implemented.
With this approach, the state-of-the-art randomised algorithm \cite{JPV21} (and the related perfect sampling algorithm \cite{HSW21})
is shown to be efficient if $\Delta\lesssim q^{k/3}$.

On the other hand, in fact, efficient deterministic algorithms were obtained even before randomised algorithms in this setting \cite{GLLZ19} under similar local lemma type conditions,
based on Moitra's linear programming based approach \cite{Moi19}.
Their algorithm was later improved by Jain, Pham, and Vuong \cite{JPV21a},
and then the aforementioned work of He, Wang, and Yin \cite{HWY22c} introduced an alternative method achieving $\Delta\lesssim q^{k/5}$,
which was state-of-the-art before our work.
Note, once again, the exponential in $k$ gap between the range of parameters for efficient randomised and deterministic algorithms.

We also close this gap. Our result is summarised in \Cref{thm-h-color}.
However, in this setting, there is still an exponential in $k$ gap between the algorithmic threshold and the hardness threshold of \cite{GGW22}.
A detailed running time bound is given in \Cref{section:hyper-colouring}.
Previous works and hardness results are summarised in \Cref{tab:hc}.
Once again, for a detailed discussion and comparison with the previous best algorithm by He, Wang, and Yin \cite{HWY22c},
see \Cref{sec:intro-AJ}.

\begin{theorem}\label{thm-h-color}
  Let $k \geq 20$, $\Delta \geq 2$, and $q$ be three integers such that $\Delta \le \left( \frac{q}{64} \right)^{\frac{k-5}{3}}$. 
  There is an FPTAS for the number of proper $q$-colourings in $k$-uniform hypergraphs with maximum degree $\Delta$.
\end{theorem}

\begin{table}[htbp]
\centering
\begin{tabular}{|>{\centering\hspace{0pt}}m{0.265\linewidth}|>{\centering\hspace{0pt}}m{0.181\linewidth}|>{\centering\hspace{0pt}}m{0.220\linewidth}|>{\centering\arraybackslash\hspace{0pt}}m{0.230\linewidth}|} 
\hline
 Hypergraph colourings                                                                             & Reference & Bound & Running time  \\ 
\hhline{|====|}
\multirow{3}{0.855\linewidth}{\centering Randomised \\  counting / sampling} &   ~\cite{BDK08}    &  $ \Delta\leq q-1$  & $\tilde{O}(n^2)$ / $O(n\log{n})$               \\ 
\cline{2-4}
    & ~\cite{FHY21}     & $ \Delta\lesssim q^{k/9}$      &    $\tilde{O}(n^{2.0001})$ /  $\tilde{O}(n^{1.0001})$         \\ 
    \cline{2-4}
    & ~\cite{JPV21,HSW21}     & $ \boldsymbol{\Delta\lesssim q^{k/3}}$      &    $\tilde{O}(n^{2.0001})$ /   $\tilde{O}(n^{1.0001})$         \\ 
\hhline{|====|}
\multirow{4}{0.855\linewidth}{\centering Deterministic\\  counting}       &  ~\cite{GLLZ19}   &  $\Delta\lesssim q^{k/14}$       &   $n^{\poly(k,\Delta,\log{q})}$            \\ 
\cline{2-4}
  & ~\cite{JPV21a}   &   $ \Delta\lesssim q^{k/7}$        &   $n^{\poly(k,\Delta,\log{q})}$             \\ 
\cline{2-4}
          &  ~\cite{HWY22c}    & $ \Delta\lesssim q^{k/5}$   &     $n^{\poly(k,\Delta,\log{q})}$           \\ 
\cline{2-4}
        &\textbf{This work}      & $ \boldsymbol{\Delta\lesssim q^{k/3}}$       &   $n^{\poly(k,\Delta,\log{q})}$             \\ 
\hhline{|====|}
Hardness                                                                         &  ~\cite{GGW22}     &  for even~$q$, $\Delta\ge 5 q^{k/2}$\\  assuming $\P\neq\NP$       &               \\
\hline
\end{tabular}
\caption{Algorithms and hardness results for hypergraph colourings}
\label{tab:hc}
\end{table}

In this setting, we apply CTTP to the projected distribution, and there are two main differences from hypergraph independent sets.
First, in each step of the self-reduction, instead of a single variable, we need to evaluate a set of variables of size $k$.
A natural attempt to address this is to run CTTP for each variable,
and there is a further complication that we need to maintain consistency of the randomness used amongst different runs.
Moreover, as we are sampling from the projected distribution, we still need to sample a colouring conditioned on the projected sample.
This is addressed by noticing that CTTP not only returns the final state of the variables,
it also gives us enough information to perform the last update for them.
This information is indeed also enough to sample a colouring from the marginal distribution conditioned on the projected sample.

Similar to the case of hypergraph independent sets,
we also have improved results for linear hypergraphs, almost matching the state-of-the-art randomised algorithm \cite{FGW22a}.

\begin{theorem}\label{thm-h-color-simple}
  For any real number $\delta > 0$,
  let  $k \geq  \frac{50(1+\delta)^2}{\delta^2}$, $\Delta \geq 2$, and $q$ be three integers such that $\Delta \le \left( \frac{q}{50} \right)^{\frac{k-3}{2+\delta}}$. 
  There is an FPTAS for the number of proper $q$-colourings in $k$-uniform linear hypergraphs with maximum degree $\Delta$.
\end{theorem}

\subsection{Connection and comparisons with the Anand-Jerrum algorithm}
\label{sec:intro-AJ}

As mentioned before, the core of our derandomisation method is a logarithmic-cost marginal sampler, which may have independent interests.
Our main source of inspiration, and also the first such marginal sampler, is the recent recursive algorithm by Anand and Jerrum~\cite{AJ22} for perfect sampling in infinite spin systems.
Although the implementation of our marginal sampler also has a recursive structure,
there are some quite noticeable distinctions.
For one thing, our algorithm can be derived from Glauber dynamics,
whereas the AJ algorithm relies on spatial mixing properties and does not seem to directly correspond to any Markov chain.
Moreover, as the recursive call goes deeper and deeper,
AJ's algorithm would pin more and more variables.
In contrast, our sampler resembles Markov chains.
No variable is permanently fixed, and all variables can be refreshed if we go back long enough.
This last point actually provides us with a technical edge in the analysis, which we will discuss later.


The Anand-Jerrum algorithm soon found applications for almost uniform sampling general constraint satisfaction solutions in the local lemma regime~\cite{HWY22a}.
It was later derandomised~\cite{HWY22c}, 
which leads to the previous best deterministic approximate counting algorithms for the two applications we consider.
This is a very recent and rare exception to the common paradigm that deterministic approximate counting algorithms require drastically different techniques from randomised algorithms. 
However, the Anand-Jerrum algorithm encounters some difficulty in matching the state-of-art bounds for randomised algorithms,
particularly in the two applications we consider -- hypergraph independent sets and hypergraph colourings.
From a technical point of view, the difficulty lies in the lack of control for the aforementioned pinning of partial configurations in deep recursive calls of AJ's algorithm.
This makes it hard to analyse the time-space structure derived from the algorithm's execution history,
which is a key feature in ours and previous information percolation arguments~\cite{HSZ19,qiu2022perfect,JPV21,HSW21} to achieve the state-of-the-art bounds.
%
%
We leave it as an interesting direction whether there is a refined analysis of the Anand-Jerrum algorithm matching other methods in these contexts. 

In any case, we remark that the Anand-Jerrum algorithm does also lead to deterministic counting algorithms,
especially for spin systems with strong spatial mixing on graphs with subexponential growth.
The only thing missing from \cite{AJ22} is tail bounds for their algorithm's running time.
We provide such analysis and collate the implications in \Cref{sec:AJ}.

\subsection{Derandomising random scans}
\label{sec:intro-random}

We choose to consider systematic scan to illustrate our idea of using CTTP for derandomisation, 
but the technique can be applied to random update Markov chains as well.
In \Cref{sec:derandomise-random-scan}, we show how to use CTTP to derandomise random scan Glauber dynamics for Gibbs distributions.

A vanilla attempt is to do extra enumerations over possible scan sequences. 
However, even with a more careful argument like considering only the visited vertices, it still costs superpolynomial running time, assuming $O(\log n)$ so many times of backward deductions (which is also optimal for the truncation error we need). 
Instead, we take the advantage of the local structure of visiting patterns by introducing a ``witness tree" structure that captures all information needed to recover the CTTP process.
Enumerating all possible witness trees leads to efficient derandomisation for a random scan order.

We remark that our witness tree for Markov chains is reminiscent of the witness tree appeared in the analysis of the celebrated Moser-Tardos algorithm for algorithmic Lov\'{a}sz Local Lemma~\cite{MT10},
which may be of independent interest.


\subsection{Other related work}\label{sec:related-work}


Our coupling towards the past (CTTP) marginal sampler is also inspired by the celebrated ``coupling \emph{from} the past'' (CFTP) by Propp and Wilson~\cite{PW96} for perfect simulation of Markov chains.
Both methods share some similarities, such as running backwards in time and having underlying grand couplings. 
However, they are also very different in several aspects.
The main difference is that CFTP needs to sequentially simulate the evolution of the whole state,
which, even with some optimisation (such as using bounding chains~\cite{huber1998exact}) and under favourable conditions, would still require at least linear time.
This makes it unsuitable for our derandomisation needs.
Our CTTP, on the other hand, only guarantees the value of a single variable to be coupled from all possible starting configurations.
Even if we want to couple only a single variable in CFTP, it would be impossible to determine what variables to simulate a priori.
Our backwards deduction approach is an adaptive solution to this problem,
namely, the information revealed so far determines what variables to be revealed next.
This constitutes a big difference in implementing the two methods.

Our CTTP marginal sampler is also related to providing \emph{local access to huge random objects}~\cite{BiswasRY20}. 
For example, consider the uniform distribution $\mu$ of independent sets in a huge hypergraph $H=(V,\+E)$.
Upon being queried at a vertex $v \in V$, our algorithm returns an approximate sample $X_v$ from the marginal distribution $\mu_v$ using a sub-linear number of local neighbourhood probes.
By using public random bits, our algorithm could guarantee that the answers for different queries are consistent with each other, which means for different vertices $v$, all the answers $X_v$ come from the same $X$ such that $X$ is an approximate random sample of $\mu$.
An interesting open problem in this direction is to give (consistent or not) sub-linear time marginal samplers for general Gibbs distributions whenever an efficient global sampling algorithm exists.


There have been various works aiming to find deterministic approximations for Markov chains, especially for random walks on graphs \cite{CS06,CDFS10,SYKY17,SYKY18,MRSV21,PV22}.
However, these results do not seem to have meaningful consequences for MCMC,
where the Markov chains are essentially high-dimensional random walks.
The state space and the underlying (implicit) graph are exponentially large, making those aforementioned results difficult to apply.

Our work focuses on approximate counting via derandomising certain Markov Chain Monte Carlo samplers, 
which falls under the broader category of derandomising Monte Carlo methods. 
See~\cite{luby91deter,Luby1993Deterministic,gopalan2012DNF,de2013efficient,rocco19pseudo} for several examples of the latter.
Usually, in their contexts, it is sufficient to consider additive errors, and as a consequence, the randomised algorithms are relatively simple to get.
The emphasis is on finding techniques to derandomise them.
In contrast, our derandomisation actually comes from designing and analysing more involved randomised algorithms.

\

\section{Preliminaries}\label{sec-pre}

\subsection{Markov chain and Glauber dynamics}
Let $\Omega$ be a (finite) state space.
Let $(X_t)_{t = 1}^\infty$ be a Markov chain over the space $\Omega$ with transition matrix $P$.
We often use $P$ to refer to the corresponding Markov chain.
A distribution $\pi$ over $\Omega$ is a \emph{stationary distribution} of $P$ if $\pi = \pi P$.
The Markov chain $P$ is \emph{irreducible} if for any $x,y \in \Omega$, there exists $t$ such that $P^t(x,y) > 0$.
The Markov chain $P$ is \emph{aperiodic} if for any $x \in \Omega$, $\gcd\{t\mid P^t(x,x) > 0\} = 1$.
If the Markov chain $P$ is both irreducible and aperiodic, then it has a unique stationary distribution.
The Markov chain $P$ is \emph{reversible} with respect to distribution $\pi$ if the following \emph{detailed balance equation} holds
\begin{align*}
	\forall x, y \in \Omega,\quad \pi(x) P(x,y) = \pi(y)P(y,x),
\end{align*}
which implies $\pi$ is a stationary distribution of $P$.
The \emph{mixing time} of the Markov chain $P$ is defined by
\begin{align*}
	\forall \epsilon > 0, \quad T(P,\epsilon) \defeq \max_{X_0 \in \Omega} \max\{t \mid \DTV{P^t(X_0,\cdot)}{\mu} \leq \epsilon\},
\end{align*}
where the \emph{total variation distance} is defined by
\begin{align*}
\DTV{P^t(X_0,\cdot)}{\mu} \defeq \frac{1}{2}\sum_{y \in \Omega}\abs{P^t(X_0,y)-\mu(y)}.	
\end{align*}


In this paper, we consider two fundamental Markov chains on discrete state space.
Let $\mu$ be a distribution over $[q]^V$. We assume $V = \{v_1,v_2,\ldots,v_n\}$.
The \emph{Glauber dynamics} starts from an arbitrary $X_0 \in [q]^V$ with $\mu(X_0)>0$.
For the $t$-th transition step, the Glauber dynamics does as follows 
\begin{itemize}
	\item pick a variable $v \in V$ uniformly at random and let $X_t(u) = X_{t-1}(u)$ for all $u \neq v$;
	\item sample $X_t(v)$ from the distribution $\mu_{v}^{X_{t-1}(V \setminus \{v\})}$.
\end{itemize}
The \emph{systematic scan Glauber dynamics} starts from an arbitrary $X_0 \in [q]^V$ with $\mu(X_0)>0$.
For the $t$-th transition step, the systematic scan Glauber dynamics does as follows 
\begin{itemize}
	\item let $i(t) = (t \mod n) + 1$, pick the variable $v = v_{i(t)}$, and let $X_t(u) = X_{t-1}(u)$ for all $u \neq v$;
	\item sample $X_t(v)$ from the distribution $\mu_{v}^{X_{t-1}(V \setminus \{v\} )}$.
\end{itemize}
The only difference between the above two Markov chains is the way they pick variables.
The Glauber dynamics is an aperiodic and reversible Markov chain.
The systematic scan Glauber dynamics is not a time-homogeneous Markov chain. However by bundling $n$ consecutive updates together we can obtain a time-homogeneous Markov chain, which is aperiodic and reversible. 
\begin{theorem}[\cite{levin2017markov}]\label{thm-convergence}
Let $\mu$ be a distribution with support $\Omega \subseteq [q]^V$. Let $(X_t)_{t =0}^\infty$ denote the Glauber dynamics or the systematic scan Glauber dynamics on $\mu$. If $(X_t)_{t =0}^\infty$ is irreducible over $\Omega$, it holds that
\begin{align*}
	\forall X_0 \in \Omega,\quad \lim_{t \to \infty}\DTV{X_t}{\mu} = 0.
\end{align*}
\end{theorem}

\subsection{Lov\'{a}sz local lemma}
We introduce the setting of the (variable framework) Lov\'{a}sz Local Lemma. Let $\+{X}=\{X_1,X_2,\dots,X_n\}$ be a set of mutually independent random variables. We let $\+B=\{B_1,B_2,\dots,B_m\}$ be a set of ``bad events'' that only depends on $\+X$. For each event $A$ (not necessarily one of the bad events in $\+B$), we let $\var{A}\subseteq \+{X}$ be the set of variables in $\+{X}$ that $A$ depends on. Moreover, we let $\Gamma(A)=\{B\in \+{B}\mid B\neq A\land \var{A}\cap \var{B}\neq \emptyset\}$. The celebrated Lov\'{a}sz Local Lemma states that when certain conditions are met, the probability that no bad events occur is nonzero:

\begin{lemma}[\cite{EL75}]\label{locallemma}
    If there exists a function $x:\+{B}\rightarrow [0,1]$ such that
    \begin{align}\label{llleq}
   \forall B\in \+{B}:\quad
        {\Pr{B}\leq x(B)\prod_{B'\in \Gamma(B)}(1-x(B'))},
    \end{align}
    then  
    $$
        {\Pr{ \bigwedge\limits_{B\in \+{B}} \overline{B}}\geq \prod\limits_{B\in \+{B}}(1-x(B))>0},
    $$
\end{lemma}

When the condition \eqref{llleq} is satisfied, as observed in \cite{haeupler2011new},
the probability that any event happens, conditioning on no bad events occurs is also bounded:

\begin{lemma}[\cite{haeupler2011new}]\label{lemma:HSS}
If $\eqref{llleq}$ holds, 
then for any event $A$, 
\[
    \Pr{A\mid \bigwedge\limits_{B\in \+{B}} \overline{B}}\leq \Pr{A}\prod_{B\in \Gamma(A)}(1-x(B))^{-1},
\]
\end{lemma}

Note that for a $k$-uniform hypergraph $H=(V,\+{E})$ with $\Delta\geq 2$ and $q^{k-1}>\mathrm{e}\Delta$, by setting the bad events $B_e$ being ``$e$ is not monochromatic" for each $e\in \+{E}$, we then have $\Pr{B_e} = q^{1-k}$ for each $e\in \+{E}$. Setting $x(B_e)=\frac{1}{\Delta}$ for each $e\in \+E$ we have the condition in \eqref{llleq} is satisfied, and therefore \Cref{lemma:h-color-edge-uniform} is a direct corollary from \Cref{lemma:HSS}.

\subsection{Counting to sampling reductions} \label{sec-jvv}
The classical result of Jerrum, Valiant and Vazirani \cite{JVV86} showed that sampling and (randomised) approximate counting can be reduced to each other in polynomial time for ``self-reducible'' functions.
We do not need that level of generality, and will describe the next reductions from counting to sampling for the following \emph{graphic model}. 
Let $H = (V,\+E)$ be a hypergraph, where each vertex $v \in V$ represents a random variable that takes its value from a finite domain $[q]=\{1,2,\ldots,q\}$ and each hyperedge $e \in \+E$ represents a local constraint on the variable set $e \subseteq V$.
For each $v \in V$, there is an ``external field'' function $\phi_v: [q]\to \mathbb{R}_{\geq 0}$, 
and for each $e \in \+E$, there is a ``constraint'' function $\phi_e: [q]^{|e|} \to \mathbb{R}_{\geq 0}$.
A graphical model is specified by the tuple $\+G = (H, (\phi_v)_{v \in V}, (\phi_e)_{e \in \+E})$, namely the hypergraph associate with the family of ``external field'' and ``constraint'' functions.
For each configuration $\sigma \in [q]^V$, define its weight by
\begin{align}\label{eq-def-weight}
	w(\sigma) := \prod_{v \in V}\phi_v(\sigma_v)\prod_{e \in \+E}\phi_e(\sigma_e).
\end{align}
The Gibbs distribution $\mu$ defined by the graphical model satisfies 
\begin{align}\label{eq-def-Gibbs}
	\forall \sigma \in [q]^V, \quad \mu(\sigma):=\frac{w(\sigma)}{Z},
\end{align}
where the partition function $Z$ is given by 
\begin{align}\label{eq-def-partition}
	Z:=\sum_{\sigma \in [q]^V}w(\sigma).
\end{align}
In our applications, both counting problems for hypergraph independent sets and the hypergraph colourings can be expressed as graphic models, 
by setting the external field $\phi_v$ on each vertex to constant $1$, and the constraint function $\phi_e$ on each hyperedge to the indicator function that the respective constraint is not violated. 

At the core of \cite{JVV86} is the decomposition of the partition function into products of more tractable quantities in a telescoping manner.
Typically, these quantities are probabilities or expectations related to a sequence of distributions induced by the original model. 
In our paper, such decomposition comes in two manners: the vertex decomposition and the (hyper)edge decomposition. 
The vertex decomposition is simpler and more common, 
but it does not always work.
The edge decomposition works more generally with the sacrifice of simplicity.
We will apply the vertex decomposition to hypergraph independent sets and the edge decomposition to hypergraph colourings.

\paragraph{\textbf{Vertex decomposition.}}
We assume that $V=\{v_1,\cdots,v_n\}$. 
Let $\sigma\in[q]^V$ be a feasible configuration that $w(\sigma)>0$.  
We call $\sigma$ the \emph{scheme} of vertex decomposition. 
By the chain rule of the Gibbs distribution, the following product holds
\begin{align}\label{eq-vtx-decom}
	Z = \frac{w(\sigma)}{\mu(\sigma)} = w(\sigma) \prod_{i=1}^n \frac{1}{\mu_{v_i}^{{\sigma_{< i}}}\tp{\sigma_{v_i}}},  
\end{align}
where we use $\sigma_{< i}$ to denote the partial configuration of $\sigma$ restricted on vertex set $\{v_1,v_2,\ldots,v_{i-1}\}$, and $\mu_{v_i}^{\sigma_{<i}}$ denote the marginal distribution on $v_i$ induced from $\mu$ conditional on $\sigma_{<i}$.
Therefore, to approximate the partition function $Z$, it suffices to approximate each marginal probability $\mu_{v_i}^{{\sigma_{< i}}}\tp{\sigma_{v_i}}$, since the weight function $w(\sigma)$ is easy to compute. 
The scheme $\sigma$ is said to be \emph{$b$-bounded} if 
\begin{align*}
\forall i \in [n], \quad \mu_{v_i}^{{\sigma_{< i}}}\tp{\sigma_{v_i}} \geq b.
\end{align*}

\paragraph{\textbf{Edge decomposition.}}
Let $H=(V,\mathcal{E})$ be a hypergraph. 
We assume $V=\{v_1,\cdots,v_n\}$ and $\mathcal{E}=\{e_1,\cdots,e_m\}$. 
In the edge decomposition, the original hypergraph $H$ is decomposed into a sequence of hypergraphs $H_0,H_1,\cdots,H_m$, given by $H_i:=\{V,\{e_1,\cdots,e_i\}\}$. 
In other words, the $H_i$ sequence is obtained by, starting from independent vertices, adding one edge from the original hypergraph in each step. 
Let $w_i$, $Z_i$ and $\mu_i$ be the weight function, the partition function and the Gibbs distribution induced by the hypergraph $H_i$ respectively. Then
\begin{equation}\label{eq-edge-decom}
\begin{gathered}
	Z = Z_m = Z_0 \prod_{i=1}^m \frac{Z_i}{Z_{i-1}} = Z_0 \prod_{i=1}^m \sum_{\tau \in [q]^V}\frac{\mu_{i-1}(\tau) w_i(\tau)}{w_{i-1}(\tau)}\\
	 \overset{(\star)}{=} Z_0 \prod_{i=1}^m \sum_{\tau_{e_i} \in [q]^{e_i}}\mu_{i-1,e_i}(\tau_{e_i}) \phi_{e_i}(\tau_{e_i})= Z_0 \prod_{i=1}^m \Ex_{z \sim \mu_{i-1,e_i}}[\phi_{e_i}(z)],
\end{gathered}
\end{equation}
where $\mu_{i-1,e_i}(\cdot)$ denotes the marginal distribution on $e_i$ projected from $\mu_{i-1}$, and thus $(\star)$ holds because $\frac{w_i(\tau)}{w_{i-1}(\tau)} = \frac{w_i(\tau_{e_i})}{w_{i-1}(\tau_{e_{i-1}})}$. 
Therefore, it suffices to approximate the expectation $\Ex_{z \sim \mu_{i-1,e_i}}[\phi_{e_i}(z)]$ in order to approximate the partition function. 
Again, the edge decomposition is called \emph{$b$-bounded} if
\begin{align*}
	\forall 1 \leq i \leq m, \quad \frac{ \Ex_{z \sim \mu_{i-1,e_i}}[\phi_{e_i}(z)] }{\max_{ y \in [q]^{e_i}} \phi_{e_i}(y) } \geq b.
\end{align*}

\section{Derandomisation for deterministic counting} \label{sec-reduction}
Our idea for deterministic counting is very simple --- we just enumerate all possible random choices.
In this section, we give a quick formalisation of this idea,
and in subsequent sections, we tackle the main challenge of finding algorithms with logarithmic random choices.
We consider randomised algorithms whose whole randomness comes from drawing random variables from discrete distributions,
such as lines of the form:
``draw $r\sim D$'', 
where $D$ is a probability distribution over a finite sample space $\Omega$ of constant size.
The specification of $D$ is also computed by the algorithm if necessary. 
Aside from these samples, there is no randomness involved in the algorithm.
This motivates us to consider the following random oracle model.
%

\vspace{6pt}
\paragraph{\textbf{Random oracle model}}
We consider randomised algorithms that are deterministic algorithms with access to a random oracle $\mathsf{Draw}(\cdot)$, 
which, given as input the description of a distribution $D$, 
returns an independent random value $r\in \Omega$ distributed according to $D$.

This model allows us to quantify the number of random choices made in algorithms, as given in the next definition.

\begin{definition}
\label{def-sampling-alg}
Let $t,r:\mathbb{N}\to\mathbb{N}$ be two nondecreasing functions
and let $c\ge 2$ be a constant.
We say that a randomised algorithm $\+{A}$ has \emph{time cost $t(n)$} and  \emph{draws at most $r(n)$ random variables over domains of sizes at most $c$}, 
if for any $n\in\mathbb{N}$, in the worst case of the inputs of size $n$ and all possible random choices,
the algorithm $\+{A}$ terminates within $t(n)$ steps of computation, 
and accesses the random oracle $\mathsf{Draw}(\cdot)$ for at most $r(n)$ times such that each time it draws from a sample space of size at most~$c$.
\end{definition}

We are interested in those randomised algorithms that have $\-{poly}(n)$ time cost  and draws at most $O(\log n)$ random variables over constant-sized domains,
because such randomised algorithms can be transformed into polynomial-time deterministic algorithms for computing the output distributions,
due to a standard routine for derandomisation by enumerating all random choices. 
%

\begin{proposition}\label{proposition-deran}
Let $\+A$ be a randomised algorithm with time cost $t(n)$ and draws at most $r(n)$ random variables over domains of sizes at most $c$.
There is a deterministic algorithm $\+B$ that, on any input $\Pi$ of size $n$, outputs the distribution of $\+A(\Pi)$ in time $O(t(n)c^{r(n)})$. 
\end{proposition}
\begin{proof}
Consider the decision tree $\+{T}=\+{T}(\Pi)$ for adaptively querying the random oracle $\mathsf{Draw}(\cdot)$ by the algorithm $\+{A}$ on an input $\Pi$ of size $n$.
Since the algorithm  $\+A$ draws at most $r(n)$ random variables over domains of sizes at most $c$ in the worst case of inputs and random choices,
the decision tree has a branching number at most $c$ and depth at most $r(n)$. Therefore, there are at most $c^{r(n)}$ leaves in $\+{T}$.
Since $\+A$ has a time cost $t(n)$  in the worst case of inputs and random choices, the computation cost for the path from the root to each leaf in $\+{T}$ is bounded by $t(n)$.
Hence, the entire tree $\+{T}$ can be computed in $O(t(n)c^{r(n)})$ time.
Note that each leaf in $\+{T}$ corresponds to a possible output value for $\+A(\Pi)$, whose probability is given by that of the random choices along the path.
Therefore, the distribution of the output $\+A(\Pi)$ can be computed within $O(t(n)c^{r(n)})$ time by aggregating over all leaves in $\+{T}$.
\end{proof}

\paragraph{\textbf{Implications to approximate counting}}
Recall the self-reductions using vertex/edge decompositions defined in \Cref{sec-jvv} and the corresponding marginal distributions $\mu_{v_i}^{\sigma_{<i}}$ in \eqref{eq-vtx-decom} and $\mu_{i-1,e_i}$ in \eqref{eq-edge-decom}.
Then a straightforward consequence to \Cref{proposition-deran} is that
the partition functions can be approximated deterministically in polynomial-time 
as long as one can sample approximately from the marginal distribution in polynomial-time drawing $O(\log n)$ random variables whose domain sizes are upper bounded by a constant.
This is formally stated below.

\begin{corollary}\label{cor-reduction}
Let $\+{G}$ be a class of graphical models, where each instance $\+{I}\in\+{G}$ is provided with a $b$-bounded vertex decomposition (or a $b$-bounded edge decomposition).
If for every $\epsilon\in(0,1)$
there exists a randomised algorithm $\+{A}$ such that
for every instance $\+{I}\in\+{G}$ of $n$ vertices and $m$ edges, 
and every possible marginal distribution $\mu_{v_i}^{\sigma_{<i}}$ used in the vertex decomposition (or $\mu_{i-1,e_i}$ in the edge decomposition), 
the algorithm $\+{A}$ returns a $Y_i$ within time $t(\epsilon,n)$, by drawing at most $r(\epsilon,n)$ random variables of domain sizes at most $c$,
such that
\[
\DTV{Y_i}{\mu_{v_i}^{\sigma_{<i}}}\leq\frac{b\varepsilon}{10n}, \qquad \left(\text{or }\,\DTV{Y_{i}}{\mu_{i-1,e_{i}}}\leq\frac{b\varepsilon}{10m}\,\,\text{ for the edge decomposition,}\right)
\]
then there exists a deterministic algorithm $\+{B}$ that  for every $\epsilon\in(0,1)$ and every instance $\+{I}\in\+{G}$,
returns an $\epsilon$-approximation of the partition function of $\+{I}$ within time $O((m+n)t(\epsilon,n)c^{r(\epsilon,n)})$.
\end{corollary}

The proof of \Cref{cor-reduction} is rather straightforward.
The only slight complication is to convert additive errors into relative errors,
which is made possible by our $b$ bounded assumptions.

\section{Coupling towards the past} \label{sec:marginal-sampler}

In this section, we introduce the coupling towards the past idea and present a marginal sampler by evaluating the state of a single variable in stationary Markov chains.
Suppose the chain has evolved sufficiently long. 
Our goal is to evaluate the current state by revealing as little randomness as possible.
It suffices to perform the last update,
and the update function is determined by both a fresh random variable and the states of other variables.
The fresh random variable may allow us to determine the state of the target variable directly without knowing any state of the rest.
This is possible when marginal probability lower bounds are available.
If we cannot make such a quick decision, we recursively reveal the necessary information on other variables required for this update.
In the case of Gibbs distributions, this step amounts to finding out the states of the neighbours of the target variable.
However, in the application of hypergraph colourings, to overcome the irreducibility barrier, we need to sample from a ``projected distribution" instead of the uniform distribution over all proper colourings. This ``projected distribution" we sample from is no longer a Gibbs distribution,
and the revealing step becomes more complicated.
Finally, we will truncate once we have revealed too much information to ensure that this process is efficient on randomness.

There is an implicit grand coupling (i.e.~a coupling for chains starting from all possible initial configurations) underlying the construction above.
We may consider all random variables drawn beforehand, and then all chains use the same values.
The approach above is just delayed revelation, or alternatively constructing the grand coupling recursively towards the past.
If the CTTP process terminates at time $-t$,
then it means that under these random choices,
all chains starting from a time $T<-t$, no matter what the initial configurations are, lead to the same value at time $0$ for the target variable.

To carry out the plan, we first introduce an implementation of the standard systematic scan Glauber dynamics which utilises the lower bound information as much as possible.
This is in \Cref{section-scan-GD}.
Then in \Cref{section-backward-deduction}, we flip the order of evaluation and deduce the state of a single variable backwards.
The algorithm is given in \Cref{Alg:resolve} and its correctness is shown in \Cref{theorem:resolve-cor}.
Finally, in \Cref{sec:approx-resolve}, we give the truncated algorithm, \Cref{Alg:appresolve}, and bound its error in \Cref{thm:dtv-truncate}.

\subsection{Simulating Glauber dynamics assuming marginal lower bound}\label{section-scan-GD}
Let $(X_t)$ be a convergent Markov chain on the state space $[q]^V$ with its stationary distribution $\mu$.
In particular, consider the chain $(X_t)_{-\infty<t\le 0}$ running from time $-\infty$ to time 0.
Drawing a sample from the marginal distribution $\mu_v$ for an arbitrary $v\in V$ can be realised by evaluating $X_0(v)$.

For technical reasons, we consider a class of chains known as \emph{systematic scan Glauber dynamics}. 
Enumerate the variables as $V=\{v_1,v_2,\ldots,v_n\}$, and for any $t\in\mathbb{Z}$, define:
\begin{align}
i(t)\defeq (t\bmod n)+1. \label{eq:i-scan}
\end{align}
%
%
The rule for the $t$-th transition ($X_{t-1}\to X_{t}$) is:
\begin{itemize}
\item
pick the variable $v=v_{i(t)}$ where $i(t)$ is defined in~\eqref{eq:i-scan};
\item
let $X_t\in[q]^V$ be constructed as that $X_t(u)=X_{t-1}(u)$ for all $u\neq v$, and $X_t(v)$ is drawn independently according to the marginal distribution $\mu_v^{X_{t-1}(V\setminus\{v\})}$.
\end{itemize}
This chain converges to the (unique) stationary distribution $\mu$ when it is irreducible.
\begin{remark}[choice of systematic scan Glauber dynamics]
    The CTTP is not restricted to systematic scan Glauber dynamics. Assuming a fixed scan order as in the systematic scan simplifies our expositions. But such simplification does not change the nature of the derandomisation problem.
    Later in \Cref{sec:derandomise-random-scan}, we show that with some extra efforts, the CTTP can be applied to derandomise the standard Glauber dynamics with random scan order.
\end{remark}

To perform the update, it is usually not necessary to know all of $X_{t-1}(V\setminus\{v\})$.
Typically there is a subset $\Lambda\subseteq V\setminus\{v\}$ such that $\mu_v^{X_{t-1}(\Lambda)}=\mu_v^{X_{t-1}(V\setminus\{v\})}$.
Let $\sigma_\Lambda = X_{t-1}(\Lambda)$.
The marginal distribution $\mu_v^{\sigma_\Lambda}$ can be decomposed as follows if it is suitably lower bounded.

\begin{definition}\label{definition:margin-lower-bound}
Let $\mu$ be a distribution over $[q]^V$. Let $\boldsymbol{b}=(b_1,b_2,\ldots,b_q)\in[0,1]^q$.
\begin{itemize}
\item
\textbf{Marginal lower bound}: 
$\mu$ is said to be \emph{$\boldsymbol{b}$-marginally lower bounded} if for any $v\in V$, $\Lambda\subseteq V\setminus\{v\}$ and any feasible ${\sigma_\Lambda}\in[q]^\Lambda$, it holds that $\mu_v^{\sigma_\Lambda}(j)\ge b_j$ for all $j\in [q]$.
\end{itemize}
For a $\boldsymbol{b}$-marginally lower bounded distribution $\mu$ over $[q]^V$, for each $v\in V$, 
we define the following distributions (we follow the convention $0/0=0$):
\begin{itemize}
\item
\textbf{Lower bound distribution $\mu_v^{\-{LB}}$ over $\{\perp\}\cup[q]$}: $\mu_v^{\-{LB}}=\mu^{\-{LB}}$ for all $v\in V$ such that
\begin{align*}
\mu^{\-{LB}}(\perp)
\defeq
1-\sum_{i=1}^qb_i
\quad\text{ and }\quad
\forall j\in [q],\quad
\mu^{\-{LB}}(j) 
\defeq
\frac{b_j}{\sum_{i=1}^qb_i}.
\end{align*}
\item
\textbf{Padding distribution $\mu_v^{\-{pad},\sigma_\Lambda}$ over $[q]$}: 
for $\Lambda\subseteq V\setminus\{v\}$ and feasible $\sigma_\Lambda\in[q]^{\Lambda}$,
\[
	\forall j \in [q], \quad \mu_v^{\-{pad},\sigma_\Lambda}(j) \defeq \frac{\mu^{\sigma_\Lambda}_v(j) - b_j}{1 - \sum_{i=1}^q b_i}.
\]
\end{itemize}
\end{definition}

With above definitions,
drawing a sample $c\sim\mu^{\sigma_\Lambda}_v$ according to the marginal distribution $\mu^{\sigma_\Lambda}_v$ can be simulated by the following two steps:
\begin{enumerate}
\item
draw $c\sim \mu_v^{\-{LB}}$;
\item
if $c=\perp$, override $c$ by drawing $c \sim \mu_v^{\-{pad},\sigma_\Lambda}$.
\end{enumerate}

\bigskip
Assume that $\mu$ has suitable marginal lower bounds.
Fix an integer $T \geq 0$.
The systematic scan Glauber dynamics $(X_t)_{-T\le t\le 0}$ from time $-T$ to $0$ can be generated by the following process $\+P(T)$.
We use the convention that for any $v\in V$ and $t<-T$, $X_{t}(v) = X_{-T}(v)$.

\begin{center}
  \begin{tcolorbox}[=sharpish corners, colback=white, width=1\linewidth]
  	\begin{center}
	\textbf{\emph{The systematic scan Glauber dynamics $\+{P}(T)$}}
  	\end{center}
  	\vspace{6pt}
   \begin{itemize}
	\item Initialize $X_{-T} \in [q]^V$ as an arbitrary feasible configuration.
	\item For $t=-T+1,-T+2,\ldots,0$,  the configuration $X_{t}$ is constructed as follows:
	\begin{enumerate}
		\item[(a)] pick $v = v_{i(t)}$, where $i(t)$ is defined in~\eqref{eq:i-scan}, and let $X_{t}(u) \gets X_{t-1}(u)$ for all $u \neq v$; 
		\item[(b)] draw $r_t \sim \mu^{\mathrm{LB}}$ independently, and  let $X_t(v) \gets r_t$ if $r_t \neq \perp$; otherwise, 
\begin{align}\label{eq-findconfig}
\sigma_\Lambda \gets \config(t), \quad\text{(by accessing $X_{t-1}$ and $\+{R}_{t-1}\defeq (r_s)_{-T < s < t}$)} 
\end{align} 
		and draw $X_t(v) \sim \mu^{\mathrm{pad},\sigma_\Lambda}_{v}$ independently.
	\end{enumerate}
\end{itemize}
  \end{tcolorbox} 
\end{center}

In \eqref{eq-findconfig}, we use a subroutine $\config(t)$ satisfying the following condition,
which allows us to perform the update without revealing the whole $X_{t-1}$.

\begin{condition}\label{cond:invariant-boundary}
  The procedure $\config(t)$ always terminates and returns $\sigma_\Lambda\in[q]^\Lambda$ satisfying that $\mu_v^{\sigma_\Lambda}=\mu_v^{X_{t-1}(V\setminus\{v\})}$ for $v=v_{i(t)}$.
\end{condition}

The implementation of $\config(t)$ will be application specific.
It is a deterministic procedure with oracle access to previous random variables $\+{R}_{t-1}= (r_s)_{-T < s < t}$ and to the current configuration $X_{t-1}$ generated in the process $\+{P}(T)$. 
These accesses are provided by the following two oracles:
\begin{itemize}
    \item \textbf{lower bound oracle $\+B(s)$}: given any $s < t$, returns $r_s$ if $s>-T$ and $\perp$ otherwise;
    \item \textbf{configuration oracle $\+C(u)$}: given $u \in V$, returns $X_{t-1}(u)$.
\end{itemize}
Typically, we will query $\+B(s)$ as much as possible,
and default to $\+C(u)$ only if $\+B(s)$ returns $\perp$.
It is even possible, in some situations, not needing to query $\+C(u)$ after $\+B(u)$ returns $\perp$.
This is because not all variables queried in $\config(t)$ are necessary to determine $\Lambda$ and $\sigma_\Lambda$,
and this phenomenon will become self-evident in the applications later.
Although using only $\+C(u)$ is sufficient to achieve \Cref{cond:invariant-boundary},
the use of $\+B(s)$ is crucial and allows us to reduce the number of random variables used for the backward deduction in the next subsection.

To better understand $\config(t)$, take a Gibbs distribution $\mu$ as an example.
In Gibbs distributions,
each variable $v\in V$ has a neighbourhood $N(v)\subseteq V\setminus\{v\}$ conditioned on which $v$ is independent from the rest of the variables, namely its non-neighbours.
%
Consequently, for any feasible $X_{t-1}\in[q]^{V}$,
\[
\mu_v^{X_{t-1}(V\setminus\{v\})}=\mu_v^{X_{t-1}{(N(v))}}.
\]
%
Then a straightforward implementation of $\config(t)$ is to return the configuration $\sigma_{\Lambda}\gets X_{t-1}(N(v))$ by retrieving it from $X_{t-1}$ via the oracle $\+C(\cdot)$.
However, as we want to reduce the number of accesses to $\+C(\cdot)$, 
we instead try to infer $X_{t-1}(N(v))$ from already drawn samples $\+{R}_{t-1}=(r_s)_{-T < s < t}$. 

The following gives an alternative implementation of $\config(t)$ for the Gibbs distribution $\mu$.
For any $u\in V$ and integer $t$, denote by $\upd_u(t)$ the last time before $t$ at which $u$ is updated, i.e.
\begin{align}\label{eq:update-time}
\upd_u(t)\defeq \max\{s\leq t \mid v_{i(s)}=u\},
\end{align}
where $i(s)$ is specified by the scan order  defined in~\eqref{eq:i-scan}.
\begin{center}
  \begin{tcolorbox}[=sharpish corners, colback=white, width=1\linewidth]
  	\begin{center}
	\textbf{\emph{An implementation of $\config(t)$ for Gibbs distribution $\mu$}}
  	\end{center}
  	\vspace{6pt}
   \begin{itemize}
	\item Let $\Lambda\gets N(v)$, where $v = v_{i(t)}$ and $i(t)$ is defined as~\eqref{eq:i-scan}. 
	\item For each $u\in \Lambda$: 
	\begin{enumerate}
		\item \label{boundary-line-oracle-B}
		let $s=\upd_u(t)$, if $s>-T$ and $r_{s}\neq \perp$, then $\sigma(u)\gets r_{s}$; \hfill  (\emph{oracle query $\+B(s)$})
		\item \label{boundary-line-oracle-C}
		otherwise, $\sigma(u)\gets X_{t-1}(u)$; \hfill (\emph{oracle query $\+C(u)$})
	\end{enumerate}
	\item return $\sigma_{\Lambda}$.
\end{itemize}
  \end{tcolorbox} 
\end{center}
This implementation satisfies \Cref{cond:invariant-boundary} due to the conditional independence property mentioned earlier for Gibbs distributions.


In general, $\mu$ can be an arbitrary distribution over $[q]^V$ and the subroutine $\config(t)$ will be implemented specifically depending on $\mu$.
Nevertheless, the following proposition is easy to verify.
\begin{proposition}
As long as the subroutine $\config(t)$ satisfies \Cref{cond:invariant-boundary}, 
the process $\+P(T)$ generates a faithful copy of the systematic scan Glauber dynamics $(X_t)_{-T\le t\le 0}$ for $\mu$. 
\end{proposition}

\subsection{Evaluating stationary states via backward deductions}\label{section-backward-deduction}
Fix an arbitrary integer $T\ge 0$.
Consider the Markov chain $(X_t)_{-T\le t\le 0}$ generated by the process $\+P(T)$. 
We present an algorithm that outputs the random variable $X_0(v)$ for $v\in V$.
Instead of simulating the process $\+P(T)$ chronologically from time $-T$ to 0,
our algorithm uses a backward deduction which tries to infer the correct value of $X_0(v)$ by accessing as few random variables as possible for resolving $X_0(v)$.

The algorithm is described in \Cref{Alg:resolve}.
It is a recursive algorithm that has an input argument $t\le 0$ and maintains two global data structures $M$ and $R$, initialized respectively as  $M_0=\perp^{\mathbb{Z}}$ and $R_0=\emptyset$.
Since all recursive calls access and update the same $M$ and $R$, we sometimes write  $\resolve_T(t)=\resolve_T(t;M,R)$ for short.

For $-T<t\le 0$, $\resolve_T(t)$ tries to calculate the result of the update at time $t$, which is 
$X_t(v_{i(t)})$. In particular, for $-n< t< 0$, we have $X_t(v_{i(t)})=X_0(v_{(t\bmod n) +1})=X_0(v_{t+n+1})$, and $X_0(v_{i(0)})=X_0(v_{1})$.

\begin{algorithm} 
\caption{$\resolve_T(t; M,R)$} \label{Alg:resolve}
  \SetKwInput{KwPar}{Parameter}
  \SetKwInput{KwData}{Global variables}
 \KwIn{
an integer $t \leq 0$;}
\KwData{a map $M: \mathbb{Z}\to [q]\cup\{\perp\}$ and a set $R$;}

\KwOut{a value $X_t(v_{i(t)}) \in [q]$;}
\lIf{$t \leq -T$}{\Return $X_{-T}(v_{i(t)})$\label{Line-boundary}}
\lIf{$M(t) \neq \perp$}{\Return $M(t)$\label{Line-resolve-memoization}}
$M(t) \gets \basicsample(t; R)$ and \lIf{$M(t) \neq \perp$}{\Return $M(t)$\label{Line-resolve-direct-return}}
$\sigma_\Lambda \gets \config(t)$, with the oracle queries being replaced by   
{\begin{itemize}
	\item \textbf{upon} querying $\+B(s)$: 
      \textbf{if} {$s > -T$}, \textbf{then} {\Return $\basicsample(s; R)$}; \lElse{\Return $\perp$}
	\item \textbf{upon} querying $\+C(u)$: \Return $\resolve_T(\upd_u(t); M,R)$, where $\upd_u(t)$ is defined in~\eqref{eq:update-time}\;
\end{itemize}}\label{Line-resolve-item}
 $M(t) \gets$ a random value drawn independently according to $\mu^{\text{pad},\sigma_\Lambda}_{v}$\; \label{Line-resolve-sample}
\Return $M(t)$\; \label{Line-resolve-final-return}
\end{algorithm}

The algorithm recursively deduces the outcome of the update at time $t\le 0$.
The data structure $M(t)$ stores the resolved outcomes of the updates at time $t$, and the set $R$ stores the generated samples from the lower bound distribution $\mu^{\-{LB}}$.
Each value in $M$ will only be updated at most once, and the set $R$ will never remove elements.
Moreover, to implement the subroutine $\config(t)$,
each query to the lower bound oracle $\+{B}(s)$ is replaced by accessing the sample $r_s$ from the lower bound distribution $\mu^{\-{LB}}$, realized by \Cref{Alg:randombit};
each query to the configuration oracle $\+{C}(u)$ is replaced by a recursive call to $\resolve_T(\upd_u(t); M,R)$ for evaluating $X_{t-1}(u)$.
Here, the \emph{principle of deferred decision} is applied, so that the decision of the random choices of the $(r_s)_{-T<s<t}$ in $\+P(T)$ is deferred to the moments when they are accessed in the backward deduction.

\begin{algorithm}
\caption{$\basicsample(t;R)$} \label{Alg:randombit}
\SetKwInput{KwData}{Global variables}
 \KwIn{an integer $t\le 0$;}
 \KwData{a set $R$ of pairs $(s,r_s)\in\mathbb{Z}\times ([q]\cup\{\perp\})$;}
 \KwOut{a random value in $[q]\cup \set{\perp}$ distributed as $\mu^{\mathrm{LB}}$.}
\lIf{$(t,r) \in R$}{\Return $r$}
\Else{
  	draw $r_t \sim \mu^{\mathrm{LB}}$\;\label{alg-line:randombit-draw}
  	$R \gets R \cup \{(t,r_t)\}$\;
  	\Return $r_t$\;
  }
\end{algorithm}

When $T$ is set to $\infty$, the program $\resolve_{\infty}(t)$ is still well-defined, as the only difference is that $t\le -T$ never triggers.
Indeed $\resolve_{\infty}(t)$ is what we call the coupling towards the past process.
It tries to calculate the result of the update at time $t\le 0$ in a chain running from the infinite past to time $0$.
By \Cref{thm-convergence}, if the systematic scan Glauber dynamics $\+{P}(T)$ is irreducible, then the distribution of $X_0$ converges to $\mu$ as $T\to\infty$, regardless of the initial state. 
We give a sufficient condition for both the convergence of  the forward process  $\+{P}(T)$ and the termination of the backward program $\resolve_{\infty}(t)$.

\begin{condition}\label{condition:sufficient-correctness}
The lower bound distribution $\mu^{\-{LB}}$ in \Cref{definition:margin-lower-bound} satisfies that
 $\mu^{\-{LB}}(\perp)<1$.
\end{condition}

The next theorem states that for any finite~$T\ge 0$, \Cref{Alg:resolve}  always correctly evaluates the chain;
and by setting $T=\infty$, it returns a sample from the marginal distribution $\mu_v$ with probability $1$ if \Cref{condition:sufficient-correctness} holds.

\begin{theorem}\label{theorem:resolve-cor}
Let $\mu$ be a distribution over $[q]^V$, $T\ge 0$ be an integer, and 
$(X_t)_{-T\le t\le 0}$ be generated by the process $\+{P}(T)$ whose $\config(t)$ subroutine satisfies \Cref{cond:invariant-boundary}.
For any $-T\le t \leq 0$, the followings hold:
\begin{itemize}
\item
$\resolve_{T}(t)$ terminates in finite steps and returns a sample identically distributed as $X_t(v_{i(t)})$;
\item
if further \Cref{condition:sufficient-correctness} holds, then $\resolve_{\infty}(t)$ terminates with probability $1$, and  returns a sample distributed as $\mu_v$ where $v=v_{i(t)}$.
\end{itemize}
\end{theorem}
\begin{proof}
Fix a finite $T \geq 0$. It is straightforward to see  that $\resolve_{T}(t)$ terminates in finite steps since the time $t$ decreases in the recursive calls to $\resolve_T$ and the procedure terminates once $t\le -T$.
%

We then show that the output of $\resolve_{T}(t)$ is identically distributed as $X_t(v_{i(t)})$ by a coupling between the process $\+{P}(T)$ and  $\resolve_{T}(t)$ for $-T \leq t \leq 0$.
Consider the following implementations of $\+P(T)$ and $\resolve_{T}(t)$.
For each $- T <  \ell \leq 0$, let $U_\ell \in [0,1)$ be a real number sampled uniformly and independently at random.
%
Whenever the process $\+P(T)$ tries to draw an $X_\ell(v_{i(\ell)})$ according to the padding distribution $\mu^{\mathrm{pad},\sigma_\Lambda}_{v_{i(\ell)}}$, 
we simulate this by assigning $X_\ell(v_{i(\ell)})$ the value $c \in [q]$ satisfying
\begin{align}\label{eq:proof:resolve-coupling}
U_\ell\in\left[\,\sum_{j=1}^{c-1} \mu^{\mathrm{pad},\sigma_\Lambda}_{v_{i(\ell)}} (j),\,\,\sum_{j=1}^{c} \mu^{\mathrm{pad},\sigma_\Lambda}_{v_{i(\ell)}} (j)\,\right).
\end{align}
Similarly, in the procedure $\resolve_{T}(t)$, when a recursion $\resolve_{T}(\ell; M,R)$ tries to draw an $M(\ell)$ according to $\mu_{v_{i(\ell)}}^{\-{pad},\sigma_{\Lambda}}$ in \Cref{Line-resolve-sample} of \Cref{Alg:resolve},
we simulate this by assigning $M(\ell)$ the same $c \in [q]$ satisfying \eqref{eq:proof:resolve-coupling}.
It is easy to see that such implementations faithfully simulate the original processes $\+P(T)$ and $\resolve_{T}(t)$ respectively.

Besides the random choices for sampling from the padding distributions, which are provided by the sequence $(U_\ell)_{-T<\ell\leq 0}$, the only remaining random choices in the two processes $\resolve{}_{T}(t)$ and $\+P(T)$ are the outcomes $(r_{\ell})_{-T<\ell \leq 0}$ for sampling from the lower bound distribution $\mu^{\-{LB}}$.

We define the following coupling between $\resolve{}_{T}(t)$ and $\+P(T)$:
\begin{itemize}
	\item the two processes use the same random choices for $(r_{\ell})_{-T<\ell \leq 0}$ and $(U_{\ell})_{-T<\ell \leq 0}$.
\end{itemize}

Fix any evaluation of the random choices $\mathbf{r}=(r_{\ell})_{-T < \ell \leq 0}$ and $\mathbf{U}=(U_{\ell})_{-T < \ell \leq 0}$.
Both $\resolve_T(t)$ and the $(X_{t})_{-T<t\le 0}$ generated according to $\+P(T)$ are fully deterministic given $(\mathbf{r},\mathbf{U})$. 
Furthermore, for any mapping $M:\mathbb{Z}\rightarrow [q]\cup \set{\perp}$ 
and any set $R$ of $(s,r'_s)\in\mathbb{Z}\times([q]\times \{\perp\})$ pairs, 
we say that $(M,R)$ are 
consistent with $(\mathbf{r},\mathbf{U})$, if the followings hold
\begin{itemize}
    \item $\forall -T\le\ell\leq 0:\quad M(\ell)\neq\perp\implies M(\ell)=X_{\ell}(v_{i(\ell)})$;
    \item $\forall -T< \ell\leq 0:\quad (\ell, r_\ell')\in R\implies r_\ell'=r_\ell$.
\end{itemize}

Under the coupling above, by an induction on $t$ from $-T$ to $0$, 
one can  routinely verify that for any $(M,R)$ consistent with $(\mathbf{r},\mathbf{U})$, 
the output of $\resolve_T(t;M,R)$ is precisely $X_{t}(v_{i(t)})$ generated according to $\+P(T)$ using the same random choices $(\mathbf{r},\mathbf{U})$,
and the states of $(M,R)$ during the execution of $\resolve_T(t;M,R)$ remain consistent with $(\mathbf{r},\mathbf{U})$. 
This shows that $\resolve_T(t)$ is identically distributed as $X_t(v_{i(t)})$ for any finite $T\ge 0$ and any $-T\le t\le 0$.
In other words,
\begin{align}\label{eqn:resolve-Xt}
  \DTV{X_{t}(v)}{\resolve_{T}(t)} = 0.
\end{align}

Next, we deal with the infinite case.
We claim that \Cref{condition:sufficient-correctness} implies both the irreducibility of $\+{P}(T)$ and the termination of $\resolve_{\infty}(t)$ with probability~1. 
%

%
For the irreducibility of $\+{P}(T)$, by \Cref{condition:sufficient-correctness}, there must exist $c_0\in [q]$ such that $\mu^{\-{LB}}(c_0)>0$.
Let $\sigma\in [q]^V$ be the constant configuration such that $\sigma(v)=c_0$ for all $v\in V$. 
Then by the definition of $\mu^{\-{LB}}$ it follows that $\sigma$ is feasible and also can be reached via transitions of $\+P(T)$ from all feasible configurations.
It is also straightforward to verify that for any two feasible configurations $\tau,\tau'\in [q]^V$ and any $-T<t\leq 0$, if $P_t(\tau,\tau')>0$ then $P_t(\tau,\tau')>0$, where $P_t$ denotes the one-step transition matrix of $\+P(T)$ at time $t$.
Therefore any feasible configuration is also reachable from $\sigma$. 
This shows the irreducibility of $\+{P}(T)$.

Then we show the termination of $\resolve_{\infty}(t_0)$ for any $t_0\leq 0$. 
For each $t\leq t_0$, define the event:
\[
  \+B_{t}: r_{s}\neq \perp \text{ for all }s\in [t-n+1,t].
\]
We claim that if $\+{B}_t$ happens for some $t\leq t_0$,  then no recursive calls of $\resolve_{\infty}(s;M,R)$ would be incurred for any $s\leq t-n$. Assume for the sake of contradiction that there exists a maximum $s^*\leq t-n$ such that $\resolve_{\infty}(s^*;M,R)$ is called. As $s^*\leq t-n< t_0$, $\resolve_{\infty}(s^*;M,R)$ must be recursively called directly within another instance of $\resolve_{\infty}(s';M,R)$ such that $s^*<s'$. Note that by \Cref{Line-resolve-item} of \Cref{Alg:resolve} and \eqref{eq:update-time} we also have $s^*>s'+n$. We then have two cases:
\begin{enumerate}
    \item $s'\leq t-n$, this contradicts the maximality assumption for $s^*$.
    \item Otherwise $s'>t-n$. By $s^*\leq t-n$ and $s^*>s'+n$ we have $s'\in [t-n+1,t]$. 
      Also by the assumption that $\+{B}_t$ happens, we have $r_{s'}\neq \perp$, therefore $\resolve_{\infty}(s';M,R)$ would have terminated at \Cref{Line-resolve-direct-return} of \Cref{Alg:resolve} without incurring any recursive call. This also leads to a contradiction and thus proves the claim.
\end{enumerate}

Note that by \Cref{condition:sufficient-correctness}, for any $t\leq t_0$, we have $\Pr{\+{B}_t}\ge p\defeq (1-\mu^{\-{LB}}(\perp))^n >0$.
For any $L>0$, let $\+E_L$ be the event that there is a recursive call to $\resolve_{\infty}(t^*;M,R)$ where $t^* \le t_0-L n$.
By the claim above,
\begin{align*}
  \Pr{\+E_L} \le \Pr{\bigwedge\limits_{j=0}^{L-1}\tp{\neg \+{B}_{t_0-jn}}} = \prod\limits_{j=0}^{L-1}\Pr{\neg \+{B}_{t_0-jn}}\leq (1-p)^L,
\end{align*}
where the equality is due to independence of $(r_t)_{t\leq t_0}$.
Thus, with probability $1$ there is only a finite number of recursive calls, namely $\resolve_{\infty}(t_0)$ terminates with probability $1$.

For any $t\le 0$, since $\resolve_{\infty}(t)$ terminates with probability $1$,
its output distribution is well defined.
For any $\eps>0$, consider a sufficiently large $L$ such that $(1-p)^L\le \eps$.
For any $T\ge Ln-t$, we couple $\resolve_{\infty}(t)$ with $\resolve_{T}(t)$ using the same random variables drawn.
As the coupling fails only if $\+E_L$ happens, by the coupling lemma,
\begin{align*}
  \DTV{\resolve_{T}(t)}{\resolve_{\infty}(t)} \le \Pr{\+E_L}\le \eps.
\end{align*}
This implies 
\begin{align}\label{eqn:resolve-dtv}
  \lim_{T\rightarrow\infty}\DTV{\resolve_{T}(t)}{\resolve_{\infty}(t)}=0.
\end{align}
The irreducibility of $\+{P}(T)$ and \Cref{thm-convergence} implies that
\begin{align}\label{eqn:limiting-dist}
  \lim_{T\rightarrow\infty}\DTV{\mu_v}{X_{T,t}(v)}=0,
\end{align}
where $v=v_{i(t)}$ and $X_{T,t}$ is the state of $\+P(T)$ at time $t\ge T$. 
Moreover, for any $T>0$,
\begin{align}\label{eqn:triangle-DTV}
  \DTV{\mu_v}{\resolve_{\infty}(t)} &\le \DTV{\mu_v}{\resolve_{T}(t)} + \DTV{\resolve_{T}(t)}{\resolve_{\infty}(t)}.
\end{align}
Combining the above, we have
\begin{align*}
  &\DTV{\mu_v}{\resolve_{\infty}(t)} \\
  \le~&~\limsup_{T\rightarrow\infty}\DTV{\mu_v}{\resolve_{T}(t)}  + \limsup_{T\rightarrow\infty}\DTV{\resolve_{T}(t)}{\resolve_{\infty}(t)}\tag{by \eqref{eqn:triangle-DTV}}\\
  =~&~\limsup_{T\rightarrow\infty}\DTV{\mu_v}{\resolve_{T}(t)} \tag{by \eqref{eqn:resolve-dtv}}\\
  \le~&~\limsup_{T\rightarrow\infty}\DTV{\mu_v}{X_{T,t}(v)} + \limsup_{T\rightarrow\infty}\DTV{X_{T,t}(v)}{\resolve_{T}(t)} = 0, \tag{by \eqref{eqn:limiting-dist} and \eqref{eqn:resolve-Xt}}
\end{align*}
and the theorem follows.
\end{proof}

Note that in the proof above, $L$ has to be at least $\Omega(p^{-1})$ to guarantee a constant probability upper bound for the event $\+E_L$.
Since $p$ is an exponentially small quantity, this argument usually does not yield a useful mixing time bound.
This is mainly a proof for convergence and correctness.

\begin{remark}  \label{rem:CTTP}
  In fact, if $\resolve_{\infty}(t)$ terminates with probability $1$ for any $-n<t\le 0$,
  then the Glauber dynamics has to be irreducible.
  This is because, fix an arbitrary set of random choices,
  and run $\resolve_{\infty}(t)$ for all $-n<t\le 0$.
  Because of the termination assumption,
  there is a finite $T<0$ beyond which all of them terminates,
  and we obtain a configuration $\sigma$ such that $\sigma(v_i) = \resolve_{\infty}(\upd_i(0))$.
  It is then straightforward to see that if we run $P(T)$ under the same random choices,
  no matter what the initial configuration is, the configuration at time $0$ is exactly $\sigma$.
  This implies that all feasible configurations can reach $\sigma$ and vice verse,
  and thus the Glauber dynamics is irreducible.
  
  On the other hand, $\resolve_{\infty}(t)$ terminating with probability $1$ for any $-n<t\le 0$ is not that far away from \Cref{condition:sufficient-correctness}.
  Consider a distribution $\mu$.
  It is possible for $\resolve_{\infty}(t)$ to terminate with probability $1$ if there is some $v$ such that $\mu_v^{\-{LB}}(\perp)=1$.
  However, as soon as there are two variables $u$ and $v$ such that $\mu_v^{\-{LB}}(\perp)=\mu_u^{\-{LB}}(\perp)=1$,
  and one cannot deduce the value of $u$ (or $v$) without knowing the value of $v$ (or $u$),
  $\resolve_{\infty}$ will not terminate when resolving either $u$ or $v$.
\end{remark}

\subsection{Truncated simulation using bounded randomness}\label{sec:approx-resolve}
\Cref{theorem:resolve-cor} implies that under \Cref{condition:sufficient-correctness},
the coupling towards the past process, $\resolve_{\infty}(t)$, is a perfect sampler for the marginal distribution of $\mu_v$.
In our typical applications, the expected running time of this algorithm is often a constant.
However, since our end goal is derandomisation, we need an algorithm that draws no more than logarithmic random variables in the worst case.
This leads to the truncated version in \Cref{Alg:appresolve}.
Similar to \Cref{Alg:resolve}, we may sometimes drop $M$ and $R$ from the input as all recursive calls access the same data structures.

%


\begin{algorithm}
\caption{$\appresolve(t, K; M,R)$} \label{Alg:appresolve}
  \SetKwInput{KwPar}{Parameter}
  \SetKwInput{KwData}{Global variables}
  \SetKwIF{Try}{Catch}{Exception}{try}{:}{catch}{exception}{}
\KwIn{integers $t \leq 0$ and $K \geq 0$}
\KwData{a map $M:\mathbb{Z}\to[q]\cup\{\perp\}$ and a set $R$;}
\KwOut{a random value in $[q]\cup\{\perp\}$;}
initialize $M\gets\perp^{\mathbb{Z}}$ and $R\gets\emptyset$\;
\uTry{}{
\Return $\resolve_{\infty}(t;M,R)$\;
}
\Catch{$|R| \geq K$\label{alg:appresolve-line-catch}}{ 
\Return $\perp$\;
}
\end{algorithm}

The algorithm $\appresolve(t, K)$ simulates $\resolve_{\infty}(t)$ using a set $R$ of size bounded by $K$.
Recall that the set $R$ is used to store the generated samples $(r_s)_{s\le t}$ from the lower bound distribution $\mu^{\-{LB}}_v$.
Moreover, since a sample from the padding distribution $\mu_v^{\-{pad},\cdot}$ is drawn only after a corresponding sample from the lower bound distribution has been drawn, we have the following observation.

\begin{observation}
In \Cref{Alg:appresolve}, at most $K$ samples are drawn from the lower bound distribution $\mu^{\-{LB}}_v$ and at most $K$ samples are drawn from the padding distributions $\mu_v^{\-{pad},\cdot}$.
\end{observation}

Denote by $\+E_{\-{trun}}(K)$  the event  $\appresolve(t,K)=\perp$, i.e.~the exception at \Cref{alg:appresolve-line-catch} occurs.
The following theorem says that this is precisely the error for $\appresolve(t,K)$ sampling from $\mu_{v_{i(t)}}$.

\begin{theorem}\label{thm:dtv-truncate}
Let $\mu$ be a distribution over $[q]^V$. Assume \Cref{cond:invariant-boundary} and \Cref{condition:sufficient-correctness}.
For $-n<t\le 0$, $K\ge 0$, 
and $Y=\appresolve(t,K)$, 
it holds that  $\DTV{Y}{\mu_v}= \Pr{\+E_{\-{trun}}(K)}$, where $v=v_{i(t)}$.
\end{theorem}

\begin{proof}[Proof]
Suppose that sampling from the padding distribution $\mu_{v_{i(\ell)}}^{\-{pad},\sigma_{\Lambda}}$ in \Cref{Line-resolve-sample} of \Cref{Alg:resolve} is realised in the same way as in~\eqref{eq:proof:resolve-coupling}, 
using a sequence of real numbers $U_\ell \in [0,1)$ chosen uniformly and independently at random for each $\ell\le 0$.
%
%

We apply the following coupling between $\resolve{}_{\infty}(t)$ and $\appresolve{}(t,K)$:
\begin{itemize}
	\item the two processes use the same random choices for $(r_{\ell})_{\ell \leq 0}$ and $(U_{\ell})_{\ell \leq 0}$.
\end{itemize}
One can verify that if  $\+E_{\-{trun}}(K)$ does not occur, then the two processes $\resolve{}_{\infty}(t)$ and $\appresolve{}(t,K)$ are coupled perfectly. 
By the coupling lemma and \Cref{theorem:resolve-cor}, we have
\begin{align*}
\DTV{Y}{\mu_v} \leq \Pr{\+E_{\-{trun}}(K)}.
\end{align*}
Note that $\Pr{Y = \perp} =  \Pr{\+E_{\-{trun}}(K)}$ and $\mu_v(\perp) = 0$. Therefore, we have
\begin{align*}
\DTV{Y}{\mu_v} \geq \Pr{Y = \perp} - \mu_v(\perp) = \Pr{\+E_{\-{trun}}(K)}.
\end{align*}
Combining the two inequalities proves the theorem.
\end{proof}

A well-known limitation of various perfect sampling approaches, including the coupling from the past (CFTP), 
is that the sampling procedure is  ``non-interruptible'', which means that forced early termination would introduce a bias in sampling.
The same also holds for the coupling towards the past (CTTP) process $\resolve_{\infty}(t)$.
However, \Cref{thm:dtv-truncate} guarantees that this bias due to the interruption is bounded by the truncation probability.
In particular, when applying the $\appresolve(t,K)$ to draw approximate samples from the marginal distributions, 
we are especially interested in the cases where $\Pr{\+E_{\-{trun}}(K)}\le {1}/{\poly(n)}$ is achieved by a $K=O(\log n)$.

Finally, the ``convergence rate'' of the $\resolve_{\infty}(t)$ procedure, represented by the upper bound for $\Pr{\+E_{\-{trun}}(K)}$, 
depends very much on how the $\config(t)$ subroutine is implemented, which may vary on concrete models.
This is also quite similar to the case of CFTP, where the analyses of convergence rates are also highly dependent on its implementations.
Therefore,  in the current section, we do not give a generic analysis of the convergence rate of the CTTP procedure. 
Instead, in the next two sections,
we will analyse the efficiencies of CTTP for hypergraph independent set and hypergraph colouring, where the Markov chains with concrete implementations of the $\config(t)$ subroutine are specified.

\section{Hypergraph independent set}\label{section:hyper-indset}


In this section, we give FPTASes for counting hypergraph independent sets and prove \Cref{thm-h-ind} and \Cref{thm-h-ind-simple}. 
We first introduce some notations.
Given a $k$-uniform hypergraph $H=(V,\+E)$ with maximum degree $\Delta$, denote by $\Omega_H \subseteq \{0,1\}^V$ the set of all independent sets of $H$, and $Z_H=\abs{\Omega_H}$ its size. 
Let $\mu = \mu_H$ be the uniform distribution over $\Omega_H$.
We identify a subset $S\subseteq V$ with an assignment $\tau_S\in\{0,1\}^V$ by $\tau_S(v)=1$ if and only if $v\in S$, for any $v\in V$.
Recall that $\tau\in \{0,1\}^V$ is a hypergraph independent set if every hyperedge $e \in \+E$ contains at least one vertex $v \in e$ with $\tau_v = 0$.

Assume $V = \{v_1,v_2,\ldots,v_n\}$ where $ n = |V|$.
To approximate the partition function $Z_H$, namely the number of independent sets of $H$, 
we use the vertex decomposition defined in~\eqref{eq-vtx-decom} with the decomposition scheme $\sigma = \mathbf{0} $ (all-zero vector).
This reduces the task to approximating $\mu_{v_i}^{\sigma_{<i}}(0)$ for each $i\in[n]$, the probability that $v_i$ takes the value $0$ conditional on all of $v_j$ take $0$ where $j<i$.
One reason for choosing the all-zero scheme $\sigma = \mathbf{0}$ is the following marginal lower bound (recall \Cref{definition:margin-lower-bound}) by the nature of independent sets.
\begin{observation}\label{ob-hind}
For any $\Lambda\subseteq V$, any $\sigma_{\Lambda}\in \{0,1\}^{\Lambda}$ and any $v\in V\setminus \Lambda$, it holds that
\[
  \mu^{\sigma_{\Lambda}}_v(0) \geq \frac{1}{2} \text{ and } \mu^{\sigma_{\Lambda}}_v(1) \geq 0.
\]
And therefore, $\sigma = \mathbf{0}$ is a $\frac{1}{2}$-bounded vertex decomposition scheme.
\end{observation}
The lower bound for $\mu^{\sigma_{\Lambda}}_v(1)$ can be improved slightly.
It will not have any substantial improvement, 
so for clarity we do not pursue that.

Next, consider the distribution $\mu_{v_i}^{\sigma_{<i}}$ obtained by pinning the partial configuration $\sigma_{<i}$ on the hypergraph $H$. 
As another reason of choosing $\sigma = \mathbf{0}$, observe that for each vertex $v$, if the value of $v$ is fixed to be $0$, then all the constraints arising from hyperedges incident to $v$ get immediately satisfied.
Therefore, these hyperedges can be safely pruned away from the subinstance. 
Suppose we are given the partial configuration $\sigma_{<i}$ at some stage during the computation of the product in~\eqref{eq-vtx-decom}. 
All vertices with index at most $i-1$, together with hyperedges incident to \emph{any} of them, can be safely removed. 
This gives the hypergraph $H_i = (V_i,\+E_i)$ where $V_i = \{v_i,v_{i+1},\ldots,v_n\}$ and $\+E_i = \{e \in \+E \mid e \subseteq V_i\}$. 
It is straightforward to verify that 
\begin{itemize}
	\item  $\mu_{v_i}^{\sigma_{<i}} = \mu'_{v_i}$, where $\mu'$ is the uniform distribution over all independent sets in $H_i$;
	\item $H_i$ is a $k$-uniform hypergraph;
	\item the maximum degree of $H_i$ is at most that of $H$;
	\item $H_i$ is a linear hypergraph if $H$ is a linear hypergraph.
\end{itemize}
Hence, sampling from the distribution $\mu_{v_i}^{\sigma_{<i}}$ is equivalent to sampling from $\mu'_{v_i}$. 
Moreover, if $H$ satisfies the condition of \Cref{thm-h-ind} (or the condition of \Cref{thm-h-ind-simple}), then so does $H_i$. 
Using the ``reduction'' from single-site samplers to deterministic counting algorithms as per \Cref{cor-reduction}, it suffices to design single-site samplers satisfying the conditions therein. 
We summarise this into the following two lemmata, one for the general case and the other for the linear case. 
\begin{lemma}\label{lem-sampling-ind}
Let $k,\Delta \geq 2$ be two integers satisfying $2^{\frac{k}{2}}\geq \sqrt{8\mathrm{e}}k^2\Delta$.
There exists an algorithm that 
given as inputs a $k$-uniform hypergraph $H = (V,\+E)$ with maximum degree at most $\Delta$, a vertex $v \in V$ and a parameter $\gamma > 0$,
outputs a random $Y_v \in \{0,1\}$ such that $\DTV{Y_v}{\mu_v} \leq \gamma$, where $\mu$ is the uniform distribution over all independent sets in $H$.
The algorithm runs in time $O(\Delta^3k^5\log \frac{1}{\gamma})$ 
and draws at most $3\Delta^2 k^4 \ceil{\log \frac{1}{\gamma}}$ Boolean random variables. 
\end{lemma}

\begin{lemma}\label{lem-sampling-ind-simple}
Let $\delta > 0$.
Let $k \geq \frac{25(1+\delta)^2}{\delta^2}$ and $\Delta \geq 2$ be two integers satisfying 
$2^{k} \geq (100 k^3 \Delta)^{1+\delta}$.
There exists an algorithm that 
given as inputs a $k$-uniform linear hypergraph $H = (V,\+E)$ with maximum degree at most $\Delta$, a vertex $v \in V$ and a parameter $\gamma > 0$, 
outputs a random $Y_v \in \{0,1\}$ such that $\DTV{Y_v}{\mu_v} \leq \gamma$, where $\mu$ is the uniform distribution over all independent sets in $H$.
The algorithm runs in time $O((\frac{1+\delta}{\delta})^2\Delta^4 k^{10} \log \frac{1}{\gamma})$ 
and draws at most $10^4 (\frac{1+\delta}{\delta})^2 \Delta^3 k^9 \ceil {\log{\frac{1}{\gamma}}}$ Boolean random variables.
\end{lemma}

%

\begin{remark}[\textbf{perfect marginal samplers}]\label{rem:perfect-marginal-HIS}
  Applying coupling towards the past (\Cref{Alg:resolve} with $T=\infty$) to the uniform distribution of hypergraph independent sets yields a perfect marginal sampler for the marginal distributions $\mu_v$ 
  that outputs $Y^*_v$ distributed exactly as $\mu_v$ upon termination.
  Indeed, the algorithms stated in \Cref{lem-sampling-ind} and \Cref{lem-sampling-ind-simple} are obtained from truncating it.
  For any $\gamma>0$, with probability at least $1-\gamma$, the perfect marginal samplers terminate: 
  \begin{itemize}
    \item
      (on $k$-uniform hypergraphs)
      within $O\tp{\Delta^3k^5\log \frac{1}{\gamma}}$ time for computation, 
      while drawing at most $3\Delta^2 k^4 \ceil{\log \frac{1}{\gamma}}$ Boolean random variables;
    \item
      (on $k$-uniform linear hypergraphs)
      within $O\tp{(\frac{1+\delta}{\delta})^2\Delta^4 k^{10} \log \frac{1}{\gamma}}$ time for computation, 
      while drawing at most $10^4 (\frac{1+\delta}{\delta})^2 \Delta^3 k^9 \ceil {\log{\frac{1}{\gamma}}}$ Boolean random variables;
  \end{itemize}
  under the same conditions as stated in  \Cref{lem-sampling-ind} and \Cref{lem-sampling-ind-simple}, respectively.
  Similar to CFTP, these perfect samplers are ``non-interruptible'' in the sense that truncations may bias the sample.
  Nevertheless, due to the tail bounds above, the truncating error introduced is bounded by~$\gamma$ in the total variation distance.
\end{remark}

\Cref{thm-h-ind} and \Cref{thm-h-ind-simple} are straightforward consequences of \Cref{lem-sampling-ind} and \Cref{lem-sampling-ind-simple} respectively.
We first show \Cref{lem-sampling-ind} in \Cref{sec-ind-scan}, 
and then adapt the proof to the linear case and show \Cref{lem-sampling-ind-simple} in \Cref{sec-linear-ind}.

\begin{proof}[Proofs of \Cref{thm-h-ind} and \Cref{thm-h-ind-simple}]


\Cref{thm-h-ind} follows from \Cref{lem-sampling-ind} with $\gamma = \frac{\epsilon}{20n}$, \Cref{cor-reduction} with the vertex decomposition scheme $\sigma = \mathbf{0}$ and \Cref{ob-hind}. 
The running time of the approximate counting algorithm is 
\begin{align*}
	T &= O\tp{ \left(n+ \frac{\Delta n}{k}\right) \cdot \tp{ \Delta^3k^5\log \frac{20 n}{\epsilon}} \times   2^{  3\Delta^2 k^4 \ceil{\log \frac{20n}{\epsilon}}}  } = \mathrm{poly}(\Delta k) \tp{\frac{n}{\epsilon}}^{O(\Delta^2 k^4)}.
\end{align*}

\Cref{thm-h-ind-simple} is proved the same way but with \Cref{lem-sampling-ind-simple} invoked. 
The running time of the deterministic approximate counting algorithm is 
\begin{align*}
  T &= O\tp{ \frac{\Delta n}{k} \cdot \tp{\tp{\frac{1+\delta}{\delta}}^2\Delta^4 k^{10} \log \frac{20n}{\epsilon}} \cdot 2^{ 10^4 \tp{\frac{1+\delta}{\delta}}^2 k^9 \Delta^3  \ceil {\log{\frac{20n}{\epsilon}}}} } = \mathrm{poly}\tp{\frac{\Delta k(1+\delta)}{\delta}}\tp{\frac{n}{\epsilon}}^{O\tp{ \frac{\Delta^3 k^9 (1+\delta)^2}{\delta^2} }}. \qedhere
\end{align*}
\end{proof}

%

\subsection{Systematic scan Glauber dynamics for hypergraph independent sets}\label{sec-ind-scan}
To apply the general algorithm framework as in \Cref{Alg:appresolve}, we first identify the lower bound and marginal distribution in \Cref{definition:margin-lower-bound} using \Cref{ob-hind}. 
\begin{itemize}
    \item The lower bound distribution $\mu^{\-{LB}}$ is given by
\begin{align*}
	 \mu^{\-{LB}}(0)=\mu^{\-{LB}}(\perp)=\frac{1}{2}, \quad \mu^{\-{LB}}(1)=0.
\end{align*}
Hence, we can assume that $\mu^{\-{LB}}$ is defined over $\{0,\perp\}$ such that $\mu^{\-{LB}}(0)=\mu^{\-{LB}}(\perp)=\frac{1}{2}$.
    \item The padding distribution $\mu_v^{\-{pad},\sigma_\Lambda}$, given any $\Lambda\subseteq V$, any $\sigma_{\Lambda}\in \{0,1\}^{\Lambda}$ and any $v\in V\setminus \Lambda$, is defined by
    \begin{align*}
        \mu_v^{\-{pad},\sigma_\Lambda}(0)=2\mu^{\sigma_\Lambda}_v(0)-1 \text{ and } \mu_v^{\-{pad},\sigma_\Lambda}(1)=2\mu^{\sigma_\Lambda}_v(1).	
    \end{align*}
\end{itemize}
Apparently, the Glauber dynamics for $\mu$ satisfies \Cref{condition:sufficient-correctness}. 
We then use \Cref{Alg:appresolve} on the distribution $\mu=\mu_H$ to prove \Cref{lem-sampling-ind} and \Cref{lem-sampling-ind-simple}, where the subroutine $\config(t)$ is defined in \Cref{Alg:config-h-indset}.
Recall that two oracles $\+B$ and $\+C$ in \Cref{Alg:config-h-indset} are defined in \Cref{section-scan-GD}.

\begin{algorithm}
\caption{$\config(t)$ for hypergraph independent sets } \label{Alg:config-h-indset}
  \SetKwInput{KwPar}{Parameter}
 \KwIn{hypergraph $H=(V,\+{E})$ specifying the distribution $\mu$ and an integer $t\leq 0$\;}
\KwOut{a partial configuration $\sigma_{\Lambda}\in \{0,1\}^{\Lambda}$ over some $\Lambda\subseteq V\setminus \{v_{i(t)}\}$;}
$v \gets v_{i(t)}, \Lambda \gets \emptyset, \sigma_{\Lambda}\gets \emptyset$\;
\ForAll{$e\in \+E$ s.t. $v \in e$\label{Line-indset-incident}}{ 
    \uIf{for all $u \in e \setminus \{v\}$, $\+{B}(\upd_u(t))=\perp$\label{Line-indset-satisfy}}{
        \ForAll{$u \in e \setminus \{v\}$}{$\Lambda \gets \Lambda \cup \{u\}$\;
        $\sigma_{\Lambda}(u)\gets \+{C}(u)$\; \label{Line-indset-assign}}
        \If{for all $u \in e \setminus \{v\}$, $\sigma_{\Lambda}(u)=1$\label{Line-indset-check}}{
            \Return $\sigma_{\Lambda}$\; \label{Line-indset-direct-return}
        }
    }
    \Else{
        \ForAll{$u \in e \setminus \{v\}$}{
            \If{$\+{B}(\upd_u(t))\neq \perp$}{
                $\Lambda \gets \Lambda \cup \{u\}$\;
                $\sigma_{\Lambda}(u)\gets 0$\; \label{Line-indset-assign-2}
            }
        }
    }
} 
\Return $\sigma_{\Lambda}$\; \label{Line-indset-final-return}
\end{algorithm}


\begin{lemma}\label{lemma:config-h-indset-cor}
$\config(t)$ in \Cref{Alg:config-h-indset} satisfies \Cref{cond:invariant-boundary}.
\end{lemma}

\begin{proof}

Let $X_{t-1}$ and $(r_{j})_{-T < j < t}$ be as defined in $\mathcal{P}(T)$.
Then
\begin{align}\label{eq-proof-find}
	\forall u \neq v, \qquad \+B(\upd_u(t)) \neq \perp \quad \Longrightarrow \quad  X_{t-1}(u) = \+B(\upd_u(t)) = 0,
\end{align}
where the first equality is due to the construction of $X_{t-1}$ as in $\mathcal{P}(T)$, and the second one is due to the definition of the lower bound oracle $\mathcal{B}$ and $\mu^{\-{LB}}(1)=0$. 

It is straightforward to see that \Cref{Alg:config-h-indset} is guaranteed to terminate. We then show that the output $\sigma_{\Lambda}$ of \Cref{Alg:config-h-indset} satisfies the conditional independence property in \Cref{cond:invariant-boundary}. 
Evidently $\sigma_{\Lambda} = X_{t-1}(\Lambda)$: 
if for some $u$ it holds that $\+B(\upd_u(t)) \neq \perp$ then it must be assigned $0$ in $\sigma_\Lambda$ in \Cref{Line-indset-assign-2} in \Cref{Alg:config-h-indset}, which is consistent with $X_{t-1}(u)$ because of~\eqref{eq-proof-find}; 
otherwise, its value is assigned using the configuration oracle $\+C(u)$ in \Cref{Line-indset-assign} in consistency with $X_{t-1}(u)$ by the definition of $\+C$. 
Then verify $\mu_{v}^{\sigma_\Lambda}=\mu_v^{X_{t-1}(V\setminus\{v\})}$ as below. 
\begin{itemize}
\item Suppose $\sigma_{\Lambda}$ is returned in \Cref{Line-indset-direct-return}. 
There exists $e \in \+E$ with $v \in e$ such that $e \setminus \{v\} \subseteq \Lambda$ and $X_{t-1}(u) = 1$ for all $u \in e \setminus \{v\}$. 
Conditional on such $\sigma_{\Lambda}$, $v$ must take the value 0, the same as conditioning on $X_{t-1}(V\backslash\{v\})$. 
\item Suppose $\sigma_{\Lambda}$ is returned in \Cref{Line-indset-final-return}. 
This implies that the condition in \Cref{Line-indset-check} has never been satisfied. 
Hence, for any $e \in \+E$ with $v \in e$, we either update $\sigma_{\Lambda}$ in \Cref{Line-indset-assign} or in \Cref{Line-indset-assign-2}. 
Fix an $e \in \+E$ with $v \in e$. 
If $\sigma_{\Lambda}$ is updated in \Cref{Line-indset-assign}, then there is a vertex $u \in e \setminus \{v\}$ such that $\sigma_{\Lambda}(u) = X_{t-1}(u) = \+C(u) = 0$. 
Otherwise, $\sigma_{\Lambda}$ is updated in \Cref{Line-indset-assign-2}. 
Then there is a vertex $u \in e \setminus \{v\}$ such that $\+{B}(\upd_u(t))\neq \perp$, and by~\eqref{eq-proof-find}, it must take $\sigma_{\Lambda}(u) = X_{t-1}(u) = 0$. 
In all, the constraints on all hyperedges incident to $v$, conditioned on either $\sigma_{\Lambda}$ or $X_{t-1}(V\setminus\{v\})$, are all satisfied, and thus $v$ take the value $\{0,1\}$ uniformly at random.
\end{itemize}
In both cases $\mu_{v}^{\sigma_\Lambda}=\mu_v^{X_{t-1}(V\setminus\{v\})}$. 
\end{proof}
The above argument also indicates the following observation: 
\begin{observation}
If $\sigma_{\Lambda}$ is returned in \Cref{Line-indset-direct-return}, then the padding distribution $\mu_v^{\-{pad},\sigma_{\Lambda}}$ takes value $0$ with probability $1$; if $\sigma_{\Lambda}$ is returned in \Cref{Line-indset-assign-2}, then the padding distribution $\mu_v^{\-{pad},\sigma_{\Lambda}}$ takes value $1$ with probability~$1$.
\end{observation}

By adapting the subroutine $\config(t)$ (\Cref{Alg:config-h-indset}) to $\appresolve$ (\Cref{Alg:appresolve}) and setting the threshold $K$ therein to $3\Delta^2 k^4 \lceil\log{\frac{1}{\gamma}}\rceil$,
we get an algorithm that approximately samples from the marginal distributions of $\mu$, 
where $\mu$ is uniform over all independent sets in $k$-uniform hypergraphs with the maximum degree at most $\Delta$. 
Clearing up this algorithm a bit gives \Cref{Alg:marginal-indset-sampler-app}. 

The algorithm maintains two global data structures $M$ and $R$ that might be updated as the algorithm proceeds. 
Upon the first invocation of this algorithm, 
these two global variables are initialised as $M=M_0=\perp^{\mathbb{Z}}$ and $R=R_0=\emptyset$. 
For simplicity of notations, 
we may drop both $M$ and $R$ from the list of parameters so long as the references to these two global structures are clear from the context. 
We may also drop the hypergraph $H$ from the list of parameters as it does not change throughout the process. 
Recall that the subroutine \basicsample{} in the algorithm is given in \Cref{Alg:randombit} with the lower bound distribution $\mu^{\-{LB}}$ given by $\mu^{\-{LB}}(0) = \mu^{\-{LB}}(\perp) = \frac{1}{2}$. 

\begin{algorithm}
\caption{$\appmghinds(t, \gamma;M,R)$} \label{Alg:marginal-indset-sampler-app}

  \SetKwInput{KwPar}{Parameter}
  \SetKwInput{KwData}{Global variables}
 \KwIn{a hypergraph $H = (V,\+E)$, an integer $t \leq 0$ and a real number $0 < \gamma < 1$;}
 \KwData{a map $M: \mathbb{Z}\to [q]\cup\{\perp\}$ and a set $R$; }

\KwOut{a random value in $\{0,1\}$}
\SetKwIF{Try}{Catch}{Exception}{try}{:}{catch}{exception}{}
\uTry{}{
\textbf{if} $M(t) \neq \perp$ \textbf{then} \Return $M(t)$\;\label{Line-indset-direct-return-1}
\textbf{if} $\basicsample{}(t;R) = 0$ \textbf{then} $M(t) \gets 0$ and \Return $0$\;\label{Line-indset-direct-return-2}
$v \gets v_{i(t)}$\label{Line-indset-setv-1}\;
\ForAll{$e\in \+E$ s.t. $v \in e$\label{Line-indset-incident-1}}{ 
    \If{for all $u \in e \setminus \{v\}$, ${\basicsample}(\upd_u(t);R)=\perp$\label{Line-indset-satisfy-1}}{
        \If{for all $u \in e \setminus \{v\}$, $\appmghinds(\upd_u(t),\gamma;M,R)=1$\label{Line-indset-recurse-1}}{
            $M(t) \gets 0$ and \Return $0$\; 
        }
    }
}
 $M(t) \gets 1$ and \Return $1$\; \label{Line-mt-final-return}
}
\Catch{$|R| \geq 3\Delta^2 k^4 \lceil\log{\frac{1}{\gamma}}\rceil$\label{Line-indset-catch}}{
    \Return $0$\;\label{Line-indset-force-return-1}
}
\end{algorithm}

%


The following lemma bounds the running time of  \Cref{Alg:marginal-indset-sampler-app}.
\begin{lemma}\label{lem-alg-h-2-time}
The running time of $\appmghinds(t,\gamma;M_0,R_0)$ (\Cref{Alg:marginal-indset-sampler-app}) is $O(\Delta^3 k^5 \log \frac{1}{\gamma})$.	
\end{lemma}
\begin{proof}
Consider the recursion tree generated by executing $\appmghinds(t,\gamma;M_0,R_0)$. 
Let $S$ be the set of all $t_0$ such that $\appmghinds(t_0,\gamma;M,R)$ is executed at least once. 
For any such $t_0 \in S$, one of the following two events must happen: 
(1) the whole algorithm terminates in \Cref{Line-indset-force-return-1} upon attempting to call $\basicsample$ with time $t_0$; 
(2) $(t_0,r_{t_0}) \in R$. 
Moreover, only one $t^* \in S$ triggers the first case, and hence $|S| \leq |R| + 1$. 
For any $t_0 \in S$, any execution of $\appmghinds(t_0,\gamma;M,R)$ beyond the first run returns on \Cref{Line-indset-direct-return-1}. 
Therefore, the total running time is bounded by $O(T^*|S|)$ where $T^*$ is the run time bound if we assume the recursive calls return in $O(1)$ time.
It is easy to see that $T^*=O(\Delta k)$.
Thus, the total running time is at most
\begin{align*}
	O(\Delta k |S|)  = O(\Delta k (|R|+1)) = O\tp{\Delta^3 k^5 \log \frac{1}{\gamma}}. &\qedhere
\end{align*}
\end{proof}

The correctness of \Cref{Alg:marginal-indset-sampler-app} is due to the following lemma. 
Its proof is given in \Cref{sec:proof-error-ind}. 
\begin{lemma}\label{lemma:indset-correct}
Let $0 < \gamma < 1$.
Let $H=(V,\+E)$ be a $k$-uniform hypergraph with maximum degree at most $\Delta$ such that $2^{\frac{k}{2}}\geq \sqrt{8\mathrm{e}}k^2\Delta$, 
and $\mu$ be the uniform distribution over all the independent sets in $H$. 
Then for any vertex $v \in V$, 
the output distribution $Y$ of $\appmghinds(\upd_v(0),\gamma)$ satisfies $\DTV{Y}{\mu_v} \leq \gamma$. 
\end{lemma}

Assuming the lemma above for now, we finish off \Cref{lem-sampling-ind}. 
\begin{proof}[Proof of \Cref{lem-sampling-ind}]
By \Cref{lemma:indset-correct}, for any $v\in V$, the algorithm $\appmghinds(\upd_v(0),\gamma)$ generates a distribution that is $\gamma$-close to $\mu_v$. 
Its running time is given in \Cref{lem-alg-h-2-time}. 
As each entry in $R$ must have been fetched from the random oracle exactly once, the number of calls to the random oracle equals $|R|$, which is bounded from above by $3\Delta^2 k^4 \lceil\log{\frac{1}{\gamma}}\rceil$. 
The random variables returned by the random oracle lie in $\{0,\perp\}$ because $\mu^{\-{LB}}(1)=0$. 
%
\end{proof}

\subsection{Analysis of the truncation error}\label{sec:proof-error-ind}
Unless stated otherwise, the hypergraph $H=(V,\+E)$ is fixed and satisfies the same condition as in \Cref{lemma:indset-correct} throughout this section. 
Let $t^*\leq 0$ be the timestamp of the initial call of $\appmghinds$. 

For any $\gamma>0$, consider the event that $\appmghinds(t^*,\gamma)$ terminates by \Cref{Line-indset-force-return-1}. 
Alternatively, it is equivalent to that the size of $R$ reaches $3\Delta^2 k^4 \lceil\log{\frac{1}{\gamma}}\rceil$ midst the execution of the algorithm (even if we assume the algorithm did not truncate). 
By~\Cref{thm:dtv-truncate}, it suffices to bound the probability of this event in order to bound the bias of the output distribution from the uniform distribution. 
This is given by the following lemma. 

\begin{lemma}\label{lemma:indset-efficiency}
Let $\gamma>0$ be a real, and $H=(V,\+E)$ be a $k$-uniform hypergraph with maximum degree at most $\Delta$ such that $2^{\frac{k}{2}}\geq \sqrt{8\mathrm{e}}k^2\Delta$. 
Let $\eta=\lceil\log{\frac{1}{\gamma}}\rceil$. 
Upon the termination of $\appmghinds{}(t^*,\gamma)$, the size of $R$ satisifies
    \begin{align}
        \Pr{\abs{R}\geq 3\Delta^2 k^4 \cdot \eta}\leq 2^{-\eta}.
    \end{align}
\end{lemma}

\Cref{lemma:indset-correct} follows directly from the above lemma, \Cref{lemma:config-h-indset-cor} and \Cref{thm:dtv-truncate}.



We now prove \Cref{lemma:indset-efficiency}.
For any $e\in \+E$ and $t\in \mathbb{Z}_{\leq 0}$, let
\begin{align}\label{eq:timestamp}
		\ts{}(e,t)\defeq\{\upd_v(t)\mid v\in e\}
\end{align}
be the set of timestamps of the latest consecutive updates of vertices in $e$ up to time $t$.
Given a hypergraph $H=(V,\+E)$, we introduce the following definition of the witness graph $G_H=(V_H, E_H)$. 

\begin{definition}[witness graph]\label{definition:witness-graph-indset}
For a hypergraph $H=(V,\+E)$, the \emph{witness graph} $G_H=(V_H,E_H)$ is an infinite graph with a vertex set 
\[
V_H=\{\ts{}(e,t)\mid e\in \+E,t\in \mathbb{Z}_{\leq 0}\},
\] 
and there is an undirected edge between $x,y \in V_H$ in $G_H$ if and only if $x\neq y$ and $x\cap y\neq\emptyset$.
\end{definition}
We remark that the vertex set $V_H$ is not a multiset, i.e., it does not allow multiple instances of the same element. 
For example, it is possible that for some $t\neq t'$, $\ts{}(e,t)=\ts{}(e,t')$. 
In this case, the two correspond to the same vertex even if $t\neq t'$. 
In addition, the witness graph $G_H$ is decided only by the hypergraph $H$ as $\upd_x(t)$ is a deterministic function. 
For any vertex $x \in V_H$ in the witness graph,
let $e_H(x) \in \+E$ be the edge in the original hypergraph that corresponds to $x$. 

We first bound the degree of the witness graph.   
Define the following neighbourhoods of $x$ in the witness graph $G_H$: 
\begin{align}\label{eq-def-ng}
\begin{split}
	N_{\-{self}}(x) &\defeq \{y \in V_H \mid \{x,y\} \in E_H   \text{ and } e_H(x) =  e_H(y)\},\\
	N_{\-{out}}(x) &\defeq \{y \in V_H \mid \{x,y\} \in E_H  \text{ and } e_H(x) \neq  e_H(y)\},\\
	N(x) &\defeq N_{\-{self}}(x) \uplus N_{\-{out}}(x).
	\end{split}
\end{align}
In other words, we partition the neighbourhood $N(x)$ of $x$ in the witness graph $G_H$ into two parts, one containing those corresponding to the same edge in the original hypergraph, and the other one collecting the rest. 
\begin{lemma}\label{lemma:wg-degree-bound}
For any vertex $x \in V_H$, it holds that $|N_{\-{self}}(x)| \leq 2k-2$ and $|N_{\-{out}}(x)| \leq (2k-1)k(\Delta - 1)$, and therefore, the maximum degree of $G_H$ is bounded by $2\Delta k^2 -2$.
\end{lemma}

\begin{proof}

Fix arbitrarily an edge $e\in \+{E}$ in the original hypergraph $H$ and a time $t\in \mathbb{Z}_{\leq 0}$. 
This induces a vertex $x$ in the witness graph $G_H$. 
We claim that for any $e'\in \+{E}$ such that $e\cap e'\neq \emptyset$, it holds that
\[\abs{\{\ts(e',t') \mid t'\in \mathbb{Z}_{\geq 0} \text{ and } \ts{}(e',t')\cap \ts{}(e,t)\neq \emptyset\}}\leq 2k-1.\]
This claim implies the lemma because
\begin{itemize}
\item by taking $e'=e$, we have $|N_{\-{self}}(x)| \leq 2k-2$ (we need to exclude $x$ itself); and
\item by taking $e'$ such that $e'\neq e$ and $e' \cap e \neq \emptyset$, we have $|N_{\-{out}}(x)| \leq (2k-1)k(\Delta - 1)$, since there are at most $(\Delta - 1) k$ choices for such $e'$. 
\end{itemize}


We then verify the claim. 
Assume $\ts(e,t)=(t_1,t_2,\dots,t_k)$ where $t_1<t_2<\dots<t_k<t_1+n$, and $\ts(e',t')=(t'_1,t'_2,\dots,t'_k)$ where $t'_1<t'_2<\dots<t'_k<t'_1+n$. 
Observe that $\ts{}(e',t')\cap \ts{}(e,t)\neq \varnothing$ implies $e\cap e'\neq \varnothing$. 
This means the following two quantities
\begin{align*}
  j_{\-{min}}\defeq&\min\left\{j\in[k] \mid v_{i(t_j)}\in e'\right\},~\text{and}\\
  j_{\-{max}}\defeq&\max\left\{j\in[k] \mid v_{i(t_j)}\in e'\right\}
\end{align*}
are well defined.  

Observe that $t'_k\geq t_{j_{\-{min}}}$ and $t'_1\leq t_{j_{\-{max}}}$. 
This can be argued as follows. 
Assume towards contradiction that $t'_k<t_{j_{\-{min}}}$. 
As $\ts{}(e',t')\cap \ts{}(e,t)\neq \varnothing$, there must be such $j$ and $j'$ that $t_j=t'_{j'}\in\ts{}(e',t')\cap \ts{}(e,t)$. 
Note that $t_j=t'_{j'}<t_{j_{\-{min}}}$ because of the ordering $t'_{j'}\leq t'_{k}$ and the assumption. 
This contradicts with the choice of $j_{\min}$. 
The same argument goes for the other inequality. 

The two inequalities above, together with $t'_k<t'_1+n$, implies $t'_1\in (t_{j_{\-{min}}}-n,t_{j_{\-{max}}}]$.  
Note further that in the interval $[t_{j_{\-{min}}}-n,t_{j_{\-{max}}}]$, there are at most $2k$ timestamps when there is a vertex in $e'$ that gets updated: 
$k$ from $[t_{j_{\-{min}}}-n,t_{j_{\-{min}}})$ and at most $k$ from $[t_{j_{\-{min}}},t_{j_{\-{max}}}]$
By definition of $j_{\-{min}}$, the vertex that gets updated at time $t_{j_{\-{min}}}$ must come from $e'$. 
Therefore, there are at most $2k-1$ timestamps in the interval $(t_{j_{\-{min}}}-n,t_{j_{\-{max}}}]$ that update vertices in $e'$,
and hence at most $2k-1$ choices for $t'_1$. 
The claim then follows by observing that the sequence $\ts{}(e',t')=(t'_1,t'_2,\cdots,t'_k)$ is uniquely determined by $t'_1$. 
\end{proof}

%
%
%
%
%
%
%
%
%
%
%

We analyse how $\appmghinds$ branches. 
For $t_0<t_1\leq t^*$, if $\appmghinds(t_0,\gamma)$ originates from $\appmghinds(t_1,\gamma)$,  
then we say the recursion $\appmghinds(t_0,\gamma)$ is triggered by $x=\ts{}(e,t_1)\in V_H$, 
where $e$ is the edge that the caller is handling. 
Also note that the same $t_0$ might be triggered by various different $t_1$'s, but only the first recursion $\appmghinds(t_0,\gamma)$ can trigger more recursive calls. 
Let
\begin{align}\label{eq:indset-bad-timestamps}
    V_{\+{B}}\defeq\{x\in V_H\mid \text{ $x$  triggers recursive calls }\} 
\end{align}
be a random subset of $V_H$. 
In fact, its randomness only comes from that used by the algorithm. 
Upon the termination of the algorithm, it generates the set $R$ of randomness used over time, and defines the set $V_{\+{B}}$. 
The size of $R$ can be upper bounded in terms of $|V_{\+B}|$.
\begin{lemma}\label{lemma:indset-randomness-bound}
$\abs{R}\leq (|V_{\+B}|+1)k^2\Delta$. 
\end{lemma}

\begin{proof}
Arbitrarily fix the values of all random bits ${(r_t)}_{t \leq 0}$.
Given this, $\appmghinds$ is deterministic. 
We show that the lemma always holds.

A timestamp $t_0$ is said to be active if $\appmghinds(t_0,\gamma)$ is invoked. 
Let $A$ be the set of all active timestamps except the initial call of $\appmghinds$. 
For any $t_0\in A$, consider the first invocation of $\appmghinds(t_0,\gamma)$. 
It invokes the subroutine $\randombit$ in \Cref{Line-indset-direct-return-2} and \Cref{Line-indset-satisfy-1}. 
The number of invocations in \Cref{Line-indset-direct-return-2} is $1$, 
and the number of invocations in \Cref{Line-indset-satisfy-1} is bounded by $k\Delta-1$, 
as each vertex is incident to at most $\Delta$ hyperedges and each hyperedge contains $k$ vertices. 
For any call of $\appmghinds(t_0,\gamma)$ beyond the first run, 
the algorithm directly returns a value in \Cref{Line-indset-direct-return-1} and does not invoke $\randombit$. 
Together with the initial call at time $t^*$, we have
\begin{align*}
	|R| \leq (|A|+1)k\Delta. 
\end{align*}

For each $t_0 \in A$, consider the first invocation of $\appmghinds(t_0,\gamma)$ and suppose it is invoked by $\appmghinds(t_1,\gamma)$ when it is processing the hyperedge $e$. 
This actually defines a map $f:A \to V_{\+B}$ where $f(t_0) = \ts{}(e,t_1)$.
By definition, for any $t'$ such that $f(t')=x'$, it must hold that $t' \in x'$ and $|x'| = k$. 
This implies that for each $x' \in V_{\+B}$, there are at most $k$ different $t' \in A$ such that $f(t') = x'$.
This gives
\begin{align}
	|A| \leq k |V_{\+B}|.
\end{align}

Hence $|R| \leq (|A|+1)k\Delta \leq (|V_{\+B}|+1)k^2\Delta$.
\end{proof}

%

\begin{lemma}\label{lemma:indset-connected}
The subgraph of $G_H$ induced by $V_{\+B}$ is connected. 
Furthermore, if $V_{\+B} \neq \emptyset$, then there exists $x \in V_{\+B}$ such that $t^* \in x$. 
\end{lemma}

\begin{proof}
Again, arbitrarily fix the values of all random bits ${(r_t)}_{t \leq 0}$.
Given this, $\appmghinds$ is deterministic. 
We show that the lemma always holds.
We may also assume $V_{\+B} \neq \emptyset$ or otherwise the lemma trivially holds.

Order the set $V_{\+B} = \{x^{(1)},x^{(2)},\ldots,x^{(\ell)}\}$ by the time a vertex joins this set, i.e., triggers a recursion. 
Here $x^{(i)}$ is the $i$-th vertex that triggers a recursion. 
The last index $\ell$ is finite because $\appmghinds(t,\gamma)$ terminates within a finite number of steps. 
We claim that for any $i \in [\ell]$, there exists a path $P=(y_1,y_2,\ldots,y_m)$ in $G_H$ such that 
\begin{itemize}
	\item $t^* \in y_1$ and $x^{(i)} = y_m$;
	\item for all $j\in[m]$, $y_j \in V_{\+B}$.
\end{itemize}
This immediately proves the lemma. 
We prove the claim by induction on index $i$.

\noindent\textbf{Base case. \quad} Consider the initial invocation $\appmghinds(t^*,\gamma)$. 
It must trigger some recursion, or otherwise $V_{\+B} = \emptyset$, violating our assumption. 
Therefore, $t^* \in x^{(1)}$, and we can simply take the path $P=(x^{(1)})$. 

\noindent\textbf{Induction step. \quad} 
Fix an integer $1 <j \leq \ell$.
Suppose the claim holds for all $x^{(i)}$ for $i < j$. 
We prove this for $x^{(j)}$. 
If $x^{(j)}$ joins $V_{\+B}$ whilst $\appmghinds(t^*,\gamma)$ executes, then we can simply take $P=(x^{(j)})$. 
Otherwise, suppose $x^{(j)}$ joins $V_{\+B}$ whilst $\appmghinds(t_0,\gamma)$ executes for some $t_0 < t^*$, and this call is invoked by another $\appmghinds(t_1,\gamma)$ where $t_0<t_1\leq t^*$ when the caller is processing some hyperedge $e \in \+E$ in \Cref{Line-indset-recurse-1}. 
Then there exists $i < j$ such that $x^{(i)} = \ts(e,t_1)$. 
By induction hypothesis, there is a path $P= (y_1,y_2,\ldots,y_{m'})$ for $x^{(i)}$ with $y_{m'} = x^{(i)}$. 
Note that $t_0 \in x^{(i)}$ because $x^{(i)}$ triggers a recursive call of $\appmghinds(t_0,\gamma)$, and $t_0 \in x^{(j)}$ because $v_j$ joins $V_{\+B}$ whilst $\appmghinds(t_0,\gamma)$.
Thus the edge $(x^{(i)},x^{(j)})$ exists in the witness graph, and we can construct the new path as $P' = (P,x^{(j)})$ to meet the conditions. 
\end{proof}

\begin{lemma}\label{lemma:indset-all-1}
For all $x \in V_{\+B}$ and all $t'\in x$, it holds that $r_{t'} = \perp$.
\end{lemma}

\begin{proof}
Fix a $x=\ts{}(e,t_0) \in V_{\+B}$.
As $x$ triggers another subroutine in \Cref{Line-indset-recurse-1}, 
the current hyperedge $e$ must satisfy the condition in \Cref{Line-indset-satisfy-1}. 
This means $r_{t'}=\perp$ for all $t' \in x \setminus \{t_0\}$ (recall the definition in~\eqref{eq:timestamp}). 
Moreover, $r_{t_0} = \perp$, as otherwise $r_{t_0} = 0$ (recall the lower bound distribution $\mu^{\-{LB}}$), 
in which case $\appmghinds(t_0,\gamma)$ must have terminated in \Cref{Line-indset-direct-return-2} before reaching \Cref{Line-indset-recurse-1}.
\end{proof}

To prove \Cref{lemma:indset-efficiency}, we need the notion of $2$-trees~\cite{Alon91}. 
Given a graph $G=(V,E)$, its power graph $G^2=(V,E_2)$ has the same vertex set, while an edge $(u,v)\in E_2$ if and only if $1\leq \dist_G(u,v)\leq 2$. 

\begin{definition}[$2$-tree]\label{definition:2-tree}
Let $G=(V,E)$ be a graph. A set of vertices $T\subseteq V$ is called a \emph{$2$-tree} of $G$, if
\begin{itemize}
    \item  for any $u,v\in T$, $\text{dist}_G(u,v)\geq 2$, and
    \item  $T$ is connected on $G^2$.
\end{itemize}
\end{definition}

Intuitively, a $2$-tree is an independent set that does not spread far away. 
The next lemma bounds the number of $2$-trees of a certain size containing a given vertex.
\begin{lemma}[{\cite[Corollary 5.7]{FGYZ21b}}]\label{lemma:2-tree-number-bound}
Let $G=(V,E)$ be a graph with maximum degree $D$, and $v\in V$ be a vertex. 
The number of 2-trees in $G$ of size $\ell\geq 2$ containing $v$ is at most $\frac{(\mathrm{e}D^2)^{\ell-1}}{2}$.
\end{lemma}

The next lemma shows the existence of a large $2$-tree.
\begin{lemma}[{\cite[Observation 5.5]{FGYZ21b}} and {\cite[Lemma 4.5]{JPV21}}]\label{lemma:big-2-tree}
Let $G = (V, E)$ be a graph of maximum degree $D$, 
$H =(V(H),E(H))$ be a connected finite subgraph of $G$, and $v \in V (H)$ a vertex in $H$. 
Then there exists a $2$-tree $T$ of $H$ containing $v$
such that $\abs{T} = \lfloor |V (H)|/(D + 1) \rfloor$.
\end{lemma}


\begin{proof}[Proof of \Cref{lemma:indset-efficiency}]
By \Cref{lemma:indset-randomness-bound}, it suffices to show $\Pr{\abs{V_{\+B}}\geq 2\Delta k^2\cdot \eta}\leq 2^{-\eta}$ for any integer $\eta\geq 1$, as
\begin{align*}
	\Pr{|R| \geq 3\Delta^2 k^4 \cdot \eta} \leq \Pr{(|V_{\+B}| + 1)\Delta k^2 \geq 3\Delta^2 k^4 \cdot \eta} \leq \Pr{|V_{\+B}| \geq 2\Delta k^2 \cdot \eta} \leq \tp{\frac{1}{2}}^\eta.
\end{align*}

Fix an integer $\eta \geq 1$ and assume $\abs{V_{\+B}}\geq 2\Delta k^2\cdot \eta$.
Recall that $V_{\+B}$ is finite because $\appmghinds$ terminates within a finite number of steps.
By \Cref{lemma:indset-connected}, there exists $y\in V_{\+B}$ such that $t^*\in y$. 
Also by \Cref{lemma:wg-degree-bound} and \Cref{lemma:big-2-tree}, there exists a $2$-tree $T\subseteq V_{\+B}$ of size $i$ such that $y\in T$. 
For $x\in V_{H}$, let $\+{T}^{\eta}_{x}$ denote the set of $2$-trees of $V_{H}$ of size $\eta$ containing $x$. 
Then by \Cref{lemma:indset-connected}, \Cref{lemma:big-2-tree}, and a union bound over all $2$-trees, we have
\begin{equation*}
 \Pr{\abs{V_{\+B}}\geq 2 \Delta k^2\cdot \eta} \leq   \sum\limits_{T\in \+{T}^{\eta}_{x}}\Pr{T\subseteq V_{\+B}}.
\end{equation*}
By \Cref{lemma:indset-all-1}, the event $T\subseteq V_{\+B}$ implies $r_{t}=\perp$ for all timestamps $t$ involved in $T$. 
Also note that, any pair $x,y$ of vertices in $T$ do not share timestamps. 
Thus, we have $\Pr{T\subseteq V_{\+B}}\leq 2^{-k\cdot \abs{T}}$.
Using \Cref{lemma:wg-degree-bound} and \Cref{lemma:2-tree-number-bound}, it holds that 
\begin{align*}
     \sum\limits_{T\in \+{T}^{\eta}_{x}}\Pr{T\subseteq V_{\+B}} \leq  2\Delta k^2\cdot (4\mathrm{e}\Delta^2k^4)^{\eta-1}\cdot \tp{\frac{1}{2}}^{k\eta} \leq  \tp{\frac{4\mathrm{e}\Delta^2 k^4}{2^k} }^\eta\leq  2^{-\eta},
\end{align*}
where the last inequality is due to $2^k \geq 8\mathrm{e}\Delta^2 k^4$.
\end{proof}

%

\subsection{Improved bounds for linear hypergraphs}\label{sec-linear-ind}
We now give a marginal sampler for linear hypergraphs and prove \Cref{lem-sampling-ind-simple}. 
Let $\delta > 0$ be a constant. 
Let $k \geq \frac{25(1+\delta)^2}{\delta^2}$ and $\Delta \geq 2$ be two integers satisfying 
$2^{k} \geq (100 k^3 \Delta)^{1+\delta}$.
Given as inputs a linear $k$-uniform hypergraph $H=(V,\+E)$ with maximum degree $\Delta$ and a parameter $\gamma>0$,
the algorithm is the same as \Cref{Alg:marginal-indset-sampler-app}, 
except that the truncation condition in \Cref{Line-indset-catch} is changed from $|R| \geq 3\Delta^2 k^4\left\lceil\log{\frac{1}{\gamma}}\right\rceil$ to:
\begin{align}\label{eq-ind-simple}
|R| \geq 10^4 \tp{\frac{1+\delta}{\delta}}^2 \Delta^3 k^9  \left \lceil \log{\frac{1}{\gamma}}\right \rceil . 	
\end{align}
Much of the analysis for the general case can be applied to the linear case, with a much refined condition on the maximum degree. 

The running time of the modified algorithm can be bounded in the same way as \Cref{lem-alg-h-2-time}. 
\begin{lemma}\label{lem-time-linear-ind}
The running time of the modified algorithm is $O\left((\frac{1+\delta}{\delta})^2\Delta^4 k^{10} \log \frac{1}{\gamma}\right)$.		
\end{lemma}

Next, we bound the truncation error.

\begin{lemma}\label{lemma:indset-linear-efficiency}
Denote $\eta=\lceil\log{\frac{1}{\gamma}}\rceil$. 
Upon the termination of the modified algorithm, 
the size of $R$ satisfies
\begin{align}
    \Pr{\abs{R}\geq  10^4 \tp{\frac{1+\delta}{\delta}}^2 \Delta^3 k^9 \cdot \eta}\leq 2^{-\eta}.
\end{align}
\end{lemma}

\Cref{lem-sampling-ind-simple} then follows by \Cref{lem-time-linear-ind}, \Cref{lemma:indset-linear-efficiency}, and \Cref{thm:dtv-truncate}.
The rest of this section is dedicated to the proof of \Cref{lemma:indset-linear-efficiency}.

As the condition of maximum degree is much more refined, we need to rely heavily on the property that the original hypergraph $H$ is linear. 
Recall that in the analysis in \Cref{sec:proof-error-ind}, we were actually dealing with the witness graph $G_H$ (\Cref{definition:witness-graph-indset}). 
However, for some pair of adjacent vertices in the witness graph $G_H$, the overlap $\abs{x\cap y}$ is not necessarily $1$ even the underlying hypergraph $H$ is linear.
To see this, recall the partition of neighbourhood of $x\in V_H$ constructed in~\eqref{eq-def-ng}. 
For any neighbour $y$ that corresponds to a different edge in $H$, i.e., $y \in N_{\-{out}}(x)$, it indeed holds that $|x\cap y| = 1$;
however, for those $u \in N_{\-{self}}(v)$, there is no guarantee on the size of the intersection $|x \cap y|$. 
To handle such kind of ``neighbours'' that are actually doppelg\"{a}ngers, 
we introduce the following new notion of the \emph{self-neighbourhood powered witness graph}.
\begin{definition}\label{definition:self-neighbourhood-witness-graph-indset}
Let $G_H=(V_H,E_H)$ be the witness graph as in \Cref{definition:witness-graph-indset}.
The self-neighbourhood powered witness graph $G_H^{\-{self}} = (V_H, E_H^{\-{self}})$ is defined on the same vertex set as $G_H$.
Its edge set $E_H^{\-{self}} = E_H \cup E'$ is given by
\begin{align*}
E' = \{\{x,y\} \mid   (\exists w \in V_H \text{ s.t. } w \in N_{\-{self}}(y) \land w \in N(x)) \lor (\exists w \in V_H \text{ s.t. } w \in N(y) \land w \in N_{\-{self}}(x)) \},	
\end{align*}
where $N_{\-{self}}(x)$ and $N(x)$ are defined as in~\eqref{eq-def-ng}.
\end{definition}

We first bound the maximum degree of $G_H^{\-{self}}$. 
\begin{lemma}\label{lem-max-degree-G-self}
The maximum degree of $G_H^{\-{self}}$ is at most $10k^3\Delta - 1$.
\end{lemma}
\begin{proof}
For any $x \in V_H$, by \Cref{lemma:wg-degree-bound}, $|N_{\-{self}}(x)| \leq 2k$ and $N(x) \leq 2\Delta k^2 - 1$. Hence, the maximum degree of $G_H^{\-{self}}$ is at most 
$
2\Delta k^2 - 1 + 2 \times 2k \times (2\Delta k^2 - 1) \leq  10k^3\Delta - 1.
$
\end{proof}

Recall from the proof of general case that we only need to analyse the size of a connected vertex set $V_{\+B}$ in $G_H$, as \Cref{lemma:indset-randomness-bound} still applies. 
To employ linearity at some point, we instead work on $G_H^{\-{self}}$. 
As there are more edges in $G_H^{\-{self}}$ than in $G_H$, small components are more likely to emerge. 
Therefore, for any connected component $V_{\+B}$ in $G_H$, we can find a subset $V_{\+B}^{\-{lin}} \subseteq V_{\+B}$ that is connected in $G_H^{\-{self}}$, fulfilling linearity. 
This is formally described as the following lemma. 
\begin{lemma}\label{lemma-vb-linear}
Given a $k$-uniform linear hypergraph $H = (V,\+E)$ with the maximum degree $\Delta$, let $G_H=(V_H, E_H)$ be its witness graph.
Fix a vertex $x \in V_H$, and let $V_{\+B} \subseteq V_H$ be a finite subset containing $x$ and connected in $G_H$.
Then, there exists $V_{\+B}^{\-{lin}} \subseteq V_{\+B} $ such that
\begin{enumerate}[label = (L\arabic*)]
	\item $x \in V_{\+B}^{\-{lin}} $ and $|V_{\+B}^{\-{lin}}| \geq \lfloor\frac{|V_{\+B}|}{2k+1}\rfloor$, \label{enum:vblin-size}
	\item the induced subgraph $G_H^{\-{self}}[V_{\+B}^{\-{lin}}]$ is connected, and \label{enum:vblin-connected}
	\item for any two distinct vertices $x_1,x_2 \in V_{\+B}^{\-{lin}}$, it holds that $|x_1 \cap x_2| \leq 1$. \label{enum:vblin-linear}
\end{enumerate}
\end{lemma}

\begin{proof}
We construct one such $V_{\+B}^{\-{lin}}$ explicitly. 
Let $G=G_H[V_{\+B}]$. 
The condition in the lemma says $G$ is connected. 
For any vertex $y \in V_B$, denote by $\Gamma_G(y)$ the neighbourhood of $y$ in $G$, and for any subset $\Lambda \subseteq V_B$, define
\begin{align*}
\Gamma_G(\Lambda) \defeq \{y \in V_{\+B} \setminus \Lambda \mid \exists w \in \Lambda \text{ s.t. } y \in \Gamma_G(w) \}.
\end{align*}
The set $V_{\+B}^{\-{lin}}$ is constructed by the following algorithm. 
\begin{itemize}
  \item Initialise $V_{\+B}^{\-{lin}} \gets \emptyset$, $\Lambda \gets \emptyset$ and $V \gets V_{\+B}$. 
  \item Repeat the following until $V = \emptyset$:
\begin{enumerate}
	\item if $\Lambda = \emptyset$, let $y \gets x$; otherwise, let $y$ be an arbitrary vertex in $\Gamma_G(\Lambda)$;
	\item $V_{\+B}^{\-{lin}} \gets V_{\+B}^{\-{lin}} \cup \{y\}$;
	\item $\Lambda \gets \Lambda \cup \{y\} \cup (V_{\+B} \cap N_{\-{self}}(y) )$.
	\item $V \gets V_{\+B} \setminus \Lambda$.
\end{enumerate}
\end{itemize}

The algorithm above is well defined and terminates in finite steps. 
First, the vertex $y$ in Line (1) always exists. 
To see this, we only need to consider the case $\Lambda \neq \emptyset$. 
Then $V \neq \emptyset$ because the algorithm has not yet terminated. 
Note that $\Lambda \uplus V = V_{\+B}$ upon entering Line (1), 
which implies $\Lambda\neq V_{\+B}$. 
In addition, $G$ is connected, and hence $\Gamma_G(\Lambda) \neq \emptyset$, implying the existence of such $y$.
Secondly, the algorithm will eventually terminate because the set $V$ is initialised to be $V_{\+B}$ which is finite, and the size of $V$ decreases by at least $1$ after each iteration. 

We then verify that the output $V_{\+B}^{\-{lin}}$ satisfies the conditions in the lemma. 

\ref{enum:vblin-size}: 
Apparently $x \in V_{\+B}^{\-{lin}}$.
By \Cref{lemma:wg-degree-bound}, we add at most $2k+1$ vertices into $\Lambda$ in Line (3), thus remove at most $2k+1$ vertices from $|V|$ in each iteration. 
This implies $|V_{\+B}^{\-{lin}}|  \geq \lfloor\frac{|V_{\+B}|}{2k+1}\rfloor$.
%
%
%
%

\ref{enum:vblin-connected}: 
We show by a simple induction that $G_H^{\-{self}}[V_{\+B}^{\-{lin}}]$ is connected during the whole process.  
As a base case, $V_{\+B}^{\-{lin}} = \{x\}$ after the first iteration, and hence the claim holds. 
For the upcoming iterations, we always pick $y \in \Gamma_G(\Lambda)$, thus there exists $w \in \Lambda$ such that $y$ and $w$ are adjacent in $G$, and thus adjacent in $G_H$. 
There are two cases for $w$ joining $\Lambda$ in Line (3) in earlier iterations, that either (a) $w \in V_{\+B}^{\-{lin}}$, or (b) $w \in N_{\-{self}}(w')$ for some $w' \in V_{\+B}^{\-{lin}}$. 
In Case (a), $y$ and $w$ are adjacent in $G_H$, and thus they are adjacent $G_H^{\-{self}}$. 
In Case (b), $y$ and $w'$ are adjacent $G_H^{\-{self}}$.  
The claim then follows by the principle of induction. 

\ref{enum:vblin-linear}:
Fix $x_1,x_2 \in V_{\+B}^{\-{lin}}$. Suppose the algorithm first adds $x_1$ into $V_{\+B}^{\-{lin}}$ and then $x_2$. 
After $x_1$ gets added, the algorithm puts all $(V_{\+B} \cap N_{\-{self}}(x_1) )$ into $\Lambda$, implying $x_2 \notin N_{\-{self}}(x_1)$. 
Hence, there are only two cases for $x_2$: either (a) $x_2 \notin N(x_1)$, and thus $x_1 \cap x_2 = \emptyset$; 
or (b) $x_2 \in N_{\-{out}}(x_1)$, in which case, since the hypergraph $H$ is linear, $|x_1 \cap x_2 | = 1$.
\end{proof}

To make the most of linearity, we use the following $2$-block trees, first introduced in~\cite{FGW22a}, instead of $2$-trees as in the general case.

\begin{definition}[$2$-block-tree]\label{definition:2-block-tree}
Let $\theta,\ell\geq 1$ be two integer parameters, and $G=(V,E)$ be a graph. 
We call a collection of vertex sets $\{C_1,C_2,\cdots,C_\ell\}$ a \emph{2-block-tree} of block size $\theta$ and tree size $\ell$ in $G$, if
\begin{itemize}
  \item for any $1\leq i\leq \ell$, the set $C_i\subset V$ has size $\theta$, and the induced subgraph $G[C_i]$ is connected; 
  \item $\dist_G(C_i,C_j) := \min_{u \in C_i,v \in C_j}\dist_G(u,v)\geq 2$ for any distinct $1\leq i,j\leq \ell$; 
  \item $\bigcup_{i=1}^{\ell}C_i$ is connected on $G^2$, where $u$ and $v$ are adjacent in $G^2$ if and only if $1\leq \dist_G(u,v)\leq 2$.
\end{itemize}
\end{definition}

In fact, a $2$-tree is a $2$-block-tree but with $\theta=1$. 
As we see from last section, the point of using $2$-trees is to secure a bound on the independent vertices in the original hypergraph. 
Therefore, we only look at independent hyperedges, by which we may drop too much. 
The observation is that, if the hypergraph is simple, then a block, as defined above, has a much better lower bound on the number of distinct vertices in the original hypergraph than just dropping dependent hyperedges, so long as $\theta\ll k$. 

We can construct a $2$-block-tree out of a connected induced subgraph. 
\begin{lemma}\label{lemma:existence-2-block-tree}
Let $\theta \geq 1$ be an integer. Let $G = (V , E)$ be a graph with maximum degree $\Delta$. Given a finite subset $C\subseteq V$ such that the subgraph $G[C]$ induced by $C$ is connected, and a vertex $v\in C$, then there exists	a $2$-block-tree $\{C_1,C_2, \dots ,C_{\ell} \}$ in $G$ with block size $\theta$ and tree size $\ell = \lfloor |C|/(\theta^2\Delta^2) \rfloor$ such that $v \in C_1$ and $C_i \subseteq C$ for all $i \in [\ell]$.
\end{lemma}
The above lemma is proved by the technique in \cite[Proposition 16]{FGW22a}, which we shall defer to \Cref{app-A} for completeness. 
We comment that \cite[Proposition 16]{FGW22a} actually proves some additional results, but here in \Cref{lemma:existence-2-block-tree}, we only requires $\ell = \Omega_{\theta,\Delta}(|C|)$ as this is enough for our application.

We also need a bound on the number of 2-block-trees.
\begin{lemma}[{\cite[Lemma 20]{FGW22a}}]\label{lemma:2-block-tree-number-bound}
 Let $\theta \geq 1$ be an integer. Let $G = (V , E)$ be a graph with maximum degree $\Delta$. For any integer
$\ell \geq 1$, any vertex $v\in V$, the number of $2$-block-trees $\{C_1,C_2, \dots ,C_{\ell} \}$ with block size $\theta$ and tree size $\ell$ such
$v\in \cup_{i=1}^{\ell}C_i$ is at most $ (\theta\mathrm{e}^{\theta}\Delta^{\theta+1})^{\ell}$.
\end{lemma} 

Now we prove \Cref{lemma:indset-linear-efficiency}.
\begin{proof}[Proof of \Cref{lemma:indset-linear-efficiency}]
Let $ V_{\+{B}}$ be the set generated by the modified algorithm (see \eqref{eq-ind-simple}) as defined in \eqref{eq:indset-bad-timestamps}. 
Again, $V_{\+{B}} \subseteq V_H$ is a finite subset because the algorithm terminates after a finite number of steps. 
Choose the parameter
\begin{align*}
	\theta \defeq \left \lceil\frac{4(1+\delta)}{\delta} \right \rceil,
\end{align*}
and by this choice, $\Pr{|R| \geq 10^4 (\frac{1+\delta}{\delta})^2 k^9 \Delta^3 \eta} \leq \Pr{|R| \geq 400\theta^2k^9 \Delta^3 \eta}$.
Using \Cref{lemma:indset-randomness-bound}, it holds for any positive integer $\eta$ that
\begin{align*}
	\Pr{|R| \geq 400\theta^2k^9 \Delta^3 \eta} \leq \Pr{ (|V_{\+B}|+1)k^2\Delta \geq 400\theta^2k^9 \Delta^3 \eta } \leq \Pr{|V_{\+B}| \geq 300\theta^2k^7 \Delta^2 \eta  }.
\end{align*}
Hence, it suffices to show
\begin{align*}
	\Pr{|V_{\+B}| \geq 300\theta^2k^7 \Delta^2 \eta  } \leq \tp{\frac{1}{2}}^\eta
\end{align*}
for any integer $\eta\geq 1$. 

Fix an integer $\eta\geq 1$ and assume $|V_{\+B}| \geq 300\theta^2k^7 \Delta^2 \eta$. 
As the set $ V_{\+{B}}$ is non-empty, by \Cref{lemma:indset-connected}, there exists $x \in V_{\+B}$ such that $t^* \in x$. 
Fix such a vertex $x$. 
By \Cref{lemma-vb-linear}, we can find the set $ V_{\+{B}}^{\-{lin}} \subseteq V_{\+B}$ with size $ |V_{\+{B}}^{\-{lin}}| \geq \lfloor \frac{|V_{\+B}|}{2k + 1} \rfloor \geq \lfloor\frac{|V_{\+B}|}{3k} \rfloor \geq \theta^2 (10k^3\Delta)^2 \eta$ containing $x$ and fulfilling the conditions in \Cref{lemma-vb-linear}, 
and it is straightforward to find a subset $U \subseteq V_{\+{B}}^{\-{lin}}$ with size \emph{exactly} $|U| = \theta^2 (10k^3\Delta)^2 \eta$ such that $x \in U$ and the rest of \Cref{lemma-vb-linear} are satisfied by $U$.
By \Cref{lem-max-degree-G-self}, the maximum degree of $G_H^{\-{self}}[U]$ is at most $10k^3\Delta$.
Since $G_H^{\-{self}}[U]$ is a finite connected subgraph in  $G_H^{\-{self}}$, by \Cref{lemma:existence-2-block-tree}, we can find a $2$-block-tree $\{C_1,C_2, \dots ,C_{\eta} \}$ in $G_H^{\-{self}}$ with block size $\theta$ and tree size $\eta$ such that $x \in C_1$ and $C_j \subseteq U \subseteq V_{\+B}$ for all $j\in [\eta]$. 
By \ref{enum:vblin-linear} in \Cref{lemma-vb-linear}, for any distinct $x_1,x_2 \in \cup_{j=1}^\eta C_j$, it holds that $|x_1 \cap x_2| \leq 1$.

The discussion above motivates the following notation. 
For any $x$, let $\+T^{\eta,\theta}_x$ be the set of all 2-block-trees $\{C_1,C_2,\cdots,C_\eta\}$ with block size $\theta$ and tree size $\eta$ in $G_H^{\-{self}}$ such that
\begin{itemize}
	\item $x \in C_1$;
	\item let $ C = \cup_{j=1}^\eta C_j$, then for any $w_1,w_2 \in C$, $|w_1 \cap w_2| \leq 1$.
\end{itemize}
Hence, if $|V_{\+B}|\geq 300\theta^2k^7 \Delta^2 \eta $, then there exists a vertex $x \in V_H$ satisfying $t^* \in v$ together with a 2-block-tree $\{C_1,C_2,\ldots,C_\eta\} \in \+T^{\eta,\theta}_x$ such that $C_j \subseteq V_{\+B}$ for all $j \in [\eta]$.
By a union bound over all possible vertices $x$ and all  2-block-trees in $\+T^{\eta,\theta}_x$, we have 
\begin{align*}
\Pr{|V_{\+B}| \geq 3\theta^2 k\Delta^2 \eta  } \leq \sum_{x \in V_H: t \in x} \sum_{\{C_1,\ldots,C_\eta\} \in \+T^{\eta,\theta}_x} \Pr{\forall j \in [\eta], C_j \subseteq V_{\+B}}.
\end{align*}
Fix a 2-block-tree $\{C_1,\ldots,C_\eta\} \in \+T^{\eta,\theta}_x$. 
By definition, for any $j$ and $\ell$ that $j \neq \ell$, we have $\dist_{G_H^{\-{self}}}(C_j,C_\ell) \geq 2$, and thus for any $x_j \in C_j$ and $x_\ell \in C_\ell$, it holds that $x_j \cap x_\ell = \emptyset$. 
For any $j\in[\eta]$, and any two $w_1, w_2 \in C_j$, it holds that $|w_1 \cap w_2| \leq 1$ by the definition of $\+T^{\eta,\theta}_x$.
This implies that $C_j$ contains at least $\theta (k-\theta)$ distinct timestamps, and hence the 2-block-tree $\{C_1,\cdots,C_\eta\}$ contains at least $\theta (k-\theta) \eta$ distinct timestamps. 
Since the 2-block-tree is a subset of $V_{\+B}$, by~\Cref{lemma:indset-all-1}, every timestamp $t$ in the 2-block-tree must take the corresponding random oracle output $r_t = \perp$.
This gives
\begin{align*}
\Pr{\forall j \in [\eta], C_j \subseteq V_{\+B}}	\leq \tp{\frac{1}{2}}^{\theta(k-\theta)\eta}.
\end{align*}
Next, we count the number of possible 2-block-trees in $\+T^{\eta,\theta}_x$, which can be upper bound by the number of all 2-block-trees $\{C_1,C_2,\ldots,C_\eta\}$ with block size $\theta$ and tree size $\eta$ in $G_H^{\-{self}}$ such that $x \in C_1$. 
By \Cref{lemma:2-block-tree-number-bound} and \Cref{lem-max-degree-G-self}, we have
\begin{align*}
\abs{\+T^{\eta,\theta}_x(U)} \leq (\theta \mathrm{e}^\theta (10 k^3 \Delta)^{\theta+1} )^{\eta}.	
\end{align*}
Hence, we only need to prove that 
\begin{align*}
(\theta \mathrm{e}^\theta (10 k^3 \Delta)^{\theta+1} )^{\eta}	\tp{\frac{1}{2}}^{\theta(k-\theta)\eta} \leq \tp{\frac{1}{2}}^\eta,
\end{align*}
which is equivalent to 
$2^{k-\theta} \geq (2\theta)^{1/\theta} \mathrm{e} (10k^3 \Delta)^{(1+\theta)/\theta}$. 
We derive this as follows. 
Observe the following inequalities
\begin{align}
\delta\geq\frac{2}{\frac{4}{\delta}+3}\geq\frac{2}{\left\lceil\frac{4}{\delta}+3\right\rceil}=\frac{2}{\theta-1}; \label{equ:linear-inequality-1} \\
k\geq\frac{25(1+\delta)^2}{\delta^2}\geq\theta^2\implies\frac{k-\theta}{k}\geq\frac{\theta-1}{\theta}. \label{equ:linear-inequality-2}
\end{align}
Then we have
\begin{align*}
2^{k-\theta}=(2^k)^{\frac{k-\theta}{k}} & \geq\left((100k^3\Delta)^{1+\delta}\right)^{\frac{k-\theta}{k}} \tag{By condition}\\
&\geq\left(10^{1+\frac{2}{\theta-1}}(10k^3\Delta)^{1+\frac{2}{\theta-1}}\right)^{\frac{k-\theta}{k}} \tag{By~\eqref{equ:linear-inequality-1}}\\
&\geq\left(5^{\frac{\theta}{\theta-1}}(10k^3\Delta)^{1+\frac{2}{\theta-1}}\right)^{\frac{k-\theta}{k}}\\
&\geq 5(10k^3\Delta)^{\frac{1+\theta}{\theta}} \tag{By~\eqref{equ:linear-inequality-2}}\\
&\geq (2\theta)^{1/\theta} \mathrm{e}(10k^3\Delta)^{\frac{1+\theta}{\theta}}. \tag{By $\theta\geq 4$}
\end{align*}
This finishes the proof. 
\end{proof}

\section{Hypergraph colouring}\label{section:hyper-colouring}

In this section, we give FPTASes for the number of proper colourings in a hypergraph in the local lemma regime and prove \Cref{thm-h-color} and \Cref{thm-h-color-simple}.  
Given a $k$-uniform hypergraph $H=(V,\+{E})$ with maximum degree $\Delta$, we use $\Omega_H\subseteq [q]^V$ to denote the set of all proper $q$-colourings of $H$, and $Z=\abs{\Omega_H}$ to denote its size. 
Our goal is to approximate $Z$.

We use the edge decomposition scheme in \eqref{eq-edge-decom} to count the number of hypergraph colourings.  Assume $\+{E}=\{e_1,e_2,\dots,e_m\}$. Let $\+{E}_i=\{e_1,e_2,\dots,e_i\}$, $H_i=(V,\+E_i)$. Let $\mu_{i}$ denote the uniform distribution over all proper $q$-colourings of $H_i=(V,\+E_i)$ for each $0\leq i\leq m$, where $\+E_i=\{e_1,e_2,\ldots,e_i\}$. According to \eqref{eq-edge-decom}, approximating $Z_H$ boils down to approximating $\Ex_{z \sim \mu_{i-1,e_i}}[\phi_{e_i}(z)]$  for each $1\leq i\leq m$, where $\phi_{e_i}(z) = 0$ if and only if all $v \in e_i$ take the same value in $z$. Hence, 
\begin{align} \label{eqn:Ex-to-monochromatic}
  \Ex_{z \sim \mu_{i-1,e_i}}[\phi_{e_i}(z)] = \Pr[X \sim \mu_{i-1}]{e_i\text{ is not monochromatic in } X}.
\end{align}


Next, for each $0\leq i\leq m$, it is straightforward to verify the following properties:
\begin{itemize}
    \item $H_i$ is a $k$-uniform hypergraph if $H$ is a $k$-uniform hypergraph;
    \item the maximum degree of $H_i$ is at most the maximum degree of $H$;
    \item $H_i$ is a linear hypergraph if $H$ is a linear hypergraph.
\end{itemize}

When the Lov\'asz local lemma applies,
any edge decomposition scheme is suitably lower bounded, as shown by the next lemma.

\begin{lemma}\label{lemma:h-color-edge-uniform}
  If $\Delta\geq 2$ and $q>(\mathrm{e}\Delta k)^{\frac{1}{k-1}}$, any edge decomposition scheme for $H=(V,\+{E})$ is $\frac{1}{2}$-bounded.
\end{lemma}
\begin{proof}
  For any ordering of the hyperedges of $H$, fix a hypergraph $H_{i-1}=(V,\+E_{i-1})$.
  Note that the maximum degree of $H_{i-1}$ is at most $\Delta$.
Suppose each vertex takes a colour from $[q]$ uniformly and independently.
For each hyperedge $e \in \+E_{i-1}$, define $B_e$ as the event that $e$ is monochromatic.
Let $x(B_e) = \frac{1}{\Delta k}$ for all $e \in \+E_{i-1}$. Then \eqref{llleq} in \Cref{locallemma} is satisfied. Let $A$ denote the event that all vertices in set $e_i$ are monochromatic.  By~\Cref{lemma:HSS}, 
\begin{align*}
\Pr[X \sim \mu_{i-1}]{e_i\text{ is monochromatic in } X} \leq \tp{\frac{1}{q}}^{k-1}\tp{1-\frac{1}{\Delta k}}^{-(\Delta-1)k} \leq \frac{1}{\Delta k} \leq \frac{1}{2},
\end{align*}
and thus the probability that $e_i$ is not monochromatic is at least $\frac{1}{2}$.
\end{proof}

By \Cref{cor-reduction}, it suffices to give the following sampling algorithms for general hypergraphs and linear hypergraphs, respectively. 

\begin{lemma}\label{lem-sampling-color}
Let $\Delta\geq 2$, $k \geq 20$ and $q\geq 64\Delta^{\frac{3}{k-5}}$ be three integers and  $0 < \gamma < 1$ be a real number.
There exists an algorithm such that given any $k$-uniform hypergraph $H = (V,\+E)$ with  the maximum degree at most $\Delta$, any subset of vertices $S \subseteq V$ such that $\abs{S}\leq k$,  outputs a random $Y \in [q]^{S}$ such that $\DTV{Y}{\mu_{S}} \leq \gamma$, where $\mu$ is the uniform distribution over all proper $q$-colourings in $H$. The time cost of this algorithm is $O(\Delta^6k^{11}(\log^2{\frac{1}{\gamma}})\cdot q^{8\Delta^3k^6\log{\frac{1}{\gamma}}})$. 
It draws at most $8\Delta^2 k^5 \lceil\log \frac{1}{\gamma}\rceil+1$ random variables over a size-$(q+1)$ domain.
\end{lemma}

\begin{lemma}\label{lem-sampling-color-simple}
Let $\delta > 0$ and $0 < \gamma < 1$ be two real numbers.
Let $k \geq \frac{50(1+\delta)^2}{\delta^2}$, $\Delta \geq 2$ and $q$ be integers satisfying 
$q\geq 50\Delta^{\frac{2+\delta}{k-3}}$. 
There exists an algorithm such that given any $k$-uniform linear hypergraph $H = (V,\+E)$ with  the maximum degree at most $\Delta$, any subset of vertices $S \subseteq V$ such that $\abs{S}\leq k$, outputs a random $Y \in [q]^{S}$ such that $\DTV{Y}{\mu_{S}} \leq \gamma$, where $\mu$ is the uniform distribution over all proper $q$-colourings in $H$. The time cost of this algorithm is $O\big((\frac{1+\delta}{\delta})^4\Delta^8k^{21}(\log^2{\frac{1}{\gamma}}) q^{6\cdot 10^4(\frac{1+\delta}{\delta})^2\Delta^4k^{11}\log{\frac{1}{\gamma}}}\big)$. 
It draws at most $6 \cdot 10^4 (\frac{1+\delta}{\delta})^2 \Delta^3 k^{10}   \lceil {\log{\frac{1}{\gamma}}}\rceil+1$ random variables over a size-$(q+1)$ domain.
\end{lemma}

\begin{remark}  \label{rem:perfect-marginal-HC}
  Unlike in \Cref{rem:perfect-marginal-HIS},
  algorithms in \Cref{lem-sampling-color} and \Cref{lem-sampling-color-simple} are not direct truncations of perfect marginal samplers.
  In fact, we will apply CTTP to a projected distribution,
  rather than to the original uniform distribution over all proper hypergraph colourings.
  There are some extra steps after truncating CTTP, which will be explained in more details in \Cref{sec:marginal-projected}.
\end{remark}

\Cref{thm-h-color} and \Cref{thm-h-color-simple} are direct consequences of \Cref{lem-sampling-color} and \Cref{lem-sampling-color-simple}, respectively.

\begin{proof}[Proof of \Cref{thm-h-color} and \Cref{thm-h-color-simple}]
By \Cref{lemma:h-color-edge-uniform}, any edge decomposition scheme is $\frac{1}{2}$-bounded under the conditions of \Cref{thm-h-color} and \Cref{thm-h-color-simple}.
%
Plugging $\gamma=\frac{\eps}{20m}$ into \Cref{lem-sampling-color},
where $m\le \Delta n$ is the number of hyperedges, 
we have an algorithm that draws at most $8\Delta^2 k^5 \lceil\log \frac{20m}{\epsilon}\rceil+1$ random variables over a size-$(q+1)$ domain 
with time cost $O(\Delta^6k^{12}\log^2{\frac{20m}{\epsilon}}\cdot q^{8\Delta^3k^6\log{\frac{20m}{\epsilon}}})$ 
for sampling from the conditional distribution within $\frac{\epsilon}{20m}$ total variation distance.
By \Cref{cor-reduction}, we have a deterministic approximate counting algorithm with running time
\begin{align*}
	T = O\tp{\Delta n\cdot \Delta^6k^{11}\log^2{\frac{20\Delta n}{\epsilon}}\cdot q^{8\Delta^3k^6\log{\frac{20\Delta n}{\epsilon}}} \cdot   (q+1)^{  9\Delta^2 k^5 \log \frac{20\Delta n}{\epsilon}} } = \mathrm{poly}(\Delta k) \tp{\frac{\Delta n}{\epsilon}}^{O(\Delta^3 k^6\log{q})}.
\end{align*}
This proves \Cref{thm-h-color}.
\Cref{thm-h-color-simple} is proved the same way but with  \Cref{lem-sampling-color-simple} invoked. 
The running time of the deterministic approximate counting algorithm is 
\begin{align*}
	T &= O\tp{\Delta n \cdot \tp{\tp{\frac{1+\delta}{\delta}}^4\Delta^8 k^{21} \log^2 \frac{20\Delta n}{\epsilon}} \cdot q^{6\cdot 10^4\tp{\frac{1+\delta}{\delta}}^2\Delta^4k^{11}\log{\frac{20\Delta n}{\epsilon}}} \cdot (q+1)^{7 \cdot 10^4 (\frac{1+\delta}{\delta})^2\Delta^3 k^{10}   \log{\frac{20\Delta n}{\epsilon}}}  }\\
	&= \mathrm{poly}\tp{\frac{(1+\delta)\Delta k}{\delta}}\tp{\frac{\Delta n}{\epsilon}}^{O\tp{ \frac{(1+\delta)^2\Delta^4 k^{11}\log{q}}{\delta^2} }}. 
\end{align*}
This proves \Cref{thm-h-color-simple}.
\end{proof}

\subsection{Projection schemes and projected distribution}\label{section:projection-schemes}
Unlike in the case of hypergraph independent sets, the single-site Glauber dynamics for hypergraph colouring is not necessarily irreducible. 
We use the following ``projection scheme" introduced in \cite{FHY21} to resolve this issue.

\begin{definition}[projection scheme for hypergraph colourings~\cite{FHY21}]\label{definition:projection-scheme}
Let $1\leq s \leq q$ be an integer.
A (balanced) \emph{projection scheme}
 $h:[q]\rightarrow [s]$ satisfies for any $i\in [s]$, $\abs{h^{-1}(i)}=\lfloor \frac{q}{s}\rfloor$ or $\abs{h^{-1}(i)}=\lceil \frac{q}{s}\rceil$.
\end{definition}
 We extend $h$ to colourings of $V$ as well.  For any $X\in [q]^V$, we use $Y=h(X)$ to denote the \emph{projected image}:
\[
    \forall v\in V,\quad Y_v=h(X_v),
\]
i.e., the colouring is projected independently for each vertex.  Also for any subset of vertices $\Lambda \subseteq V$ , we will use a similar notation $Y_{\Lambda}=h(X_\Lambda) = (h(X_v))_{v \in \Lambda}$.

Consider the projection scheme $h(\cdot)$ defined in \Cref{definition:projection-scheme}, where the integer parameter $1\leq s\leq q$ will be fixed later. We can naturally associate it with the following projected distribution.
\begin{definition}[projected distribution]\label{definition:projected-distribution}
The projected distribution $\psi: [s]^V\rightarrow [0,1]$ is the distribution $Y=h(X)$ where $X\sim \mu$, where $\mu$ is the uniform distribution over all proper hypergraph $q$-colourings of $H=(V,\+{E})$.  Formally,
\begin{align*}
\forall \sigma\in [s]^V, \quad 	\psi(\sigma)=\sum\limits_{\tau\in h^{-1}(\sigma)}\mu(\tau).
\end{align*}
\end{definition}

For any $\Lambda\subseteq V$ and $\sigma_{\Lambda}\in [s]^{\Lambda}$, we let $\psi^{\sigma_{\Lambda}}$ denote the distribution over $[s]^V$ obtained from $\psi$ conditional on $\sigma_\Lambda$. 
For any $S \subseteq V$, let $\psi^{\sigma_{\Lambda}}_S$ denote the marginal distribution on $S$ projected from $\psi^{\sigma_{\Lambda}}$. If $S = \{v\}$ only contains a single vertex, we write  $\psi^{\sigma_{\Lambda}}_v$ for simplicity.

We also slightly abuse the notation: for any $\Lambda\subseteq V$ and $\sigma_{\Lambda}\in [s]^{\Lambda}$, we let $\mu^{\sigma_{\Lambda}}$ denote the distribution of  $X\sim \mu$ conditional on $h(X(v))=\sigma_{\Lambda}(v)$ for each $v\in \Lambda$. Note that this conditional distribution differs from the usual conditional distribution $\mu^{\sigma_{\Lambda}}$ when $\sigma_{\Lambda}\in [q]^{\Lambda}$ is a partial colouring.
We shall explain the meaning of $\sigma_\Lambda$ when using $\mu^{\sigma_{\Lambda}}$.
Also, for any $S \subseteq V$, let $\mu^{\sigma_{\Lambda}}_S$ denote the marginal distribution on $S$ projected from $\mu^{\sigma_{\Lambda}}$. If $S = \{v\}$ only contains a single vertex, we write  $\mu^{\sigma_{\Lambda}}_v$ for simplicity.

An important  corollary from \Cref{lemma:HSS} is the ``local uniformity" property for the projected distribution $\psi$ in \Cref{definition:projected-distribution}, which states that for any $\Lambda \subseteq V$, any $v\in V \setminus \Lambda $,  conditional on any partial configuration $\sigma\in [s]^{\Lambda}$, $\psi^{\sigma}_{v}$ is close to the uniform distribution over $[s]$.
\begin{lemma}[local uniformity]\label{cor:local-uniformity}
Let $1\leq  s \leq q$.
If $\lfloor q / s \rfloor^k\geq 4\mathrm{e}qs \Delta k$, then for any $v\in V$, any subset $\Lambda\subseteq V\setminus \set{v}$ and partial configuration $\sigma_{\Lambda}\in [s]^{\Lambda}$, it follows that
\begin{align*}
\forall j\in [s],\quad \frac{\abs{h^{-1}(j)}}{q}\left(1-\frac{1}{4s}\right)\leq\psi^{\sigma_\Lambda}_v(j) \leq  \frac{\abs{h^{-1}(j)}}{q}\left(1+\frac{1}{s}\right).	
\end{align*}
\end{lemma}

To prove \Cref{cor:local-uniformity}, we need the following lemma for the general list hypergraph colouring problem.  Let $H =
(V , \+{E})$ be a $k$-uniform hypergraph with maximum degree at most $\Delta$. Let $\{Q_v\}_{v\in V}$ be a set of colour lists.  We
say $X \in \otimes_{v\in V}Q_v$ is a proper list colouring if no hyperedge in $\+{E}$ is monochromatic with respect to $X$.  Let $\mu$ denote the uniform distribution of all proper list hypergraph colourings of $H$. 
\begin{lemma}[{\cite[Lemma 7]{GLLZ19}} and {\cite[Lemma 6]{FGW22a}}]\label{lemma:GLLZ}
Suppose $q_0 \leq |Q_v| \leq q_1$ for any $v \in V$. For any  $r \geq k \geq 2$, if $q_0^k \geq \mathrm{e}q_1r\Delta$, then for any $v \in V$ and any colour $c \in Q_v$,
\[
    \frac{1}{|Q_v|}\left(1-\frac{1}{r}\right)\leq\mu_v(c)\leq \frac{1}{|Q_v|}\left(1+\frac{4}{r}\right).
\]
\end{lemma}
We are now ready to prove \Cref{cor:local-uniformity} using \Cref{lemma:GLLZ}.
\begin{proof}
Note that $\psi^{\sigma_\Lambda}_v$ is the marginal distribution induced by a list hypergraph colouring instance with hypergraph $H$ and colour lists satisfying
\[ \forall u \in V,\quad  Q_u=\begin{cases}h^{-1}(\sigma_{\Lambda}(u)) & u\in \Lambda\\ [q] & u\notin \Lambda\end{cases}.\]
Note that the size of the colour list $Q_u$ of each vertex $u\in V$ satisfies 
\[
  \lfloor q/s\rfloor\leq \abs{h^{-1}(\sigma_{\Lambda}(u))}\leq  \abs{Q_v}\leq q 
\]
By setting $q_0=\lfloor q /s \rfloor$, $q_1=q$ and $r = 4sk$ in \Cref{lemma:GLLZ}, we have
\begin{align*}
\forall j\in [s], \quad \frac{\abs{h^{-1}(j)}}{q}\left(1-\frac{1}{4s}\right)\leq \psi^{\sigma_\Lambda}_v(j)\leq \frac{\abs{h^{-1}(j)}}{q}\left(1+\frac{1}{s}\right). &\qedhere
\end{align*}
\end{proof}

\subsection{Marginal sampler for the projected distribution}\label{sec:marginal-projected}
Here, we give the marginal sampler for the projected distribution of hypergraph colourings (\Cref{Alg:marginal-color-sampler-app}).
We want to apply the truncated resolve algorithm, \Cref{Alg:appresolve}.
One requirement for the correctness of \Cref{Alg:appresolve} is the irreducibility of systematic scan Glauber dynamics,
and we achieve this by running it on the projected distribution $\psi$, with a suitably chosen $s$. 
Furthermore, we will specify suitable marginal lower bounds for $\psi$,
and specify the $\config(t)$ subroutine.


Formally, we use the systematic scan Glauber dynamics $(Y_t)_{-T\leq t\leq 0}$ on the projected distribution $\psi$.  By \Cref{cor:local-uniformity}, for any $\Lambda\subseteq V$, any $\sigma_{\Lambda}\in [s]^{\Lambda}$ and any $v\in V\setminus \Lambda$, it holds that
\[
    \forall j\in [s],\quad \psi^{\sigma_{\Lambda}}_v(j)\geq \rho_j \defeq \frac{\abs{h^{-1}(j)}}{q}\left(1-\frac{1}{4s}\right).
\]
Hence,  the distribution $\psi$ has the $(\rho_1,\rho_2,\cdots,\rho_s)$-marginal lower bound. Note that this immediately shows that the systematic scan Glauber dynamics satisfies \Cref{condition:sufficient-correctness} as $\rho_j>0$ for all $j \in [s]$.
Define the lower bound distribution for (projected) hypergraph colouring instances $\psi^{\-{LB}}$ by
\begin{align}\label{eq-def-basic}
	 \psi^{\-{LB}}(\perp) = 1 - \sum\limits_{1\leq i\leq s}\rho_i=\frac{1}{4s} \qquad\text{and}\qquad \forall j\in [s],\quad \psi^{\-{LB}}(j)=\rho_j.
\end{align}
For any $\Lambda\subseteq V$, any $\sigma_{\Lambda}\in [s]^{\Lambda}$ and any $v\in V\setminus \Lambda$,  define the padding distribution $\psi_v^{\-{pad},\sigma_\Lambda}$ over $[s]$ by
\begin{align*}
\forall j\in [s],\quad\psi_v^{\-{pad},\sigma_{\Lambda}}(j)=4s(\psi^{\sigma_\Lambda}_{v}(j)-\rho_j).
\end{align*}

In order to apply \Cref{Alg:appresolve}, 
we still need to specify the $\config(t)$ subroutine for $\psi$.  
Note that the projected distribution $\psi$ is no longer a Gibbs distribution, but we can still get the conditional independence property by constructing $\sigma_{\Lambda}$ via a breadth-first search (BFS). 
%
%
Given a projected configuration $\sigma_{\Lambda}\in [s]^{\Lambda}$ on a subset of vertices $\Lambda\subseteq V$, we say a hyperedge $e\in \+{E}$ is \emph{satisfied} by $\sigma_{\Lambda}$ 
if there are $u,v\in e\cap \Lambda$ such that $\sigma_{\Lambda}(u)\neq \sigma_{\Lambda}(v)$.
If $e$ is satisfied by $\sigma_{\Lambda}$, for all $X\in [q]^V$ such that $\sigma_{\Lambda}=h(X_{\Lambda})$, the hyperedge $e$ cannot be monochromatic with respect to $X$. 
Let $H^{\sigma_{\Lambda}}=(V,\+{E}')$ be the hypergraph obtained from $H$ after removing all hyperedges in $H$ satisfied by $\sigma_{\Lambda}$. 
Our subroutine $\config(t)$ uses BFS to find the connected component $V'$ in $H^{Y_{t-1}(V \setminus \{v_{i(t)}\} )}$ containing the target variable $v=v_{i(t)}$.
Once the component is found, it also finds $\sigma_\Lambda = Y_{t-1}(\Lambda)$ on some $\Lambda\supset V'$ such that $V'$ is disconnected from the rest in $H^{\sigma_{\Lambda}}$.
We show in \Cref{lemma:config-h-col-cor} that this $\sigma_{\Lambda}$ satisfies \Cref{cond:invariant-boundary}.

The pseudocode for $\config(t)$ is given in \Cref{Alg:config-h-col}.
Recall that it requires two oracles $\+B$ and $\+C$ in \Cref{section-scan-GD}.

%

%

%

%


\begin{algorithm}
\caption{$\config(t)$ for (projected) hypergraph colourings } \label{Alg:config-h-col}
  \SetKwInput{KwPar}{Parameter}
 \KwIn{a hypergraph $H=(V,\+{E})$, a set of colours $[q]$, a projection scheme $h: [q] \to [s]$ that defines the projected distribution $\psi$, and an integer $t\leq 0$\;}
\KwOut{a partial configuration $\sigma_{\Lambda}\in [s]^{\Lambda}$ over some $\Lambda\subseteq V\setminus \{v_{i(t)}\}$;}
$v \gets v_{i(t)}, \Lambda \gets \emptyset, \sigma_{\Lambda}\gets \emptyset,V'\gets \set{v}$\;
\While{$\exists e\in \+{E}$ s.t. $e\cap V'\neq \emptyset$, $e\cap (V\setminus V')\neq \emptyset$ and $e$ is not satisfied by $\sigma_{\Lambda}$\label{Line-color-while}}{
      choose such $e$ with the lowest index\; \label{Line-color-choose}
    \If{$\exists j \in [s]$ s.t.  $\forall u \in e \setminus \{v\}$, $\+B( {\upd_u(t)}) \in \{\perp,j\}$\label{Line-color-satisfy}}{
       \ForAll{$u \in e \setminus \{v\}$}{$\Lambda \gets \Lambda \cup \{u\}$ and $\sigma_{\Lambda}(u)\gets \+{C}(u)$ \; \label{Line-color-assign}}
       $V'\gets V'\cup e$ \;\label{Line-color-update}
    }
    \Else{
    	 \ForAll{$u \in e \setminus \{v\}$}{
    	 \If{$\+{B}(\upd_u(t))\neq \perp$}{
    	 $\Lambda \gets \Lambda \cup \{u\}$ and $\sigma_{\Lambda}(u)\gets \+{B}(\upd_u(t))$\;\label{Line-color-assign-2}}}
    }
} 
\Return $\sigma_{\Lambda}$\; \label{Line-color-final-return}
\end{algorithm}


\begin{lemma}\label{lemma:config-h-col-cor}
$\config(t)$ in \Cref{Alg:config-h-col} satisfies \Cref{cond:invariant-boundary}.
\end{lemma}

\begin{proof}
  The assumptions on the oracles $\+B$ and $\+C$ ensure that
  \begin{align}\label{eq-proof-find-2}
    \forall u \neq v, \qquad \+B(\upd_u(t)) \neq \perp \quad \Longrightarrow \quad  Y_{t-1}(u) = \+B(\upd_u(t)),
  \end{align}  
  and $\+C(u)=Y_{t-1}(u)$.
  We then verify the termination and conditional independence property in \Cref{cond:invariant-boundary}.
%
Note that in the while loop in Lines \ref{Line-color-while}-\ref{Line-color-assign-2}, if the condition in \Cref{Line-color-satisfy} is satisfied, then $e$ is added into $V'$ in \Cref{Line-color-update}; otherwise $e$ is satisfied by $\sigma_{\Lambda}$ after the loop by \Cref{Line-color-assign-2}. 
This shows each $e\in \+{E}$ can be chosen in \Cref{Line-color-choose} at most once, and therefore the while loop eventually terminates.
After the while loop in Lines \ref{Line-color-while}-\ref{Line-color-assign-2} terminates,
the following holds:
\begin{enumerate}
  \item $v\in V'$\label{item:config-h-col-cor-1};
  \item $V' \subseteq \Lambda \cup \{v\}$\label{item:config-h-col-cor-2};
  \item for all $e\in \+{E}$ s.t.~$e\cap V'\neq \emptyset$ and $e\cap (V\setminus V')\neq \emptyset$, $e$ is satisfied by $\sigma_{\Lambda}$.\label{item:config-h-col-cor-3}
\end{enumerate}
Property~\eqref{item:config-h-col-cor-1} holds because $v$ is added to $V'$ in the initialisation, and $V'$ never removes vertices.
Property~\eqref{item:config-h-col-cor-2} holds because if vertices in some $e \in \+E$ are added to $V'$ in \Cref{Line-color-update}, then all vertices in $e \setminus \{v\}$ are added into $\Lambda$ in \Cref{Line-color-assign}.
As the while loop terminates, the condition in \Cref{Line-color-while} no long holds,
which is exactly Property~\eqref{item:config-h-col-cor-3}.

Since $\sigma_{\Lambda}$ is a restriction of $Y_{t-1}(V\setminus\{v\})$ on $\Lambda$,
Property~\eqref{item:config-h-col-cor-1} and Property~\eqref{item:config-h-col-cor-3} imply that conditioned on either $\sigma_{\Lambda}$ or $Y_{t-1}(V\setminus\{v\})$,
the marginal distribution for $v$ is unchanged if we remove all hyperedges crossing both $V'$ and $V\setminus V'$.
Thus, under both conditioning, the marginal distribution for $v$ is the same as its marginal distribution on $H[V']$ given $\sigma_{\Lambda}$ and $Y_{t-1}(V\setminus\{v\})$ restricted to $V'\setminus \{v\}$, respectively.
By Property~\eqref{item:config-h-col-cor-2},
the latter two conditionings restricted to $V'\setminus \{v\}$ are the same.
Thus, $\mu^{\sigma_\Lambda}_v=\mu^{Y_{t-1}(V\setminus\{v\})}_v$.
\end{proof}

Plug $\config(t)$ of \Cref{Alg:config-h-col} into $\appresolve$ of \Cref{Alg:appresolve} and set the threshold $K$ to be $4\Delta^2 k^5 \lceil\log{\frac{1}{\gamma}}\rceil$
This gives \Cref{Alg:marginal-color-sampler-app} for approximately sampling from the marginal distributions on one vertex projected from $\psi$.
%
Recall that as in \Cref{Alg:resolve},
when plugging $\config(t)$, 
we need to replace the oracles $\+B$ and $\+C$ with suitable calls to $\basicsample$ and recursive calls, respectively.

\begin{algorithm}
\caption{$\appmghcol(t, \gamma; M, R)$} \label{Alg:marginal-color-sampler-app}

\SetKwInput{KwData}{Global variables}
 \KwIn{a hypergraph $H = (V,\+E)$, a set of colours $[q]$, a projection scheme $h: [q] \to [s]$ that defines the projected distribution $\psi$, an integer $t \leq 0$ and a real number $0 < \gamma < 1$;}
 \KwData{a map $M: \mathbb{Z}\to [q]\cup\{\perp\}$ and a set $R$;}
\KwOut{a random value in $[s]$;}
\SetKwIF{Try}{Catch}{Exception}{try}{:}{catch}{exception}{}
\uTry{}{
\textbf{if} $M(t) \neq \perp$ \textbf{then} \Return $M(t)$\; \label{Line-color-memoization-1}
\textbf{if} $\basicsample(R, t) \neq \perp$ \textbf{then} $M(t) \gets \basicsample(R, t)$ and \Return $M(t)$\;\label{Line-color-direct-return-1}
$v \gets v_{i(t)}, \Lambda \gets \emptyset, \sigma_{\Lambda}\gets \emptyset,V'\gets \set{v}$\;\label{Line-color-init-1}
\While{$\exists e\in \+{E}$ s.t. $e\cap V'\neq \emptyset$, $e\cap (V\setminus V')\neq \emptyset$ and $e$ is not satisfied by $\sigma_{\Lambda}$\label{Line-color-while-1}}{
       choose such $e$ with the lowest index\; \label{Line-color-choose-1}
    \If{$\exists j \in [s]$ s.t. $\forall u \in e\setminus \set{v} $, $\basicsample({\upd_u(t)};R)\in \{\perp,j\}$\label{Line-color-satisfy-1}}{
        \ForAll{$u \in e\setminus \set{v}$}{$\Lambda \gets \Lambda \cup \{u\}$ and $\sigma_{\Lambda}(u)\gets \appmghcol(\upd_u(t),\gamma;M,R)$ \; \label{Line-color-assign-1}}
       $V'\gets V'\cup e$\;\label{Line-color-update-1}
    }
    \Else{
        $U\gets \{u \in e\setminus \set{v}  \mid  \basicsample({\upd_u(t)};R)\neq \perp\}$\;\label{Line-color-assign-2-1-1.5}
    	 \ForAll{$ u \in U$}{$\Lambda \gets \Lambda \cup \{u\}$ and $\sigma_{\Lambda}(u)\gets  \basicsample({\upd_u(t)};R)$\;\label{Line-color-assign-2-1}}
    }
} 
enumerate all $ X\in \otimes_{u \in \Lambda}h^{-1}(\sigma_u)$ to compute $\mu^{\sigma_{\Lambda}}_{v}$ and then compute $\psi^{\text{pad},\sigma_{\Lambda}}_{v}$ using~\eqref{eq-compute-psi}\;  \label{Line-color-enum}
sample $c\sim \psi^{\text{pad},\sigma_{\Lambda}}_{v}$\; \label{Line-color-sample}
 $M(t) \gets c$ and \Return $M(t)$\; \label{Line-color-final-return-1}
}
\Catch{$|R|\geq 4\Delta^2 k^5 \lceil\log{\frac{1}{\gamma}}\rceil$\label{Line-color-force-return-1}}{ 
\Return $1$\;
}
\end{algorithm}

\Cref{Alg:marginal-color-sampler-app} takes a hypergraph $H=(V,\+E)$, a colour set $[q]$, a projection scheme $h:[q] \to [s]$, an integer $t < 0$, and a parameter $\gamma$ as the input.
To avoid notation cluttering, we drop $(H,[q],h)$ from the input, since these are the same throughout all recursive calls.
We denote the algorithm by $\appmghcol(t,\gamma;M,R)$, where $M$ and $R$ are two global data structures maintained by the algorithm.
When approximately sampling from $\psi_v$, find an integer $t$ such that $-n<t\le 0$ and $v = v_{i(t)}$, 
and then evoke $\appmghcol(t,\gamma;M_0,R_0)$, where $M_0=\perp^{\mathbb{Z}}$ and $R_0=\emptyset$.
Recall that the subroutine \basicsample{} in the algorithm is given in \Cref{Alg:randombit} with the lower bound distribution $\mu^{\-{LB}}$ given by~\eqref{eq-def-basic}.

The correctness of \Cref{Alg:marginal-color-sampler-app} is justified by \Cref{thm:dtv-truncate} and \Cref{lemma:config-h-col-cor}.
In \Cref{Line-color-sample} of \Cref{Alg:marginal-color-sampler-app}, we sample from the padding distribution $\psi^{\text{pad},\sigma_{\Lambda}}_{v}$.
%
%
To implement it, we enumerate all colourings $ X\in \otimes_{u \in \Lambda}h^{-1}(\sigma_u)$ on subset $\Lambda$ to compute $\mu^{\sigma_{\Lambda}}_{v}$ and then compute $\psi^{\text{pad},\sigma_{\Lambda}}_{v}$ by
\begin{align}\label{eq-compute-psi}
\forall j \in [s],\quad	\psi^{\-{pad},\sigma_{\Lambda}}_v(j) = 4s\left(\tp{\sum_{c \in h^{-1}(j) }\mu^{\sigma_{\Lambda}}_v(c)} - \rho_j\right).
\end{align}

\Cref{Alg:marginal-color-sampler-app} is not the algorithm stated in \Cref{lem-sampling-color} and \Cref{lem-sampling-color-simple} as it only samples from the projected value on a single vertex.
In the upcoming \Cref{sec-colouring-set}, we show how to use \Cref{Alg:marginal-color-sampler-app} as a subroutine to build the algorithm for \Cref{lem-sampling-color} and \Cref{lem-sampling-color-simple}.

\subsection{Estimating the marginal distribution on a subset}\label{sec-colouring-set}
We now describe an algorithm such that given a hypergraph $H=(V,\+{E})$ and a subset $S \subseteq V$, it approximately samples from $\mu_S$, where $\mu_S$ is the marginal distribution on $S$ induced by the uniform distribution $\mu$ of all proper $q$-colourings of $H$. 
We then use this algorithm to prove \Cref{lem-sampling-color} and \Cref{lem-sampling-color-simple}. 
Our algorithm builds on \Cref{Alg:marginal-color-sampler-app}, which draws approximate samples from the projected marginal distribution $\psi_v$. 
Suppose we want to sample from the marginal distribution $\mu_S$.  One straightforward approach is that
\begin{itemize}
\item sample $Y\sim \psi$;
\item sample $X\sim \mu$ conditioning on $h(X)=Y$ and output $X_S$.
\end{itemize}
However, the approach above uses too much randomness.
Our sampling algorithm essentially follows it but without sampling the full configurations of $X$ and $Y$.
Note that the process above specifies a joint distribution $(Y,X)\sim \pi$ such that $Y \sim \psi $ and $X \sim \mu^{Y}$.
Consider the following computational problem.
Suppose there is a random sample $Y \sim \psi$ together with an oracle $\+Y$ such that given any $v \in V$, $\+Y$ returns the value of $Y_v$. 
The algorithm needs to draw a random sample $X_S$ using as few oracle queries as possible.
We use the idea in \Cref{Alg:config-h-col}. 
Start from $S$ and use BFS to find a subset $\Lambda \supseteq S$ together with a partial configuration $Y_\Lambda$ such that $X_S$ is independent from $Y_{V \setminus \Lambda} $ in $\pi$ conditional on $Y_{\Lambda}$.
In other words, the partial configuration $Y_\Lambda$ gives enough information to compute $X_S$.
The algorithm for sampling $X_S \sim \mu_S$ is presented in \Cref{Alg:marginal-edge-color-sampler}, where the oracle $\+Y$ mentioned above is implemented by $\appmghcol$ (\Cref{Alg:marginal-color-sampler-app}). 

$\appmghcoledge$ (\Cref{Alg:marginal-edge-color-sampler}) together with subroutines $\appmghcol$ and $\randombit$ maintain two global data structures $M$ and $R$, which are initialised as $M = M_0 = \perp^{\mathbb{Z}}$ and $R = R_0 = \emptyset$ in \Cref{Line-DS-color-init} of $\appmghcoledge$. 
Once  $|R|\geq 4\Delta^2 k^5 \lceil\log{\frac{1}{\gamma}}\rceil$, $\appmghcoledge$ stops immediately and outputs the all-one configuration on $S$.

\begin{algorithm}
\caption{$\appmghcoledge(S,\gamma)$} \label{Alg:marginal-edge-color-sampler}
  \SetKwInput{KwPar}{Parameter}
  \SetKwInput{KwData}{Global variables}
 \KwIn{a hypergraph $H = (V,\+E)$, a set of colours $[q]$, a projection scheme $h: [q] \to [s]$ that defines the projected distribution $\psi$, a subset of $S \subseteq V$ and a real number $0 < \gamma < 1$;}
 \KwData{a map $M: \mathbb{Z}\to [q]\cup\{\perp\}$ and a set $R$;}
\KwOut{a random assignment in $\tau\in [q]^{S}$}
\SetKwIF{Try}{Catch}{Exception}{try}{:}{catch}{exception}{}
$M \gets \perp^{\mathbb{Z}}$ and $R \gets \emptyset$\;\label{Line-DS-color-init}
\uTry{}{
$\sigma\gets \emptyset,V'\gets S$\;\label{Line-color-edge-init}
\ForAll{$u\in S$}{
$\sigma(u)\gets \appmghcol{}(\upd_u(0),\gamma;M,R)$\;\label{Line-color-edge-S}}
\While{$\exists e\in \+{E}$ s.t. $e\cap V'\neq \emptyset$, $e\cap (V\setminus V')\neq \emptyset$ and $e$ is not satisfied by $\sigma$\label{Line-color-edge-while}}{
      choose such $e$ with the lowest index\; \label{Line-color-edge-choose}
    \If{$\exists j \in [s]$ s.t. $\forall u \in e$, $\basicsample({\upd_u(0)};R) \in \{\perp,j\}$\label{Line-color-edge-satisfy}}{
      \ForAll{$ u \in e$}
      {$\sigma(u)\gets \appmghcol(\upd_u(0),\gamma; M,R)$ \; \label{Line-color-edge-assign}}
       $V'\gets V'\cup e$\;\label{Line-color-edge-update}
    }
    \Else{
    	$U\gets \{u \in e  \mid  \basicsample({\upd_u(0)};R)\neq \perp\}$\;\label{Line-color-edge-assign-1.5}
    	 \ForAll{$ u \in U$}
         {$\sigma(u)\gets  \basicsample({\upd_u(0)};R)$\;\label{Line-color-edge-assign-2}}
    }
} 
enumerate all colourings $ X\in \otimes_{u \in V'}h^{-1}(\sigma_u)$ on $V'$ to compute the marginal distribution of $S$ on the sub-hypergraph $H[V']$,
which is equivalent to $\mu^{\sigma}_{S}$\;  \label{Line-color-edge-enum}
Sample $\tau\sim \mu^{\sigma}_{S}$\; \label{Line-color-edge-tau-sample}
\Return $\tau$\; \label{Line-color-edge-final-return}
}
\Catch{$|R|\geq 4\Delta^2 k^5 \lceil\log{\frac{1}{\gamma}}\rceil$\label{Line-color-edge-catch} }{ 
\Return $1^S$\;\label{Line-color-edge-force-return-1} 
}
\end{algorithm}


Since we truncate the algorithm when $|R| \geq 4\Delta^2 k^5 \lceil\log{\frac{1}{\gamma}}\rceil$,  we have the following lemma that bounds the running time of  \Cref{Alg:marginal-edge-color-sampler}. 
\begin{lemma}\label{lem-alg-edge-col-time}
The running time of \Cref{Alg:marginal-edge-color-sampler} is $O(\Delta^6k^{11}\log^2{\frac{1}{\gamma}}\cdot q^{8\Delta^3k^6\log{\frac{1}{\gamma}}})$.
\end{lemma}
To prove \Cref{lem-alg-edge-col-time}, we need the following lemma.
\begin{lemma}\label{lem-color-lambda-bound}
  For any hypergraph $k$-uniform $H=(V,\+{E})$ with the maximum degree at most $\Delta$, any integer $t < 0$ and  any $0<\gamma\leq \frac{1}{2}$, any $M:\mathbb{Z} \to [q] \cup \{\perp\}$ and any set $R$, let $\Lambda$ be either
  \begin{itemize}
      \item the set of assigned variables in $\sigma$ in  \Cref{Line-color-enum} of the execution $\appmghcol{}(t,\gamma;M,R)$;
      \item the set of assigned variables in $\sigma$ in \Cref{Line-color-edge-enum} of the execution of $\appmghcoledge{}(S,\gamma)$,
  \end{itemize}  
  it holds that
\[
    \abs{\Lambda}\leq k\Delta\cdot \abs{V'}\leq 4\Delta^3 k^6\left\lceil\log{\frac{1}{\gamma}}\right\rceil.
\]
\end{lemma}
\begin{proof}
We only show the case for \Cref{Alg:marginal-color-sampler-app} and the case for \Cref{Alg:marginal-edge-color-sampler} holds through a similar argument. The size of the set $V'$ increases from during the while loop.  By the condition in \Cref{Line-color-while-1}, we have for all $u\in \Lambda$, there exists $e\in \+{E}$ and $w\in V'$ such that $u,w\in e$. This shows that
\[
    \abs{\Lambda}\leq k\Delta\cdot \abs{V'}.
\] 
Moreover, note that by \Cref{Line-color-direct-return-1} and \Cref{Line-color-satisfy-1} we have $\abs{R}\geq \abs{V'}$.
Therefore by the truncation condition at \Cref{Line-color-force-return-1} we have  
\[  
    \abs{V'}\leq 4\Delta^2 k^5\left\lceil\log{\frac{1}{\gamma}}\right\rceil.
\]
Combining the two inequalities proves the lemma.
\end{proof}

\begin{proof}[Proof of \Cref{lem-alg-edge-col-time}]
  Consider the whole process of $\appmghcoledge(S,\gamma)$. Let $A$ be the set of all $t_0$ such that $\appmghcol(t_0,\gamma;M,R)$ (for some $M$ and $R$) is executed at least once. For any $t_0 \in A$, one of the following two events must happen (1) the whole algorithm terminates in \Cref{Line-color-force-return-1}; or (2) $(t_0,r_{t_0}) \in R_F$, where $R_F$ is the final $R$ when whole algorithm terminates. Moreover, there is at most one $t^* \in A$ making the first event happen.  Hence $|A| \leq |R_F| + 1=O\left(\Delta^2 k^5\log{\frac{1}{\gamma}}\right)$. 
For any $t_0 \in A$, once $\appmghcol(t_0,\gamma;M,R)$ returns, $M(t_0)\neq \perp$.
Any further calls to $\appmghcol(t_0,\gamma;M,R)$ would not execute beyond \Cref{Line-color-memoization-1} of \Cref{Alg:marginal-color-sampler-app} and would have cost $O(1)$ time.
Consider the recursion tree of the whole execution.
All subsequent recursive calls are leaf nodes.
We attribute the time cost of the subsequent recursive calls to its parent,
and this way the total time cost becomes the total cost of all first calls.
Assuming the cost of each recursive call is $O(1)$,
let $T^*$ be the upper bound of the running time of both
\begin{itemize}
  \item $\appmghcol(t,\gamma;M,R)$ of \Cref{Alg:marginal-color-sampler-app}, and
  \item $\appmghcoledge(S,\gamma)$ of \Cref{Alg:marginal-edge-color-sampler}.
\end{itemize}
Then 
the total running time of can be bounded by $(|A|+1)T^*$. 
We then analyze the time cost of the first item for \Cref{Alg:marginal-color-sampler-app}.

Note that the cost of the while loop in Lines \ref{Line-color-while-1}-\ref{Line-color-assign-2-1} of \Cref{Alg:marginal-color-sampler-app}
is at most a constant multiple of the number of executions of \Cref{Line-color-assign-1} and \Cref{Line-color-assign-2-1}.
If the condition in \Cref{Line-color-satisfy-1} is satisfied, then $e$ is added into $V'$ in \Cref{Line-color-update-1}; otherwise $e$ is satisfied by $\sigma$ after the loop by \Cref{Line-color-assign-2-1}. This shows that each $e\in \+{E}$ can be chosen in \Cref{Line-color-choose-1} at most once. 
Each time \Cref{Line-color-assign-1} or \Cref{Line-color-assign-2-1} executes,
some $u$ is added to $\Lambda$ due to an hyperedge $e\ni u$.
Thus it happens at most $\Delta$ times for each $u\in\Lambda$.
By \Cref{lem-color-lambda-bound}, $\abs{\Lambda}\le O\left(\Delta^3 k^6\log{\frac{1}{\gamma}}\right)$.
Thus the total cost of the while loop is at most $O\left(\Delta^4 k^6\log{\frac{1}{\gamma}}\right)$. 
Moreover, in \Cref{Line-color-enum},
we enumerate $q^{\abs{\Lambda}+1}$ possible configurations,
each of which takes $O(\abs{\Lambda}\Delta)$ time to check.
The total time cost of \Cref{Line-color-enum} is
\begin{align*}
	O(\abs{\Lambda}\Delta q^{\abs{\Lambda}+1})=O\tp{\Delta^4k^6\log{\frac{1}{\gamma}}\cdot q^{8\Delta^3k^6\log{\frac{1}{\gamma}}}}.
\end{align*}
The same time cost bound holds for \Cref{Alg:marginal-edge-color-sampler} by the same argument.
Therefore it suffices to take
\begin{align*}
  T^* = O\tp{\Delta^4k^6\log{\frac{1}{\gamma}}\cdot q^{8\Delta^3k^6\log{\frac{1}{\gamma}}}}.
\end{align*}
Since $|A| = O(\Delta^2 k^5 \log \frac{1}{\gamma})$, the total running time of $\appmghcoledge(S,\gamma)$ is at most $O(\Delta^6k^{11}\log^2{\frac{1}{\gamma}}\cdot q^{8\Delta^3k^6\log{\frac{1}{\gamma}}})$.
The lemma follows.
\end{proof}

The next lemma shows the correctness of \Cref{Alg:marginal-edge-color-sampler}.  Let $\+E_{\-{trun}} = \+E_{\-{trun}}(4\Delta^2 k^5 \lceil\log{\frac{1}{\gamma}}\rceil)$ be the event that the truncation in \Cref{Line-color-edge-force-return-1} of \Cref{Alg:marginal-edge-color-sampler} occurs.  
Similar to \Cref{thm:dtv-truncate},
we bound the total variation distance between the output of \Cref{Alg:marginal-color-sampler-app} and the target distribution $\mu_S$ by $\Pr{\+E_{\-{trun}}}$.
\begin{lemma}\label{lemma:edge-col-correct}
Let $H=(V,\+E)$ be a $k$-uniform hypergraph with the maximum degree at most $\Delta$, $[q]$ be a set of colours, and $h: [q] \to [s]$ be a projection scheme that defines the projected distribution $\psi$. 
Let $S\subseteq V$ be a subset of vertices and $\gamma > 0$ be a real number.  
Let $X_S$ be the output of $\appmghcoledge(S,\gamma)$.
If $\lfloor q / s \rfloor^k\geq 4\mathrm{e}qs \Delta k$, it holds that $\DTV{X_S}{\mu_S} \leq \Pr{\+E_{\-{trun}}}$, where $\+E_{\-{trun}} = \+E_{\-{trun}}(4\Delta^2 k^5 \lceil\log{\frac{1}{\gamma}}\rceil)$ and $\mu$ is the uniform distribution over all proper $q$-colourings of $H$.
\end{lemma}
The proof of \Cref{lemma:edge-col-correct} is similar to the proof of \Cref{thm:dtv-truncate}. 
The main difference is that the equivalent to \Cref{Alg:appresolve}, $\appmghcol$, is called multiple times in \Cref{Alg:marginal-edge-color-sampler}.  
Intuitively, we use $\appmghcol$ to access a full configuration $Y \in [s]^V$, where $Y$ approximately follows the projected distribution $\psi$.  We may query the values of $Y_u$ for different vertices $u$.  
The answers returned by $\appmghcol$ for different $u$ should be consistent with each other.  
Since this proof is somewhat repetitive to that of \Cref{thm:dtv-truncate},
we delay the proof of \Cref{lemma:edge-col-correct} to \Cref{sec-proof-colour-correct}.

We also have the following lemma that bounds the probability that the size of $R$ becomes too large.
\begin{lemma}\label{lemma:color-efficiency}
Let $H=(V,\+E)$ be a $k$-uniform hypergraph with the maximum degree at most $\Delta$, $[q]$ be a set of colours, and $h: [q] \to [s]$ be a projection scheme that defines the projected distribution $\psi$.  
Let $S\subseteq V$ be a subset of vertices and $\gamma > 0$ be a real number.  
If $k\geq 20$, $q\geq 4s$, $\lfloor q / s \rfloor^k\geq 4\mathrm{e}qs \Delta k$, $s \geq 6\Delta^{\frac{2}{k-2}}$ and $\abs{S}\leq k$, 
then upon the termination of $\appmghcoledge{}(S,\gamma)$, it holds that 
\begin{align}
    \Pr{\abs{R}\geq 4\Delta^2 k^5 \cdot \eta}\leq \cdot 2^{-\eta}
\end{align}
where $\eta=\lceil\log\frac{1}{\gamma}\rceil$. In particular, this shows that 
\[
    \Pr{\+E_{\-{trun}}}\leq \gamma.
\]
\end{lemma}

\Cref{lemma:color-efficiency} is proved in \Cref{sec:proof-error-color}. We are now ready to prove \Cref{lem-sampling-color}. 
\begin{proof}[Proof of \Cref{lem-sampling-color}]
  Let $s=\left\lfloor q^{\frac{2}{3}}\right\rfloor$.
  As $q\geq 64\Delta^{\frac{3}{k-5}}$, we have
  \begin{align*}
    q&\geq 4s, & s&\geq 16\Delta^{\frac{2}{k-5}}-1\geq 6\Delta^{\frac{2}{k-2}},
  \end{align*}
  and by $k\geq 20$,
  \[
    \left\lfloor\frac{q}{s}\right\rfloor\geq 4\Delta^{\frac{1}{k-5}} - 1\geq 2(k\Delta)^{\frac{1}{k-5}}.
  \]
  Therefore
  \[
    \left\lfloor\frac{q}{s}\right\rfloor^{k}\geq 2^{k-5}(k \Delta) \left\lfloor \frac{q}{s} \right\rfloor^5 \geq 2^{11} \cdot 4\mathrm{e}k\Delta\cdot \left\lfloor\frac{q}{s}\right\rfloor^5\geq 4\mathrm{e}qs \Delta k,
  \]
  where the second inequality is by $k\geq 20$, and the last inequality is by $8\left\lfloor\frac{q}{s}\right\rfloor^2\geq s$ and $8\left\lfloor\frac{q}{s}\right\rfloor^3\geq q$ from $s=\left\lfloor q^{\frac{2}{3}}\right\rfloor$.
  Therefore the conditions in \Cref{lemma:color-efficiency} are met.
We use $\appmghcoledge(S,\gamma)$ to sample from the distribution $\mu_S$. By \Cref{lem-alg-edge-col-time}, the running time is $O(\Delta^6k^{11}\log^2{\frac{1}{\gamma}}\cdot q^{8\Delta^3k^6\log{\frac{1}{\gamma}}})$. 
It is straightforward to verify that \Cref{Alg:marginal-edge-color-sampler} uses at most $8\Delta^2 k^5 \left\lceil\log{\frac{1}{\gamma}}\right\rceil+1$ random variables with domain size at most $q+1$.
By combining \Cref{lemma:edge-col-correct} and \Cref{lemma:color-efficiency}, the output is a random sample that is $\gamma$-close to $\mu_v$ in the total variation distance.
This proves the lemma.
\end{proof}

\subsection{Truncation error of hypergraph colouring}\label{sec:proof-error-color}
We now prove \Cref{lemma:color-efficiency}.

We need some notations from \Cref{section:hyper-indset}.
Recall $\ts{}(e,t)$ in \eqref{eq:timestamp} for each $e\in \+{E}$ and $t\leq 0$. 
Also recall the witness graph $G_H=(V_H,E_H)$ in \Cref{definition:witness-graph-indset}. 
Note that the maximum degree of $G_H$ is bounded by $2\Delta k^2-2$ by \Cref{lemma:wg-degree-bound}. 

Fix a $k$-uniform hypergraph $H=(V,\+{E})$ with the maximum degree at most $\Delta$ together with a set $[q]$ of colours, a projection scheme with parameter $s$, a subset of vertices $S \subseteq V$, and a real number $\gamma > 0$ satisfying the condition in \Cref{lemma:color-efficiency}. 
We analyse the algorithm $\appmghcoledge{}(S,\gamma)$.

As we want to resolve the marginal distribution for a set $S$ of vertices,
we modify the definition of $G_H$ to obtain another graph $G_H^{S}=(V_H^S, E_H^S)$ by adding a new vertex representing the last update times of all $v\in S$ before $0$ and corresponding edges.

\begin{definition}[modified witness graph]\label{definition:witness-graph-color}
  For a hypergraph $H=(V,\+E)$ and a subset of vertices $S\subseteq V$,
  if $S \in \+E$, define the modified witness graph $G^S_H= G_H$;
  otherwise, let $\ts(S,0)=\set{\upd_u(0)\mid u\in S}$, and define the modified witness graph $G^S_H=(V^S_H,E^S_H)$ with respect to $S$ as follows:
\begin{itemize}
    \item $V^S_H=V_H\cup \{\ts(S,0)\}$,
    \item $E^S_H=E_H\cup \{(x,\ts(S,0))\mid x\in V_H\land x\cap \ts(S,0)\neq \emptyset\}$,
\end{itemize}
where $G_H=(V_H,E_H)$ is defined in  \Cref{definition:witness-graph-indset}.
\end{definition}

Note that in the modified graph $G_H^{S}$, there exists a vertex  $\ts(S,0)$ for the subset $S$.
Let $e_H(\ts(S,0))=S$.
We can then define $N_{\-{self}}(\ts(S,0)),N_{\-{out}}(\ts(S,0))$ and $N(\ts(S,0))$ as in \eqref{eq-def-ng}.
That is, $N_{\-{self}}(\ts(S,0)) = \emptyset$ and $N_{\-{out}}(\ts(S,0)) = N(\ts(S,0)) = \{ w \in V_H^S \mid (w,\ts(S,0))\in E_H^S\}$;
And for all $w \in N(\ts(S,0))$, it holds that $\ts(S,0) \in N_{\-{out}}(w)$.

The next result follows from the proof of \Cref{lemma:wg-degree-bound}. 
\begin{corollary} \label{lemma:modified-wg-degree-bound}
If $\abs{S}\leq k$, then the maximum degree of $G_H^S$ is bounded by $2\Delta k^2 -1$.
\end{corollary}
\begin{proof}
Note that compared to $G_H$, the degree of all vertices $x\in V_H=V^S_H\setminus \{\ts(S,0)\}$ increases by at most one and is therefore bounded by $2\Delta k^2 -1$ by \Cref{lemma:wg-degree-bound}. 

It then remains to bound the degree of $\ts(S,0)$. We have two cases:
\begin{itemize}
    \item If $S=e$ for some $e\in \+{E}$, then following the proof of \Cref{lemma:wg-degree-bound}, $|N_{\-{out}}(x)|=|N_{\-{out}}(\ts(e,0))|\leq (2k-1)(\Delta-1)k$ and $|N_{\-{self}}(x)|\leq 2k-1$. This together shows $|N(x)|\leq 2k^2(\Delta-1)+2k\leq 2k^2\Delta-1$.
    \item Otherwise, $|N_{\-{self}}(x)|=0$. Following the proof of \Cref{lemma:wg-degree-bound}, we have $|N_{\-{out}}(x)|\leq k\Delta\cdot (2k-1)\leq 2k^2\Delta-1$. This shows $|N(x)|\leq  2k^2\Delta-1$.
\end{itemize}
Combining the above proves the lemma.
\end{proof}

%
%
%
%
%
%
Fix an integer $t_0 \leq 0$.
$\appmghcol(t_0,\gamma;M,R)$ can only be evoked by:
\begin{enumerate}
  \item $\appmghcoledge(S,\gamma)$ in \Cref{Line-color-edge-S} when it processes $S$, and in this case, it holds that $t_0= \upd_u(0)$ for some $u \in S$; or \label{item:trigger1}
  \item $\appmghcoledge(S,\gamma)$ in \Cref{Line-color-edge-assign} when it processes $e \in \+E$, and in this case, it holds that  $t_0= \upd_u(0)$ for some $u \in e$; or \label{item:trigger2}
  \item $\appmghcol(t_1,\gamma;M,R)$ in \Cref{Line-color-assign-1} when it processes $e \in \+E$, and in this case, it holds that $t_0 < t_1 \leq 0$ and  $t_0= \upd_u(t_1)$ for some $u \in e$. \label{item:trigger3}
\end{enumerate}
In the first case, we say an instance of $\appmghcol(t_0,\gamma;M,R)$ is triggered by $\ts{}(S,0) \in V_H^S$.
In the second case, we say such an instance is triggered by $\ts{}(e,0) \in V_H^S$. 
In the last case, we say such an instance  is triggered by $\ts{}(e,t_1) \in V_H^S$.
Conversely, for any $x \in V_H^S$ in the witness graph, we say $x$ triggers some instance of \appmghcol{} if any of the above happens with the trigger being $x$ during the whole process.

Define the following random subset of vertices in $G^S_H$: 
\begin{align}\label{eq:color-bad-timestamps}
    V^{\-{col}}_{\+{B}}=\{x\in V_H^S\mid \text{ $x$  triggers some instance of \appmghcol }\},
\end{align}
where the randomness of $ V^{\-{col}}_{\+{B}}$ comes from the random variables $\{r_t\}_{t\leq 0}$ and the samples drawn from padding distribution in \Cref{Line-color-sample} of $\appmghcol$.
Since \Cref{Line-color-edge-S} always happens, it holds that
\begin{align*}
	\ts(S,0) \in  V^{\-{col}}_{\+{B}}.
\end{align*}
After \Cref{Alg:marginal-edge-color-sampler} terminates, it generates sets $V^{\-{col}}_{\+{B}}$  and $R$.
Similar to \Cref{lemma:indset-randomness-bound}, the following lemma shows that the size of $R$ can be upper bounded in terms of $|V^{\-{col}}_{\+B}|$.
\begin{lemma}\label{lemma:color-randomness-bound}
If $\abs{S}\leq k$, then $\abs{R}\leq 2\Delta k^3|V^{\-{col}}_{\+B}|$. 
\end{lemma}
\begin{proof}
We first arbitrarily fix values of all random variables ${(r_t)}_{t \leq 0}$ and the outcomes of samples drawn from the padding distribution in \Cref{Line-color-sample} in $\appmghcol$.
Then, the whole algorithm is deterministic.
We show that the lemma always holds.

We claim that for each $(t,r_t)\in R$, there exist $x,y\in V^{S}_H$ (possibly $x=y$) such that $t\in x$, $y\in V^{\-{col}}_{\+B}$
and $x\cap y\neq\emptyset$.
Namely either $x=y$ or $(x,y)\in E_H^S$.
For any $(t,r_t)\in R$, we define a map $f(t,r_t) = x \in  V^{\-{col}}_{\+B}$,
where $x$ is chosen arbitrarily from those satisfying the claim,
such as the lexicographically first. 
Since $|S| \leq k$,  by \Cref{lemma:modified-wg-degree-bound}, the maximum degree of $G^S_H$ is at most $2\Delta k^2 - 1$.
As $\abs{x}\le k$ there are at most $2\Delta k^3$ different $(t,r_t)\in R$ such that $f(t,r_t)$ maps to the same $y$. Hence,
\begin{align*}
	\abs{R}\leq |V^{\-{col}}_{\+B}|2\Delta k^3.
\end{align*}


To prove the claim, we first show the following result for $\appmghcoledge$ (\Cref{Alg:marginal-edge-color-sampler})  and $\appmghcol$ (\Cref{Alg:marginal-color-sampler-app}):
\begin{align}\label{eq-induction-app}
  \forall u \in V', \exists\, x \in V^{\-{col}}_{\+B} \text{~s.t.~} \upd_u(t_0) \in x,
\end{align}
where $t_0 = 0$ for $\appmghcoledge(S,\gamma)$ and $t_0 = t$ for $\appmghcol(t, \gamma;M,R)$.
\begin{itemize}
  \item For $\appmghcoledge(S,\gamma)$, If $u \in S$, we can take $x = \ts(S,0)$. 
    If $u \in V' \setminus S$, then $u$ must be added into $V'$ when the algorithm is processing $e \in \+E$ in \Cref{Line-color-edge-update}. 
    Then $\ts(e,0) \in V^{\-{col}}_{\+B}$ by \Cref{Line-color-edge-assign} and $\upd_u(0) \in \ts(e,0)$. 
  \item For $\appmghcol( t, \gamma;M,R)$. Note that $\appmghcol( t, \gamma;M,R)$ itself must be triggered by some $x^* \in V^{\-{col}}_{\+B}$ and $t = \upd_{v_{i(t)}}(t) \in x^*$.
    If $u = v_{i(t)}$,~\eqref{eq-induction-app} holds for $x = x^*$.
    If $u \in V' \setminus \{v_{i(t)}\}$, then $u$ must be added into $V'$ when the algorithm is processing $e \in \+E$ in \Cref{Line-color-update-1}.  
    Then $\ts(e,t) \in V^{\-{col}}_{\+B}$ by \Cref{Line-color-assign-1} and $\upd_u(t) \in \ts(e,t)$. 
\end{itemize}

Now we turn to show $x,y\in V^{S}_H$ in the claim exist.
A pair $(t,r_t)$ can be added into $R$ only through the call of $\randombit(R,t)$.
There are 3 cases.
\begin{itemize}
  \item It happens in \Cref{Line-color-edge-satisfy} of \Cref{Alg:marginal-edge-color-sampler} when some edge $e\in \+{E}$ is processed and for some $u\in e$, $t=\upd_u(0)$. It holds that $e \cap V' \neq \emptyset$. 
    Let $x=\ts(e,0) $ so that $t \in x$. 
    Fix an arbitrary $w \in e \cap V'$. By \eqref{eq-induction-app}, there is $y \in V^{\-{col}}_{\+B}$ such that $\upd_w(0) \in y$. Since $\upd_w(0) \in x$, $x \cap y \neq \emptyset$.
  \item It happens in \Cref{Line-color-direct-return-1} of \Cref{Alg:marginal-color-sampler-app} in  $\appmghcol(t,\gamma;M,R)$.
    Note that $v_{i(t)} \in V'$ in  $\appmghcol(t,\gamma;M,R)$. Hence, by \eqref{eq-induction-app}, there is $y \in V^{\-{col}}_{\+B}$ such that $ t = \upd_{v_{i(t)}}(t) \in y$. We can take $x = y$.
  \item It happens in \Cref{Line-color-satisfy-1} of \Cref{Alg:marginal-color-sampler-app} when some edge $e\in \+{E}$ is processed in an instance of $\appmghcol{}(t_0,\gamma;M,R)$ for some $t_0 \geq t$. 
    It holds that $e \cap V' \neq \emptyset$. Let $x=\ts(e,t_0) $ so that $t \in x$. Fix an arbitrary $w \in e \cap V'$. 
    By \eqref{eq-induction-app}, there is $y \in V^{\-{col}}_{\+B}$ such that $\upd_w(t_0) \in y$.  Since $\upd_w(t_0) \in x$, $x \cap y \neq \emptyset$.
\end{itemize}
Finally, we remark that we do not need to consider \Cref{Line-color-edge-assign-1.5} / \Cref{Line-color-edge-assign-2} in \Cref{Alg:marginal-edge-color-sampler} nor \Cref{Line-color-assign-2-1-1.5} / \Cref{Line-color-assign-2-1} in \Cref{Alg:marginal-color-sampler-app}, because if $\basicsample{}(t;R)$ is evoked in these lines, it must have been evoked in \Cref{Line-color-edge-satisfy}  of \Cref{Alg:marginal-edge-color-sampler} or \Cref{Line-color-satisfy-1} of \Cref{Alg:marginal-color-sampler-app} already.
This finishes the proof.
\end{proof}

\Cref{lemma:color-randomness-bound} shows that to upper bound the probability that $|R|$ is large, it suffices to upper bound the probability that $|V^{\-{col}}_{\+B}|$ is large.
The following lemma says that $V^{\-{col}}_{\+{B}}$ is connected on $G^S_H$, and the proof is also similar to the proof of \Cref{lemma:indset-connected}.
\begin{lemma}\label{lemma:color-connected}
$\ts(S,0) \subseteq V^{\-{col}}_{\+B}$ and 
the subgraph in $G^S_H$ induced by $V^{\-{col}}_{\+B}$ is connected.
\end{lemma}
\begin{proof}
Fix arbitrary random choices of all random variables $(r_t)_{t \leq 0}$ and outcomes of samples drawn from the padding distribution in \Cref{Line-color-sample} in $\appmghcol$. 
Then the algorithm is deterministic. We prove that the lemma always holds. Note that $\ts(S,0) \in  V^{\-{col}}_{\+B}$ is trivial and we only need to prove the second property.

Consider the whole execution of $\appmghcoledge(S,\gamma)$. 
During the algorithm, we say $x \in V_H^S$ joins $V^{\-{col}}_{\+B}$ once the condition in~\eqref{eq:color-bad-timestamps} is met by $x$.
%
Note that all vertices join $V^{\-{col}}_{\+B}$ in order during the execution of the algorithm. 
Let $V^{\-{col}}_{\+B} = \{x^{(1)},x^{(2)},\ldots,x^{(\ell)}\}$, where $x^{(i)}$ is the $i$-th vertex that joins the set $V^{\-{col}}_{\+B}$. Note that $x^{(1)}=\ts(S,0)$.
This $\ell$ is finite because $\appmghcoledge( S,\gamma)$ terminates within a finite number of steps.
We show that for any $i \in [\ell]$, there exists a path $P=(y_1,y_2,\ldots,y_m)$ in $G^S_H$ such that 
\begin{itemize}
	\item $y_1=\ts(S,0)$ and $y_m = x^{(i)}$;
	\item for all $1 \leq j \leq m$, $y_j \in V^{\-{col}}_{\+B}$.
\end{itemize}
This result immediately proves the lemma.

We prove the result by induction on index $i$. The base case trivially holds by taking $P=(\ts(S,0))$.

Fix an integer $1 <k \leq \ell$.
Suppose the result holds for all $x^{(i)}$ for $i < k$. We prove the result for $x^{(k)}$. We have the following cases:
\begin{enumerate}
    \item If $x^{(k)}$ joins $V^{\-{col}}_{\+B}$ in $\appmghcoledge(S,\gamma;M,R)$ when some edge $e\in \+{E}$ is processed, then we have $x^{(k)}=\ts(e,0)$.  By the condition in \Cref{Line-color-edge-while}, there must exist some $w\in V' \cap e$. The set $V'$ is initially set as $S$ in \Cref{Line-color-edge-init} and is updated only in \Cref{Line-color-edge-update} after recursive calls. We have two further cases:
    \begin{enumerate}
      \item If $w\in S$, it suffices to take $P=(x^{(1)},x^{(k)})$.
      \item Otherwise $w$ must have been added into $V'$ in \Cref{Line-color-edge-update} after some instance is triggered in \Cref{Line-color-edge-assign}, 
        which implies that there exists $e'\in \+{E}$ and $i < k$ such that $w\in e'$ and $x^{(i)} = \ts(e',0)$. 
        By the induction hypothesis, there is a path $P= (y_1,y_2,\ldots,y_{m'})$ for $x^{(i)}$ with $y_{m'} = x^{(i)}$. Note that $\upd_w(0)\in x^{(k)}=\ts(e,0)$ and $\upd_w(0)\in x^{(i)}=\ts(e',0)$. 
        It suffices to take the path $P'=(P,x^{(k)})$. 
    \end{enumerate}
  \item If $x^{(k)}$ joins $V^{\-{col}}_{\+B}$ in $\appmghcol(t_0,\gamma;M,R)$ when some edge $e\in \+{E}$ is processed, 
    then we have $x^{(k)}=\ts(e,t_0)$. 
    By the condition in \Cref{Line-color-while-1}, there must exist some $w\in V'\cap e$. 
    Here the set $V'$ is initially set as $\set{v_{i(t_0)}}$ in \Cref{Line-color-init-1} and is updated only in \Cref{Line-color-update-1} after some recursive calls are triggered in \Cref{Line-color-assign-1}. 
    We have two further cases:
    \begin{enumerate}
      \item If $w=v_{i(t_0)}$, we consider what parent evoked $\appmghcol(t_0,\gamma;M,R)$.
        \begin{enumerate}
          \item It is evoked by $\appmghcoledge(S,\gamma)$ when the latter processes some $e'\in \+{E}$ or $e'= S$.
            Then $w\in e'$ and there exists $i < k$ such that $x^{(i)} = \ts(e',0)$. 
            By the induction hypothesis, there is a path $P= (y_1,y_2,\ldots,y_{m'})$ for $x^{(i)}$ with $y_{m'} = x^{(i)}$. 
            Because $x^{(i)}$ triggers a recursive call of $\appmghcol{}(t_0,\gamma;M,R)$, $t_0\in x^{(i)}$.
            Also, $t_0=\upd_w(t_0)\in x^{(k)}=\ts(e,t_0)$.
            It suffices to take the path $P'=(P,x^{(k)})$. 
          \item It is evoked by $\appmghcol(t_1,\gamma;M,R)$ for some $t_1<t_0$ when the latter processes some $e'\in \+{E}$.
            Then $w\in e'$ and there exists $i < k$ such that  and $x^{(i)} = \ts(e',t_1)$. 
            By the induction hypothesis, there is a path $P= (y_1,y_2,\ldots,y_{m'})$ for $x^{(i)}$ with $y_{m'} = x^{(i)}$. 
            Note that $t_0\in x^{(i)}$ because $x^{(i)}$ triggers a recursive call of $\appmghcol{}(t_0,\gamma;M,R)$  and $t_0=\upd_w(t_1)\in x^{(k)}=\ts(e,t_0)$.
            It suffices to take the path $P'=(P,x^{(k)})$. 
        \end{enumerate}
      \item Otherwise $w$ must have been added into $V'$ in \Cref{Line-color-update-1} after some recursive calls, 
        which implies that there exists $e'\in \+{E}$ and $i < k$ such that $w\in e'$ and $x^{(i)} = \ts(e',t_0)$. 
        By the induction hypothesis, there is a path $P= (y_1,y_2,\ldots,y_{m'})$ for $x^{(i)}$ with $y_{m'} = x^{(i)}$. 
        Note that $\upd_w(t_0)\in x^{(k)}=\ts(e,t_0)$ and $\upd_w(t_0)\in x^{(i)}=\ts(e',t_0)$.
        It suffices to take $P'=(P,x^{(k)})$. 
    \end{enumerate}
\end{enumerate}
This finishes the induction proof.
\end{proof} 

Next, we show the following property for the set $V^{\-{col}}_{\+B}$.

\begin{lemma}\label{lemma:color-all-1}
For all $x \in V^{\-{col}}_{\+B}\setminus \set{\ts(S,0)}$, there exists $j\in [s]$ and $ t_0\in x$ such that $r_{t'} = \perp$ or $r_{t'}=j$ for all $t' \in x\setminus \set{t_0}$.
\end{lemma}
\begin{proof}
Fix  $x \in V^{\-{col}}_{\+B}\setminus \set{\ts(S,0)}$. Then $x$ joins $V^{\-{col}}_{\+B}$ in the following two cases.
\begin{enumerate}
    \item $x = \ts{}(e,0)$ triggers an instance of $\appmghcol(t_0,\gamma;M,R)$ for some $t_0 \leq 0$ when $\appmghcoledge(S,\gamma)$ processes some hyperedge $e \in \+E$ in \Cref{Line-color-edge-assign}.
      Thus, $e$ must satisfy the condition in \Cref{Line-color-edge-satisfy}, which shows that there exists $j\in [s]$ such that $r_{t'} = \perp$ or $r_{t'}=j$ for all $t' \in x$.
    \item $x = \ts{}(e,t_0)$ triggers an instance of $\appmghcol(t_1,\gamma;M,R)$ for some $t_1 < t_0$ when $\appmghcol(t_0,\gamma;M,R)$ processes some hyperedge $e \in \+E$ in \Cref{Line-color-assign-1}.
      Thus, $e$ must satisfy the condition in \Cref{Line-color-satisfy-1}, which shows that there exists some $j\in [s]$ such that $r_{t'} = \perp$ or $r_{t'}=j$ for all $t' \in x\setminus \set{t_0}$.
\end{enumerate}
This finishes the proof.
\end{proof}

Recall the definition of $2$-tree in \Cref{definition:2-tree}.
We are now ready to prove \Cref{lemma:color-efficiency}.
\begin{proof}[Proof of \Cref{lemma:color-efficiency}]
By \Cref{lemma:color-randomness-bound}, it suffices to show $\Pr{\abs{V^{\-{col}}_{\+B}}\geq 2\Delta k^2\cdot \eta}\leq 2^{-\eta}$ for each $\eta\geq 1$ as
\begin{align*}
	\Pr{|R| \geq 4\Delta^2 k^5 \cdot \eta} \leq \Pr{|V^{\-{col}}_{\+B}|2 \Delta k^3 \geq 4\Delta^2 k^5 \cdot \eta} = \Pr{|V^{\-{col}}_{\+B}| \geq 2\Delta k^2 \cdot \eta} \leq \tp{\frac{1}{2}}^{\eta}.
\end{align*}

Fix $\eta \geq 1$. 
Assume $\abs{V^{\-{col}}_{\+B}}\geq 2\Delta k^2\cdot \eta$.
Note that $V^{\-{col}}_{\+B}$ is finite because $\appmghcol$ terminates within a finite number of steps.
By \Cref{lemma:modified-wg-degree-bound}, \Cref{lemma:color-connected}, and \Cref{lemma:big-2-tree} there exists a $2$-tree $T\subseteq V^{\-{col}}_{\+B}$ of size $i$ such that $\ts(S,0)\in T$. 
For each $\eta\geq 1$, denote by $\+{T}^{\eta}_S$ the set of $2$-trees $T$ in $G^S_{H}$ of size $\eta$ such that $\ts(S,0)\in T$.
Then by 
a union bound, we have
\begin{equation*}
 \Pr{\abs{V^{\-{col}}_{\+B}}\geq 2 \Delta k^2\cdot \eta} \leq  \sum\limits_{T\in \+{T}^{\eta}_S}\Pr{T\subseteq V^{\-{col}}_{\+B}}.
\end{equation*}

By \Cref{lemma:color-all-1}, the event $T\subseteq V^{\-{col}}_{\+B}$ implies for each $x\in T\setminus \set{\ts(S,0)}$, there exists $t_0\in x$ and $j\in [s]$ such that $r_{t'} = \perp$ or $r_{t'}=j$ for all $t' \in x \setminus \{t_0\}$. 
Due to $\lfloor q / s \rfloor^k\geq 4\mathrm{e}qs \Delta k$ and \Cref{cor:local-uniformity}, for any $x\in T\setminus \set{\ts(S,0)}$, this happens with probability at most 
\[
sk\left(\frac{\left\lceil\frac{q}{s}\right\rceil}{q}\left(1+\frac{1}{s}\right)+\frac{1}{4s}\right) ^{k-1}\leq sk\left(\frac{5}{4s}\left(1+\frac{1}{s}\right)+\frac{1}{4s}\right)^{k-1} \leq 2k\left(\frac{2}{s}\right)^{k-2},
\]
where we obtain the first term by a union bound  over all possible $t_0\in x$ and $j\in [s]$ and noting $q\geq 4s$, and the second inequality is by $s\geq 6\Delta^{\frac{2}{k-2}}\geq 6$.

Since for any $2$-tree $T\subseteq V^S_H$, the timestamps in vertices of $T$ are pairwise disjoint,
and the event above are all mutually independent.
Thus, we have $\Pr{T\subseteq V^{\-{col}}_{\+B}}\leq (2k)^{(\abs{T}-1)}(2/s)^{(k-2)\cdot (\abs{T}-1)}$ and by \Cref{lemma:modified-wg-degree-bound} and \Cref{lemma:2-tree-number-bound}, it holds that 
\begin{align*}
     \sum\limits_{T\in \+{T}^{\eta}_S}\Pr{T\subseteq V^{\-{col}}_{\+B}} &\leq  \frac{(4\mathrm{e}\Delta^2k^4)^{\eta-1}}{2}\cdot (2k)^{\eta-1}\cdot \tp{\frac{2}{s}}^{(k-2)(i-1)}\\
       &\leq \frac{1}{2} \tp{\frac{8\mathrm{e}\Delta^2 k^5}{(s/2)^{k-2}} }^{\eta-1} \\
       &\leq   2^{-\eta},
\end{align*}
where the last inequality holds because $s \geq 2\left(8\mathrm{e}\Delta^2 k^5\right)^{\frac{1}{k-2}}$ from $k\geq 20$ and $s\geq 6\Delta^{\frac{2}{k-2}}$.
\end{proof}

\subsection{Improved bounds for linear hypergraphs}\label{sec-linear-color}
We now give a marginal sampler for linear hypergraphs, and prove \Cref{lem-sampling-color-simple}. 
Let $\delta > 0$ be a constant $k \geq \frac{50(1+\delta)^2}{\delta^2}$, and $q\geq 50\Delta^{\frac{2+\delta}{k-3}}$.
Given as inputs a linear $k$-uniform hypergraph $H=(V,\+E)$ with the maximum degree $\Delta$, a set of vertices $S\subseteq V$ that $\abs{S}\leq k$, and a parameter $\gamma>0$, 
The algorithm is almost the same, except that we replace the truncation condition in \Cref{Line-color-edge-force-return-1} of \Cref{Alg:marginal-edge-color-sampler} and \Cref{Line-color-force-return-1} of \Cref{Alg:marginal-color-sampler-app} with
\begin{align}\label{eq-color-simple}
|R| \geq 3\cdot 10^4 \tp{\frac{1+\delta}{\delta}}^2 \Delta^3 k^{10}  \left \lceil \log{\frac{1}{\gamma}}\right \rceil 
\end{align}

%

Like what we have done for linear hypergraph independent sets in \Cref{sec-linear-ind}, much of the analysis of the general hypergraph colourings can be applied to the linear case. 
We also reuse some other proved results from \Cref{sec-linear-ind}. 

The running time of the modified algorithm is bounded the same way as \Cref{lem-alg-edge-col-time}, whose proof we shall omit again.
\begin{lemma}\label{lem-time-linear-color}
The running time of the modified algorithm is 
\begin{align*}
	O\left(\left(\frac{1+\delta}{\delta}\right)^4\Delta^8k^{21}\log^2{\frac{1}{\gamma}} q^{6\cdot 10^4(\frac{1+\delta}{\delta})^2\Delta^4k^{11}\log{\frac{1}{\gamma}}}\right).
\end{align*}	
\end{lemma}

Next, we bound the truncation error.

\begin{lemma}\label{lemma:color-linear-efficiency}
Denote $\eta=\lceil\log\frac{1}{\gamma}\rceil$. 
If the projection scheme $[q]\to[s]$ satisfies (a) $q\geq 4s$, (b) $s \geq 6\Delta^{\frac{1+1/(1+2/\delta)}{k}}$, and (c) $\lfloor q / s \rfloor^k\geq 4\mathrm{e}qs \Delta k$, 
then upon the termination of the modified algorithm, the size of $R$ satisfies
\begin{align}
    \Pr{\abs{R}\geq  3\cdot 10^4 \tp{\frac{1+\delta}{\delta}}^2 \Delta^3 k^{10} \cdot \eta}\leq 2^{-\eta}.
\end{align}
\end{lemma}

Assuming this for now, we prove \Cref{lem-sampling-color-simple}. 
\begin{proof}[Proof of \Cref{lem-sampling-color-simple}]
Choose $s=\lfloor q^{\frac{1}{2}}\rfloor$, and then verify the three conditions of \Cref{lemma:color-linear-efficiency}. 

\begin{itemize}
  \item[(a)] Because $k\geq 50(1+\delta)^2/\delta^2\geq 50>3$, it holds that $q\geq 50\Delta^{(2+\delta)/(k-3)}\geq50>16$, and hence $q>4\sqrt{q}\geq 4s$. 
  \item[(b)] This is derived as follows.  
  \begin{gather*}
    s\geq\sqrt{q}-1>7\Delta^{\frac{1+\delta/2}{(k-3)}}-1\geq 7\Delta^{\frac{1+1/(1+2/\delta)}{k-3}}-1\geq 7\Delta^{\frac{1+1/(1+2/\delta)}{k}}-1\geq 6\Delta^{\frac{1+1/(1+2/\delta)}{k}}.
  \end{gather*}
  \item[(c)] For all $k\geq 50$, we have
  \begin{gather*}
  k^{\frac{1+1/(1+2/\delta)}{k-3}}\leq k^{\frac{2}{k-3}}<1.2<2. 
  \end{gather*}
  This gives
  \begin{gather*}
    \left\lfloor\frac{q}{s}\right\rfloor\geq 3\times 2\times\Delta^{\frac{1+1/(1+2/\delta)}{k-3}}> 3(k\Delta)^{\frac{1+1/(1+2/\delta)}{k-3}}, 
  \end{gather*}
  and hence
  \begin{gather*}
    \left\lfloor\frac{q}{s}\right\rfloor^{k}\geq 3^{k-3}(k \Delta) \left\lfloor \frac{q}{s} \right\rfloor^3 \geq 3^{40} \cdot 4\mathrm{e}k\Delta\cdot \left\lfloor\frac{q}{s}\right\rfloor^3\geq 4\mathrm{e}qs \Delta k. 
  \end{gather*}
\end{itemize}

We use the modified $\appmghcoledge(S,\gamma)$ to sample from the distribution $\mu_S$. By \Cref{lem-time-linear-color}, the running time is $O((\frac{1+\delta}{\delta})^4\Delta^8k^{21}\log^2{\frac{1}{\gamma}} q^{6\cdot 10^4(\frac{1+\delta}{\delta})^2\Delta^4k^{11}\log{\frac{1}{\gamma}}})$. 
It is straightforward to verify that \Cref{Alg:marginal-edge-color-sampler} uses at most $6\cdot 10^4 \tp{\frac{1+\delta}{\delta}}^2 \Delta^3 k^{10}  \left \lceil \log{\frac{1}{\gamma}}\right \rceil +1$ random variables with domain size at most $q+1$.
By combining \Cref{lemma:edge-col-correct} and \Cref{lemma:color-efficiency}, the output distribution is $\gamma$-close to $\mu_S$. 
This proves the lemma.
\end{proof}

The rest of this section is dedicated to the proof of \Cref{lemma:color-linear-efficiency}. 
Recall from \Cref{sec-linear-ind} that we introduce the self-neighbourhood powered witness graphs to make use of linearity under the hypergraph independent set setting. 
We provide an analogous definition regarding the modified witness graph (\Cref{definition:witness-graph-color}) whilst dealing with hypergraph colourings. 
%

\begin{definition}\label{definition:self-neighbourhood-witness-graph-color}
Let $G^S_H=(V^S_H,E^S_H)$ be the modified witness graph in \Cref{definition:witness-graph-color}.
The self-neighbourhood powered witness graph $G_H^{S,\-{self}} = (V^S_H, E_H^{S,\-{self}})$ is  defined on the same vertex set and the edge set $E_H^{S,\-{self}} = E^S_H \cup E'$ such that
\begin{align*}
E' = \{\{x,y\} \mid   (\exists w \in V^S_H \text{ s.t. } w \in N_{\-{self}}(x) \land w \in N(y)) \lor (\exists w \in V^S_H \text{ s.t. } w \in N(y) \land w \in N_{\-{self}}(x)) \}.	
\end{align*}
\end{definition}  

Observe that $G_H^{S,\-{self}}$ modifies $G_H^{\-{self}}$ by only adding the new vertex $\ts(S,0)$ and some additional edges that connects $\ts(S,0)$ with several other vertices, if $S\notin\+E$. 

The next lemma is an analogue of \Cref{lem-max-degree-G-self}.
\begin{lemma}\label{lem-max-degree-G-self-mod}
If $\abs{S}\leq k$, then the maximum degree of $G_H^{S,\-{self}}$ is at most $10k^3\Delta - 1$.
\end{lemma}
\begin{proof}
For any $x \in V_H^S$, by \Cref{lemma:wg-degree-bound}, $|N_{\-{self}}(x)| \leq 2k$ and by \Cref{lemma:modified-wg-degree-bound}, $N(x) \leq 2\Delta k^2 - 1$. Hence, the maximum degree of $G_H^{S,\-{self}}$ is at most 
$
2\Delta k^2 - 1 + 2 \times 2k \times (2\Delta k^2 - 1) \leq  10k^3\Delta - 1.
$
\end{proof}

We also have the following lemma as a counterpart of \Cref{lemma-vb-linear}. 
\begin{lemma}\label{lemma-vb-linear-1}
Given a $k$-uniform linear hypergraph $H = (V,\+E)$ with the maximum degree $\Delta$, 
and a subset of vertices $S\subseteq V$, 
let $G^S_H=(V^S_H, E^S_H)$ be the modified witness graph with respect to $S$. 
Let $V^{\-{col}}_{\+B} \subseteq V^S_H$ be a finite subset containing $\ts(S,0) \in V^{\-{col}}_{\+B}$ and connected in $G_H^S$.
Then, there exists $V_{\+B}^{\-{lin}} \subseteq V^{\-{col}}_{\+B} $ such that 
\begin{enumerate}[label = (C\arabic*)]
	\item $\ts(S,0) \in V_{\+B}^{\-{lin}} $ and $|V_{\+B}^{\-{lin}}| \geq \lfloor\frac{|V_{\+B}|}{2k+1}\rfloor$, \label{enum:vbcol-size}
	\item the induced subgraph $G_H^{S,\-{self}}[V_{\+B}^{\-{lin}}]$ is connected, and \label{enum:vbcol-connected}
	\item for any two distinct vertices $x_1,x_2 \in V_{\+B}^{\-{lin}}\setminus \set{\ts(S,0)}$, it holds that $|x_1 \cap x_2| \leq 1$. \label{enum:vbcol-linear}
\end{enumerate}
\end{lemma}

\begin{proof}
The proof is almost the same as the proof of \Cref{lemma-vb-linear}:
we apply the explicit construction algorithm there starting from the vertex $\ts(S,0)$. 
To see \ref{enum:vbcol-size}, note that the number of self-neighbours is at most $2k$ for any vertex in $G_H^S$. 
\ref{enum:vbcol-connected} follows from the construction and the same argument as \ref{enum:vblin-connected} in \Cref{lemma-vb-linear}.
\ref{enum:vbcol-linear} also follows the same argument as \ref{enum:vblin-linear} in \Cref{lemma-vb-linear}. 
We remark that the artifact $\ts(S,0)$ needs to be excluded as there is no control on how the set $S$ intersects with the hyperedges in $H$. 
\end{proof}


Now we prove \Cref{lemma:color-linear-efficiency}. 
Recall from the last section the toolkit of $2$-block-trees, including \Cref{definition:2-block-tree}, \Cref{lemma:existence-2-block-tree}, and \Cref{lemma:2-block-tree-number-bound}. 

\begin{proof}[Proof of \Cref{lemma:color-linear-efficiency}]
Let $ V^{\-{col}}_{\+{B}}$ be the set generated by the modified algorithm (see \eqref{eq-color-simple}) as defined in \eqref{eq:color-bad-timestamps}. 
Again, $ V^{\-{col}}_{\+{B}} \subseteq V^S_H$ is a finite subset because the algorithm terminates after a finite number of steps. 
Define a parameter
\begin{align*}
	\theta \defeq \left \lceil\frac{6(1+\delta)}{\delta} \right \rceil, \end{align*}
and by this choice, $\Pr{|R| \geq 3\cdot 10^4 (\frac{1+\delta}{\delta})^2 k^{10} \Delta^3 \eta} \leq \Pr{|R| \geq 600\theta^2k^{10} \Delta^3 \eta}$.
Using  \Cref{lemma:color-randomness-bound}, it holds for any positive integer $\eta$ that 
\begin{align*}
	\Pr{|R| \geq 600\theta^2k^{10} \Delta^3 \eta} \leq \Pr{ |V^{\-{col}}_{\+{B}}|2k^3\Delta \geq 600\theta^2k^{10} \Delta^3 \eta } \leq \Pr{|V^{\-{col}}_{\+{B}}| \geq 300\theta^2k^7 \Delta^2 \eta  }.
\end{align*}
Hence, it suffices to show
\begin{align*}
	\Pr{|V^{\-{col}}_{\+{B}}| \geq 300\theta^2k^7 \Delta^2 \eta  } \leq \tp{\frac{1}{2}}^\eta
\end{align*}
for any integer $\eta \geq 1$. 

Fix an integer $\eta \geq 1$, and assume 
$|V^{\-{col}}_{\+{B}}|\geq 300\theta^2k^7 \Delta^2 \eta$, by \Cref{lemma:color-connected} and \Cref{lemma-vb-linear-1}, we can find the set $ V_{\+{B}}^{\-{lin}} \subseteq V^{\-{col}}_{\+{B}}$ with size
$ |V_{\+{B}}^{\-{lin}}| \geq \lfloor \frac{|V^{\-{col}}_{\+{B}}|}{2k + 1} \rfloor \geq \lfloor\frac{|V^{\-{col}}_{\+{B}}|}{3k} \rfloor \geq \theta^2 (10k^3\Delta)^2 \eta$
such that $\ts(S,0) \in V_{\+{B}}^{\-{lin}}$ and the conditions in \Cref{lemma-vb-linear-1} get fulfilled.
Moreover, it is straightforward to find a subset $U \subseteq V_{\+{B}}^{\-{lin}}$ with size exactly $|U| = \theta^2 (10k^3\Delta)^2 \eta$ such that $\ts(S,0) \in U$ and the rest of \Cref{lemma-vb-linear-1} are satisfied by $U$.
By \Cref{lem-max-degree-G-self-mod}, the maximum degree of $G_H^{S,\-{self}}[U]$ is at most $10k^3\Delta$.
Since $G_H^{S,\-{self}}[U]$ is a finite connected subgraph in  $G_H^{S,\-{self}}$, by \Cref{lemma:existence-2-block-tree}, we can find a $2$-block-tree $\{C_1,C_2, \dots ,C_{\eta} \}$ in $G_H^{S,\-{self}}$ with block size $\theta$ and tree size $\eta$
such that $\ts(S,0) \in C_1$ and $C_j \subseteq U \subseteq V^{\-{col}}_{\+{B}}$ for all $j\in [\eta]$. 
By \Cref{enum:vbcol-linear} in \Cref{lemma-vb-linear-1}, for any distinct $x_1,x_2 \in (\cup_{j=1}^\eta C_j)\setminus \set{\ts(S,0)}$, it holds that $|x_1 \cap x_2| \leq 1$.

Define by $\+T^{\eta,\theta}_S$ the set of all $2$-block-trees $\{C_1,C_2,\ldots,C_\eta\}$ with block size $\theta$ and tree size $\eta$ in graph $G_H^{S,\-{self}}$ such that
\begin{itemize}
	\item $\ts(S,0) \in C_1$;
	\item let $ C = (\cup_{j=1}^\eta C_j)\setminus \set{\ts(S,0)}$, then for any $w_1,w_2 \in C$, $|w_1 \cap w_2| \leq 1$.
\end{itemize}
Hence, if $|V^{\-{col}}_{\+{B}}|\geq 300\theta^2k^7 \Delta^2 \eta $, then there exists a 2-block-tree $\{C_1,C_2,\ldots,C_\eta\} \in \+T^{\eta,\theta}_S$ such that $C_j \subseteq V^{\-{col}}_{\+{B}}$ for all $j \in [\eta]$.
By a union bound over all 2-block-trees in $\+T^{\eta,\theta}_S$, we have 
\begin{align*}
\Pr{|V^{\-{col}}_{\+{B}}| \geq 3\theta^2 k\Delta^2 \eta } \leq \sum_{\{C_1,\ldots,C_\eta\} \in \+T^{\eta,\theta}_S} \Pr{\forall j \in [\eta], C_j \subseteq V^{\-{col}}_{\+{B}}}.
\end{align*}

Fix a 2-block-tree $\{C_1,\ldots,C_\eta\} \in \+T^{\eta,\theta}_S$. 
By definition, for any $j$ and $\ell$ that $j \neq \ell$, we have $\dist_{G_H^{S,\-{self}}}(C_j,C_\ell) \geq 2$, 
and thus for any $x_j \in C_j$ and $x_\ell \in C_\ell$, it holds that $x_j \cap x_\ell = \emptyset$. 
For any $j\in[\eta]$, and any two $w_1, w_2 \in C_j \setminus \{\ts(S,0)\}$, it holds that  $|w_1 \cap w_2| \leq 1$ by the definition of $\+T^{\eta,\theta}_S$.
Let $C_j=\{e^j_1,e^j_2,\dots,e^j_{\theta}\}$. 
Without loss of generality, assume $\ts(S,0)=e^j_{\theta}$ if $\ts(S,0) \in C_j$. 
For each $\ell\in[\theta]$, define 
\begin{align}\label{eq-def-Bjl}
\+{B}^{j}_{\ell}: \text{There exists } x\in [s] \text{ and } y\in e^{j}_{\ell}
\text{ s.t., for all } t' \in e^{j}_{\ell}\setminus (\set{y }\cup \set{\ts(S,0)}), \text{ either } r_{t'} = \perp\text{or } r_{t'}=x.
\end{align}
We then have
\begin{align*}
&\Pr{\forall j \in [\eta], C_j \subseteq V^{\-{col}}_{\+{B}}}\\
(\text{By \Cref{lemma:color-all-1} and chain rule})\quad \leq &\prod\limits_{1\leq j\leq \eta}\prod\limits_{1\leq \ell\leq \theta} \Pr{\+{B}^{j}_{\ell}\mid \bigwedge\limits_{\substack{(j',\ell'):\\j'<j \lor (j'=j\land \ell'<\ell)}}\+{B}^{j'}_{\ell'}}\\
(\text{By $e^{j}_{\ell} \cap e^{j'}_{\ell'}=\emptyset$ for $j\neq j'$ })\quad = & \prod\limits_{1\leq j\leq \eta}\prod\limits_{1\leq \ell\leq \theta} \Pr{\+{B}^{j}_{\ell}\mid\bigwedge\limits_{\ell'<\ell} \+{B}^{j}_{\ell'}}\\
(\star)\quad \leq & \prod\limits_{1\leq j\leq \eta}\prod\limits_{1\leq \ell\leq \theta-1} \left(sk\cdot \left(\frac{\left\lceil\frac{q}{s}\right\rceil}{q}\left(1+\frac{1}{s}\right)+\frac{1}{4s}\right) ^{k-\theta }\right) \\
(\text{By $q\geq 4s$ and $s \geq 6$})\quad  \leq &(sk)^{(\theta-1)\eta}  \left(\frac{2}{s}\right)^{\eta(k-\theta)(\theta-1)}\leq (2k)^{(\theta - 1)\eta} \left(\frac{2}{s}\right)^{\eta(k-\theta - 1)(\theta-1)}
\end{align*}
The inequality $(\star)$ is due to
(1) a union bound according to the definition in~\eqref{eq-def-Bjl}; 
(2) the local uniformity property in \Cref{cor:local-uniformity}; and 
(3) \ref{enum:vbcol-linear} in \Cref{lemma-vb-linear-1} for each $e_{\ell}^j$ where $j\in[\theta-1]$.

Next, we count the number of possible 2-block-trees in $\+T^{\eta,\theta}_S$,
which can be upper bound by the number of all 2-block-trees $\{C_1,C_2,\ldots,C_\eta\}$ with block size $\theta$ and tree size $\eta$ in $G_H^{S,\-{self}}$ such that $\ts(S,0) \in C_1$. 
By \Cref{lemma:2-block-tree-number-bound} and \Cref{lem-max-degree-G-self-mod}, we have
\begin{align*}
\abs{\+T^{\eta,\theta}_S(U)} \leq  (\theta \mathrm{e}^\theta (10 k^3 \Delta)^{\theta+1} )^{\eta}.	
\end{align*}
Hence, we only need to prove that 
\begin{align*}
(\theta \mathrm{e}^\theta (20 k^4 \Delta)^{\theta+1} )^{\eta}	\tp{\frac{2}{s}}^{(k-\theta-1)(\theta-1)\eta} \leq \tp{\frac{1}{2}}^\eta\end{align*}
which is equivalent to 
\begin{align*}
	\left(\frac{s}{2}\right)^{(k-\theta-1)(\theta-1)} \geq 2\theta \mathrm{e}^\theta (20 k^4\Delta)^{\theta+1} \quad \Longleftrightarrow \quad 	\left(\frac{s}{2}\right)^{(k-\theta-1)} \geq (2\theta)^{1/(\theta-1)} \mathrm{e}^{\theta/(\theta-1)} (20k^4 \Delta)^{(\theta+1)/(\theta-1)}.
\end{align*}
We derive this as follows. 
Observe the following inequalities
\begin{align}
k \geq \frac{50(1+\delta)^2}{\delta^2} \geq \theta^2-1 \implies k-\theta-1 \geq (1-1/(\theta-1))k, \label{equ:linearcol-ineq1} \\
k \geq 50 \implies (120k^4)^{2/k}<2.3<3. \label{equ:linearcol-ineq2}
\end{align}
Then we have
\begin{align*}
  s &\geq 6\Delta^{\frac{1+1/(1+2/\delta)}{k}} \tag{By condition}\\
  &\geq 2\left(120k^4\right)^{\frac{2}{k}}\Delta^{\frac{1+1/(1+2/\delta)}{k}} \tag{By~\eqref{equ:linearcol-ineq2}}\\
  &\geq 2\left(120k^4\Delta\right)^{\frac{1+1/(1+2/\delta)}{k}}\\
  &\geq 2\left(6^{\frac{\theta-1}{k(\theta - 2)}} (20k^4 \Delta)^{\frac{\theta+1}{{k(\theta - 2)}}}\right)\tag{Use $6+\frac{6}{\delta}\leq\theta\leq7+\frac{6}{\delta}$}\\
  &\geq 2\left(6\left(20k^4\right)^{\frac{\theta+1}{\theta-1}}\right)^{\frac{1}{k-\theta-1}}. \tag{By~\eqref{equ:linearcol-ineq1}}
\end{align*}
The desired inequality then follows by noticing that $\theta\geq 6$. 
%
%
\end{proof}

\subsection{Proof of \texorpdfstring{\Cref{lemma:edge-col-correct}}{Lemma 6.11}}\label{sec-proof-colour-correct}
Now we prove \Cref{lemma:edge-col-correct}, verifying the correctness of \Cref{Alg:marginal-edge-color-sampler}.
We follow the same proof strategy as in \Cref{sec:marginal-sampler}.
We will first define a forward version marginal sampler $\+A$ in \Cref{Alg:idea} for any finite $T$,
and then give its backward counterpart $\+B$ in \Cref{Alg:B}.
Using the same randomness, $\+A$ and $\+B$ are perfectly coupled.
\Cref{Alg:marginal-edge-color-sampler} is the same as $\+B$ taking $T\rightarrow\infty$ and with truncation.
\Cref{lemma:edge-col-correct} then follows in a similar manner as in \Cref{thm:dtv-truncate}.

For the rest of this subsection, we fix a $k$-uniform hypergraph $H=(V,\+{E})$ with the maximum degree at most $\Delta$ together with a set of colours $[q]$, a projection scheme $h$ with parameter $s$, a subset of vertices $S \subseteq V$ and a real number $\gamma > 0$ satisfying the condition in \Cref{lemma:edge-col-correct}.  Fix an integer $T \geq 0$. 
For the forward marginal sampler,
recall that the systematic scan Glauber dynamics on the distribution $\psi$ is defined as follows:
\begin{itemize}
	\item let $Y_{-T} \in [s]^V$ be an arbitrary feasible configuration;
	\item for each $t$ from $-T+1$ to $0$, the transition $Y_{t-1} \to Y_{t}$ is defined as follows:
	\begin{enumerate}
		\item let $w = v_{i(t)}$, where $i(t) = (t \text{ mod } n) + 1$,
          and let $Y_{t}(u) = Y_{t-1}(u)$ for all $u \neq w$; 
		\item sample $r_t \sim \psi^{\-{LB}}$, if $r_t \neq \perp$, then let $Y_t(w) = r_t$; otherwise, let
		\begin{align*}
			&\sigma_{\Lambda} = \config(t)
		\end{align*}
		where $\config \text{ is defined in }$ \Cref{Alg:config-h-col}.
		\item sample $Y_t(w) \sim \psi^{\-{pad},\sigma_{\Lambda}}_{w}$;
	\end{enumerate}
	\item output $Y_0$.
\end{itemize}

Here we sample from the padding distribution by the method in \eqref{eq:proof:resolve-coupling} using a sequence of real random variables $\{U_t\}_{-T<t\le 0}$.

Denote by $Y_0^{\+P(T)} \in [s]^V$ the output of the above process. 
We will always consider $T \geq n$. It holds that for any $u \in V$, $Y_0^{\+P(T)}(u) = Y_{t}(v_{i(t)})$, where $t = \upd_u(0)$. 

Given $Y_0^{\+P(T)}$ generated by the process $\+P(T)$,
what we are really interested in is a sample from $\mu$ (instead of $\psi$) conditioned on $Y_0^{\+P(T)}$.
To this end, 
we introduce the algorithm $\+A\left( S, Y_0, \{r_t\}_{-T< t \leq 0}\right)$, described in \Cref{Alg:idea},
that takes a configuration $Y_0$ and a set $\{r_t\}_{-T< t \leq 0}$ as its input.
This algorithm is similar to \Cref{Alg:config-h-col} in that it uses BFS to find a boundary so that $S$ is independent from the outside,
and in addition it enumerates all possibilities inside to generate a sample.
Let $X_S^{(T)} = \+A\left(S, Y_0^{\+P(T)},  \{r_t\}_{-T< t \leq 0}\right)$,
where $Y_0^{\+P(T)}$ and $\{r_t\}_{-T< t \leq 0}$ are generated by $\+P(T)$.
Note that $\left\{X_S^{(T)}\right\}_{T\ge n}$ is an infinite sequence of random variables.
The next lemma shows that it converges to the desired marginal distribution.

\begin{algorithm}
\caption{$\+A\left( S, Y_0, \{r_t\}_{-T< t \leq 0}\right)$} \label{Alg:idea}
  \SetKwInput{KwPar}{Parameter}
 \KwIn{a hypergraph $H = (V,\+E)$, a set of colours $[q]$, a projection scheme $h: [q] \to [s]$ that defines the projected distribution $\psi$, a subset $S \subseteq V$, a configuration $Y_0 \in [s]^V$, and random variables $r_t \in [s] \cup \{\perp\}$ for $-T \leq t\leq 0$;}

\KwOut{$\sigma \in [s]^{\Lambda}$ for some $S \subseteq \Lambda \subseteq V$}
\SetKwIF{Whenever}{}{}{whenever}{do}{}{}{}
$\sigma\gets Y_0(S)$ and $V'\gets S$\; 
\While{$\exists e\in \+{E}$ s.t. $e\cap V'\neq \emptyset$, $e\cap (V\setminus V')\neq \emptyset$ and $e$ is not satisfied by $\sigma$}{
      choose such $e$ with the lowest index\;\label{Lina-A-pick} 
    \If{$\exists j \in [s]$ s.t. $\forall u \in e$, $r_{\upd_u(0)} \in \{\perp,j\}$}{
       \ForAll{$u \in e$} {$\sigma(u)\gets Y_0(u)$ \label{line-app-1}\;}
       $V'\gets V'\cup e$\label{line-app-2}\;
    }
    \Else{
        $U\gets \{u \in e \mid r_{\upd_u(0)}\neq \perp\}$\;
    	 \ForAll{$ u \in U$}{$\sigma(u)\gets  r_{{\upd_u(0)}}$\;}
    }
} 
enumerate all colourings $ X\in \otimes_{u \in V'}h^{-1}(\sigma_u)$ on $V'$ to compute the marginal distribution of $S$ on the sub-hypergraph $H[V']$,
which is equivalent to $\mu^{\sigma}_{S}$\;  
sample $X_S\sim \mu^{\sigma}_{S}$\; 
\Return $X_S$\;
\end{algorithm}

\begin{lemma}\label{lemma-app-A}
If  $\lfloor q / s \rfloor^k\geq 4\mathrm{e}qs \Delta k$, then $\DTV{X_S^{(T)}}{\mu_S} = 0$ as $T \to \infty$.
\end{lemma}
\begin{proof}
  Fix $Y =  Y_0^{\+P(T)} \in [s]^V$ and $\{r_t\}_{-T< t \leq 0}$ generated by $\+P(T)$.
We first show that the output $X_S = X_S^{(T)}$ follows the distribution $\mu^{Y}_S$. 
For each hyperedge $e \in \+E$, if $e$ is picked in \Cref{Lina-A-pick} , then after the loop, either $e \subseteq V'$ or $e$ is satisfied by $\sigma$, and thus $e$ cannot be picked in \Cref{Lina-A-pick} again. 
The algorithm $\+A$ always terminates.
Suppose the final $\sigma$ is defined on a subset $\Lambda \subseteq V$, namely $\sigma \in [s]^{\Lambda}$.
Note that $X$ is sampled from $\mu^{\sigma}_S$. We show that the configuration $\sigma$ satisfies:
\begin{itemize}
	\item $\sigma = Y_{\Lambda}$ and $S \subseteq V' \subseteq \Lambda$;
	\item for all $e \in \+E$ such that $e \cap V' \neq \emptyset$ and $e \cap (V \setminus V') \neq \emptyset$, $e$ is satisfied by $\sigma$.
\end{itemize}
If these two properties are true, then conditioned on $\sigma$, all constraints on hyperedges in the boundary of $V'$ are satisfied by $\sigma$, which implies $\mu^Y_S = \mu^\sigma_S$.
This is similar to the argument in \Cref{lemma:config-h-col-cor}.

Note that $S \subseteq V'$ because $V'$ is initialised as $S$ and it does not remove vertices.
The fact $ V' \subseteq \Lambda$ follows from \Cref{line-app-1} and \Cref{line-app-2}.
Since both $Y$ and $\{r_t\}$ are generated by $\+P(T)$,
for any $u \in V$, if $r_{\upd_u(0)} \neq \perp$, then $Y_u  = r_{\upd_u(0)}$.
Hence, it is straightforward to see  $\sigma = Y_{\Lambda}$. 
The second property follows from the fact that the while-loop has terminated.

By \Cref{lemma:config-h-col-cor}, $\config$ (\Cref{Alg:config-h-col}) satisfies \Cref{cond:invariant-boundary}.
The process $\+P(T)$ faithfully simulates the systematic scan Glauber dynamics on $\psi$.
By \Cref{cor:local-uniformity}, the systematic scan Glauber dynamics on $\psi$ is irreducible, which implies that $Y_0^{\+P(T)}$ follows the distribution $\psi$ as $T \to \infty$. Since $X_S^{(T)} \sim \mu^{Y_0^{\+P(T)}}_S$, $\DTV{X_S^{(T)}}{\mu_S} = 0$ as $T \to \infty$.
\end{proof}

The algorithm $\+A$ uses too much randomness and next we give its backwards version $\+B$ that achieves the same output distribution.
Fix the initial configuration $Y_{-T}$ of $\+P(T)$ as an arbitrary feasible configuration, say $Y_{-T}(u) = 1$ for all $u \in V$.
Given the description of the distribution $\psi$, an integer $T \geq n$ and a subset $S \subseteq V$, 
the algorithm $\+B$ returns a random variable that follows the same distribution as $X_S = \+A\left(S, Y_0^{\+P(T)}, \{r_t\}_{-T< t \leq 0}\right)$. 

To construct $\+B$, we need yet another algorithm $\+C=\+C_T(t;M,R)$,
which is to plug \Cref{Alg:config-h-col} as $\config$ into $\resolve_T(t;M,R)$ (\Cref{Alg:resolve}) in this context.
As before, $\+C$ maintains two global data structures $M$ and $R$, initialised as $\perp^{\mathbb{Z}}$ and $\emptyset$ respectively.
All recursive calls of $\+C$ access the same $M$ and $R$.
The algorithm $\+C$ uses $\basicsample{}(t;R)$ to draw random samples from $\psi^{\-{LB}}$,
and samples from the padding distribution by the method in \eqref{eq:proof:resolve-coupling}. 
Given the random variables $\{r_t, U_t\}_{-T< t \leq 0}$,
$\+C$ and $\basicsample{}$ become deterministic.
We define $C \in [s]^V$ where $C(u) = \+C_T(\upd_u(0);M,R)$ for all $u \in V$.
Note that $C$ is a function of $\{r_t, U_t\}_{-T< t \leq 0}$,
and is thus a random vector.

\begin{lemma}\label{lemma-app-C}
  Use the same random variables $\{r_t, U_t\}_{-T< t \leq 0}$ in $\+C_T(\upd_u(0);M,R)$ for all $u \in V$ and in $\+P(T)$
  to generate $C$ and $Y_0$, respectively.
  The distribution of $C$ is the same as that of $Y_0$.
\end{lemma}

\begin{proof}
By \Cref{lemma:config-h-col-cor}, the subroutine $\config$ satisfies \Cref{cond:invariant-boundary}.
Following the proof of \Cref{theorem:resolve-cor},
once we used the same random variables $\{r_\ell, U_\ell\}_{-T< t \leq 0}$,
$\+C$ and the $\+P(T)$ are perfectly coupled and the lemma follows.
\end{proof}

\Cref{lemma-app-C} says that we can use $\+C$ to simulate $Y_0$ in $\+P(T)$. 
We may also use $\basicsample$ to access $\{r_t\}_{-T< t \leq 0}$ ``on demand''.
The algorithm $\+B=\+B_T(S)$ (\Cref{Alg:B}) takes $S$ as an input,
and is the same as $\+A$ except that 
any access to $Y_0$ is replaced by a call to $\+C$, 
and any access to $\{r_t\}_{-T< t \leq 0}$ is replaced by a call to $\basicsample$.
Similar to \Cref{Alg:marginal-edge-color-sampler}, $\+B$ together with its subroutines maintains two global data structures $M$ and $R$, 
which are initialised respectively as $M_0=\perp^{\mathbb{Z}}$ and $R_0=\emptyset$.
\begin{algorithm}
\caption{$\+B_T(S)$} \label{Alg:B}
  \SetKwInput{KwPar}{Parameter}
 \KwIn{a hypergraph $H = (V,\+E)$, a set of colours $[q]$, a projection scheme $h: [q] \to [s]$ that defines the projected distribution $\psi$, a subset $S \subseteq  V$, an integer $T \leq -n$ and a parameter $\gamma>0$;}
\SetKwInput{KwData}{Global variables}
 \KwData{a map $M: \mathbb{Z}\to [q]\cup\{\perp\}$ and a set $R$;}
\KwOut{a random assignment in $\tau\in [q]^{S}$}
$M \gets \perp^{\mathbb{Z}}$ and $R \gets \emptyset$\;
$\sigma\gets \emptyset,V'\gets S$\;
$\sigma(u)\gets \+C_T(\upd_u(0);M,R)\text{ for all }u\in S$\;
\While{$\exists e\in \+{E}$ s.t. $e\cap V'\neq \emptyset$, $e\cap (V\setminus V')\neq \emptyset$ and $e$ is not satisfied by $\sigma$}{
      choose such $e$ with the lowest index\; 
    \If{$\exists j \in [s]$ s.t. $\forall u \in e$, $\basicsample({\upd_u(0)};R) \in \{\perp,j\}$}{
       \ForAll{$u \in e $} {$\sigma(u)\gets \+C_T(\upd_u(0);M,R)$ \;}
       $V'\gets V'\cup e$\;
    }
    \Else{
        $U\gets \{u \in e\mid \basicsample({\upd_u(0)};R)\neq \perp\}$\;
    	 \ForAll{$ u \in U$}{$\sigma(u)\gets  \basicsample({\upd_u(0)};R)$\;}
    }
} 
enumerate all colourings $ X\in \otimes_{u \in V'}h^{-1}(\sigma_u)$ on $V'$ to compute the marginal distribution of $S$ on the sub-hypergraph $H[V']$,
which is equivalent to $\mu^{\sigma}_{S}$\;  
sample $X_S\sim \mu^{\sigma}_{S}$\;\label{Line-B-last-sample}
\Return $X_S$\; 
\end{algorithm}

%
Denote by $B_T(S)$ to the output of the algorithm $\+B_T(S)$.
Recall that $X^{(T)}_S$ denotes the output of $\+A\left( S, Y_0^{\+P(T)}, \{r_t\}_{-T< t \leq 0}\right)$.
Suppose $\lfloor q / s \rfloor^k\geq 4\mathrm{e}qs \Delta k$.
By \Cref{lemma-app-C}, it holds that 
\begin{align}\label{eq-app-eq}
\forall T \geq n, \quad 	\DTV{B_T(S)}{X^{(T)}_S} = 0.
\end{align}
By \Cref{cor:local-uniformity}, if $\lfloor q / s \rfloor^k\geq 4\mathrm{e}qs \Delta k$, \Cref{condition:sufficient-correctness} is satisfied. 
By \Cref{lemma:config-h-col-cor}, the subroutine $\config$ satisfies \Cref{cond:invariant-boundary}.
Following the same proof of \Cref{theorem:resolve-cor},
one can verify that in this case $\+B_{\infty}(S)$ terminates with probability 1 and its output $B_{\infty}(S)$ satisfies 
\begin{align}\label{eq-app-inf}
	\lim_{T \to \infty} \DTV{B_{\infty}(S)}{B_T(S)} = 0.
\end{align}
Note that 
\begin{align*}
  \limsup_{T \to \infty} \DTV{\mu_S}{B_T(S)} &\leq \limsup_{T \to \infty} \DTV{\mu_S}{X^{(T)}_S} + \limsup_{T \to \infty} \DTV{X^{(T)}_S}{B_T(S)} \\
  \tag{by \Cref{lemma-app-A} and~\eqref{eq-app-eq}}	&= 0.
\end{align*}
Combining the above with \eqref{eq-app-inf}, we have
\begin{align}\label{eq-cor-app}
  \DTV{B_{\infty}(S)}{\mu_S} \leq \limsup_{T \to \infty} \DTV{B_{\infty}(S)}{B_T(S)} + \limsup_{T \to \infty} \DTV{\mu_S}{B_T(S)} = 0.
\end{align}
Now, we can prove \Cref{lemma:edge-col-correct}.
\begin{proof}[Proof of \Cref{lemma:edge-col-correct}]
  Let $K = 4\Delta^2 k^5 \lceil\log{\frac{1}{\gamma}}\rceil$ be the truncation threshold of $|R|$ in \Cref{Alg:marginal-edge-color-sampler}.
  We couple $\appmghcoledge(S,\gamma)$ and the algorithm $\+B_\infty(S)$ by using the same random choices $\{r_{\ell},U_{\ell}\}_{\ell \leq 0}$,
  where $\{U_\ell\}_{\ell \leq 0}$ are used to realise the padding distribution as in \eqref{eq:proof:resolve-coupling}.
  Furthermore, we use the same uniformly at random real number $U^* \in [0,1]$ 
  to realise both \Cref{Line-color-edge-tau-sample} of \Cref{Alg:marginal-edge-color-sampler} and \Cref{Line-B-last-sample} of $\+B_\infty(S)$.
%

Let $X_S$ be output of \Cref{Alg:marginal-edge-color-sampler}. 
Similar to the proof of \Cref{thm:dtv-truncate}, $\+B_\infty(S)$ and \Cref{Alg:marginal-edge-color-sampler} couple perfectly unless truncation happens.
Thus,
\begin{align*}
  \Pr[\-{coupling}]{B_\infty(S) \neq X_S}\leq \Pr{\+E_{\-{trun}}(K)},
\end{align*}
where $B_\infty(S)$ denotes the output of the algorithm $\+B_\infty(S)$,
and the the second $\mathbf{Pr}$ refers to \Cref{Alg:marginal-edge-color-sampler}.
The random variable $B_\infty(S)$ is well-defined since $\+B_\infty(S)$ terminates with probability $1$. 
Then,
\begin{align*}
	\DTV{X_S}{\mu_S} &\leq  \DTV{X_S}{B_{\infty}(S)} + \DTV{B_{\infty}(S)}{\mu_S} \\
    \tag{\text{by \eqref{eq-cor-app}}} \quad&= \DTV{X_S}{B_{\infty}(S)}\\
	&\leq \Pr[\-{coupling}]{B_{\infty}(S) \neq X_S} \leq \Pr{\+E_{\-{trun}}(K)}. \qedhere
\end{align*}
\end{proof}

\section{Derandomising random scan Glauber dynamics using CTTP}\label{sec:derandomise-random-scan}

In this section, we elaborate how one can derandomise the MCMC algorithm based on the random scan Glauber dynamics for Gibbs distributions over bounded-degree graphs. 

\subsection{CTTP for random scan Glauber dynamics}
Let $\+S:\mathbb{Z}\to V$ be an (infinite) scan sequence and let $T>0$.
We consider the following Glauber dynamics with scan sequence $\+S$.

\begin{center}
  \begin{tcolorbox}[=sharpish corners, colback=white, width=1\linewidth]
  	\begin{center}
	\textbf{\emph{The Glauber dynamics $\+{P}^{\+S}(T)$}}
  	\end{center}
  	\vspace{6pt}
   \begin{itemize}
	\item Initialize $X^{\+S}_{-T} \in [q]^V$ as an arbitrary feasible configuration.
	\item For $t=-T+1,-T+2,\ldots,0$,  the configuration $X^{\+S}_{t}$ is constructed as follows:
	\begin{enumerate}
		\item[(a)] pick $v = \+S(t)$,  and let $X^{\+S}_{t}(u) \gets X^{\+S}_{t-1}(u)$ for all $u \neq v$; 
		\item[(b)] draw $r_t \sim \mu^{\mathrm{LB}}$ independently, and  let $X^{\+S}_t(v) \gets r_t$ if $r_t \neq \perp$; otherwise, 
\begin{align}\label{eq-findconfig-fixed-scan}
\sigma_\Lambda \gets \config^\+S(t), \quad\text{(by accessing $X^{\+S}_{t-1}$ and $\+{R}_{t-1}\defeq (r_s)_{-T < s < t}$)} 
\end{align} 
		and draw $X^\+S_t(v) \sim \mu^{\mathrm{pad},\sigma_\Lambda}_{v}$ independently.
	\end{enumerate}
\end{itemize}
  \end{tcolorbox} 
\end{center}

Compared with the systematic scan Glauber dynamics $\+P(T)$ defined in \Cref{sec:marginal-sampler}, the only difference is that the vertex to update is chosen according to $\+S$ rather than \eqref{eq:i-scan}. 
Furthermore, for any $u\in V$ and integer $t$, denote by $\upd^{\+S}_u(t)$ the last time in $\+S$ before $t$ at which $u$ is updated, i.e.
\begin{align}\label{eq:update-time-fixed-scan}
\upd^{\+S}_u(t)\defeq \max\{s\leq t \mid \+S(s)=u\},
\end{align}
If no such $s$ exists, let $\upd^\+S_u(t)=-T$.
With such replacement in mind, we can define $\config^\+S(t)$, $\resolve^{\+S}_T(t; M,R)$, and $\appresolve^\+S(v, K; M,R)$, analogously to those in \Cref{sec:marginal-sampler}. 
An  analogue of \Cref{cond:invariant-boundary} with random or fixed scan  is stated as follows.
\begin{condition}\label{cond:invariant-boundary-fixed-scan}
With probability 1 over the choices of ${\+S}$,
  the procedure $\config^{\+S}(t)$ terminates and returns $\sigma_\Lambda\in[q]^\Lambda$ satisfying that $\mu_v^{\sigma_\Lambda}=\mu_v^{X_{t-1}(V\setminus\{v\})}$ for $v=v_{\+S(t)}$.
\end{condition}

The correctness of $\resolve^{\+S}_T$ is due to the next theorem, whose proof is almost identical to that of \Cref{theorem:resolve-cor}.

\begin{theorem}\label{theorem:resolve-cor-fixed-scan}
  Let $\mu$ be a distribution over $[q]^V$, $T\ge 0$ be an integer, $\+S:\mathbb{Z}\to V$ be a (fixed or random) scan sequence and 
  $(X^\+S_t)_{-T\le t\le 0}$ be generated by the process $\+{P}^\+S(T)$ whose $\config^\+S(t)$ subroutine satisfies \Cref{cond:invariant-boundary-fixed-scan}.
  For any $-T\le t \leq 0$, the followings hold:
  \begin{itemize}
  \item
  For any fixed scan sequence $\+S:\mathbb{Z}\to V$,
  $\resolve_{T}^{\+S}(t)$ terminates in finite steps and returns a sample identically distributed as $X^{\+S}_t(\+S(t))$;
  \item
  Suppose further \Cref{condition:sufficient-correctness} holds. 
  If for randomly generated sequence $\+S$ such that $\+P^\+S(T)$ converges to $\mu$ as $T\to \infty$,  
  then $\resolve^{\+S}_{\infty}(t)$ terminates with probability $1$, and  returns a sample distributed as $\mu_v$ where $v=\+S(t)$.
  \end{itemize}
\end{theorem}

On top of this, taking $\+S$ as an infinite random sequence whose each entry is chosen from $V$ uniformly and independently at random gives \Cref{Alg:appresolve-random}. This is the random scan analogue of \Cref{Alg:appresolve}.
\begin{algorithm}
\caption{$\appresolverandom(v, K; M,R)$} \label{Alg:appresolve-random}
  \SetKwInput{KwPar}{Parameter}
  \SetKwInput{KwData}{Global variables}
  \SetKwIF{Try}{Catch}{Exception}{try}{:}{catch}{exception}{}
\KwIn{a variable $v\in V$ and $K \geq 0$}
\KwData{a map $M:\mathbb{Z}\to[q]\cup\{\perp\}$ and a set $R$;}
\KwOut{a random value in $[q]\cup\{\perp\}$;}
initialize $M\gets\perp^{\mathbb{Z}}$ and $R\gets\emptyset$\;
Generate infinite scan sequence $\+S:\mathbb{Z}\to V$ where each $\+S(i)$ is chosen from $V$ uniformly and independently at random\;\label{Line:appresolve-random-gen}
\uTry{}{
\Return $\resolve^\+S_{\infty}(\upd^\+S_v(0);M,R)$\label{alg:appresolve-fixed-scan-line-return}\;
}
\Catch{$|R| \geq K$\label{alg:appresolve-fixed-scan-line-catch}}{ 
\Return $\perp$\;
}
\end{algorithm}

 Denote by $\+E_{\-{trun}}(K)$ as the event $\appresolverandom(v,K)=\perp$. We have the following theorem that says this is precisely the error for $\appresolverandom(v,K)$ sampling from $\mu_v$. This can be viewed as the random scan Glauber dynamics analogue for \Cref{thm:dtv-truncate}.

\begin{theorem}\label{thm:dtv-truncate-random-scan}
Let $\mu$ be a distribution over $[q]^V$. 
Assume  \Cref{condition:sufficient-correctness} and \Cref{cond:invariant-boundary-fixed-scan}.
For $v\in V$, $K\ge 0$, 
and $Y=\appresolverandom(v,K)$, 
it holds that  $\DTV{Y}{\mu_v}= \Pr{\+E_{\-{trun}}(K)}$.
\end{theorem}

\begin{proof}[Proof]
Suppose that sampling from the padding distribution $\mu_{v_{i(\ell)}}^{\-{pad},\sigma_{\Lambda}}$ in $\resolve^{\+S}_T$ is realised in the same way as in~\eqref{eq:proof:resolve-coupling}, 
using a sequence of real numbers $U_\ell \in [0,1)$ chosen uniformly and independently at random for each $\ell\le 0$.

We apply the following coupling between $\resolve^{\+S}_{\infty}(\upd^S_v(0))$ and $\appresolverandom{}(v,K)$:
\begin{itemize}
	\item the two processes use the same random choices for $(r_{\ell})_{\ell \leq 0}$ and $(U_{\ell})_{\ell \leq 0}$.
\end{itemize}
\sloppy
One can verify that if  $\+E_{\-{trun}}(K)$ does not occur, then the two processes $\resolve^\+S_{\infty}(t)$ and $\appresolverandom{}(v,K)$ are coupled perfectly. 
By the coupling lemma and \Cref{theorem:resolve-cor-fixed-scan}, 
\begin{align*}
\DTV{Y}{\mu_v} \leq \Pr{\+E_{\-{trun}}(K)}.
\end{align*}
Note that $\Pr{Y = \perp} =  \Pr{\+E_{\-{trun}}(K)}$ and $\mu_v(\perp) = 0$. Therefore, we have
\begin{align*}
\DTV{Y}{\mu_v} \geq \Pr{Y = \perp} - \mu_v(\perp) = \Pr{\+E_{\-{trun}}(K)}.
\end{align*}
Combining the two inequalities proves the theorem.
\end{proof}

\subsection{Derandomising CTTP for random scan Glauber dynamics}
The rest of this section strives to show how one can derandomise \Cref{Alg:appresolve-random}. 
Specifically, 
\begin{theorem}\label{thm:derandomise-random-scan}
 Let $\mu$ be a $q$-spin Gibbs distribution with the maximum degree of the underlying graph $\Delta$.
  The distribution of the output of $\appresolverandom(v,K)$ can be deterministically computed in $\poly(q^K,\Delta^K)$ time.
\end{theorem} 

\begin{remark}
The above should be treated as a generic theorem. 
Depending on the system being studied, one needs to bound the probability of visiting more than $K$ vertices during the backward deduction in order to control the error introduced by truncation. 
Typically, for fixed $q$ and $\Delta$, if taking $K=O(\log n)$ ensures the desired error bound, then the running time is polynomial in $n$. This is exactly the error bound we need for derandomising systematic scan Glauber dynamics using CTTP. That being said, for Gibbs distributions, derandomising random scan Glauber dynamics using CTTP is no harder than systematic scan Glauber dynamics. 
\end{remark}
\subsubsection{Witness tree}

We begin with introducing the structure that we ``count'' on. 
Fix the underlying graph $G=(V,E)$ for now. 
\begin{definition}[witness tree]\label{def:witness-tree}
  A \emph{witness tree} $\tau$ is a finite rooted tree together with, on each of its vertices, a vertex label from $V$ and a random choice label from $[q]\cup\{\perp\}$ where the children of some vertex in $\tau$ with vertex label $u\in V$ receive vertex labels from $N(u)$. 
\end{definition}

Hereinafter, we generally use the letters $u,v,w,\cdots$ to represent the vertices in the original graph, and $\alpha,\beta,\gamma\cdots$ for the vertices in the witness tree (and another combinatorial structure arising from the witness tree later). 

Fix an infinite scan sequence $\mathcal{S}$ and the set of random choices $\{r_t\}_{t\in\mathbb{Z}_{\leq 0}}$ (\emph{for the lower bound distribution}) in advance. 
Without loss of generality, the vertex $v$ we start is at timestamp $0$. 
We present the following algorithm (\Cref{Alg:glauber-witness}) for constructing a witness tree from the scan sequence $\+S$.
\begin{algorithm}
  \caption{Glauber dynamics witness tree} \label{Alg:glauber-witness}
    \SetKwInput{KwPar}{Parameter}
  \KwIn{a graph $G$, a scan sequence $\mathcal{S}$ and random choices $\{r_t\}_{t\in\mathbb{Z}_{\leq 0}}$, and a size threshold $K$} 
  \KwOut{an associated witness tree $\tau$}
  construct a single-vertex tree $\tau$ with the root labeled $(v,r_0)$\;
  \If{$r_0\neq\perp$}{\Return $\tau$ and $\emptyset$\;}
  $A\gets N(v)$\;
  $t\gets 0$\;
  \While{$A\neq\emptyset$ and $|\tau|<K$}{
      decrease $t$ till $u:=\mathcal{S}(t)\in A$\;
      $A\gets A-\{u\}$\;
      find the deepest vertex $\alpha'$ in $\tau$ with label $(u',\cdot)$ such that $u\in N(u')$; tie-breaking by lexicographical order of the vertex label\; \label{Alg:glauber-witness-line-tree}
      insert, as a child of $\alpha'$, a fresh vertex with label $(u,r_t)$\;
      \If{$r_t=\perp$}{$A\gets A\cup N(u)$\;}
  }
  \Return $\tau$ and $A$\; 
\end{algorithm}

\begin{remark}[Compare to the Moser-Tardos witness trees]
    The witness trees defined in \Cref{def:witness-tree} and constructed by \Cref{Alg:glauber-witness} , apart from the additional random choice labels, are the same as the witness trees constructed in the analysis of the Moser-Tardos algorithm for algorithmic Lov\'{a}sz Local Lemma~\cite{MT10}. Moreover, \Cref{Alg:glauber-witness} resembles the procedure for constructing witness trees out of an execution log for the Moser-Tardos algorithm. This is because this specific combinatorial structure succinctly captures a ``local total ordering", which is a chronological total ordering between all neighbouring vertices/events encountered in $\appresolverandom(v, K)$ or the Moser-Tardos algorithm. 
\end{remark}

The construction of the witness tree in \Cref{Alg:glauber-witness} only depends on the input graph $G$, the scan sequence $\+S$ and the random choices $\{r_t\}_{t\in \mathbb{Z}_{\leq 0}}$. 
In other words, it is independent of the implementation of the algorithm as long as it simulates the Glauber dynamics. 
In our setting, the random choices are samples from the lower bound distribution $\mu^{\mathrm{LB}}$, where a vertex no more requires any recursion if it samples anything other than $\perp$ from $\mu^{\mathrm{LB}}$. 
The set $A$ collects the vertices that are needed in order to deduce the correct marginal of current unresolved vertices, including the starting vertex. 
When the construction halts with $A=\emptyset$, the marginal of the root only depends on samples from the padding distribution for those receiving $\perp$ from the lower bound distribution. 
This motivates the definition of a full witness tree. 
\begin{definition}
  A witness tree $\tau$ constructed in \Cref{Alg:glauber-witness} is called \emph{full}, if the set $A$ is empty upon the termination of the algorithm. 
\end{definition}

We then collect some simple properties of the above process. 
First, inside the \textbf{while}-loop, one of the neighbours of the current vertex $u$ needs to be visited in prior so that $u$ can join the set $A$, and hence:
\begin{observation}
\Cref{Alg:glauber-witness} is well-defined.
Specifically, such vertex in \Cref{Alg:glauber-witness-line-tree} always exists. 
\end{observation}

We then argue that there is no need for further tie-breaking among vertices of the same level sharing the same vertex label in \Cref{Alg:glauber-witness-line-tree}. 
\begin{lemma} \label{lem:witness-tree-no-same-level}
No two distinct vertices $\alpha$ and $\beta$ of the same level in $\tau$ share the same vertex label. 
\end{lemma}

\begin{proof}
Given a vertex $\alpha$ in the tree, let $t(\alpha)$ be the variable $t$ in the algorithm when $\alpha$ is inserted. 
Suppose $\alpha$ and $\beta$ receive the same vertex label $u$ and $t(\alpha)<t(\beta)\leq 0$. 
As soon as $\beta$ gets added in the tree, the vertex $u$ is no more present in the set $A$. 
However, as $t$ decreases, $\alpha$ joins the tree with the same vertex label.
This means there must be some time $t(\alpha)<t<t(\beta)$ when the algorithm visits some $u'$ with $u\in N(u')$, detects the random choice being $\perp$, and introduces $u$ into the set $A$. 
Let $\gamma$ be the vertex in the tree created at this step. 
By definition, $\gamma$ has label $(u',\perp)$ and must be deeper than $\beta$ in the tree. 
Then as the algorithm progresses to time $t(\alpha)$, any parent that $\alpha$ can choose must be either the same or deeper level of $\gamma$. 
This implies that $\alpha$ must be deeper than $\beta$ in the tree. 
\end{proof}

Finally, note that if the tree is full, then the algorithm terminates with empty $A$. 
This implies the following observation. 
\begin{observation} \label{lem:witness-tree-full}
If $\tau$ is a full witness tree, then it holds that: for any vertex $v\in V$ in the original graph $G$, if there exists a vertex $\alpha\in\tau$ labelled $(v,\perp)$, then for every $u\in N(v)$, there must be some other vertex $\beta\in\tau$ with vertex label $u$ that is deeper than $\alpha$. 
\end{observation}

Imagine there are three updates taking place at vertices $w,w,z$ chronologically where $w\in N(z)$. 
The spin for the first update on $w$ gets overwritten by the second one, which the afterward update on $z$ depends on. 
This motivates us to extract the actual dependency out from the witness tree, which yields a directed acyclic graph (DAG). 
We formalise the construction of such dependency DAG as below. 
(In the algorithm below, as $W$ and $\tau$ shares the same vertex set, we may use the same variable to refer to the vertex in either structures, depending on the context.)
\begin{algorithm}
\caption{Witness dependency graph} \label{Alg:glauber-dependency}
  \SetKwInput{KwPar}{Parameter}
\KwIn{witness tree $\tau$ with a fixed root} 
\KwOut{a dependency DAG}
$B\gets\emptyset$\;
initialise $W$ with $\tau$'s vertex set\;
\ForEach{vertex $\alpha$ labelled $(v,r)$ in the tree such that $r=\perp$}{
\ForEach{$u\in N(v)$}{
find the shallowest vertex $\beta$ labelled $u$ that is deeper than $\alpha$ in $\tau$\; \label{Alg:glauber-dependency-line-find}
\eIf{such $\beta$ exists}{
add an arc $\beta\to\alpha$ into $W$\;
}
{
$B\gets B+u$\;
}
}
}
\Return $W$ and $B$\;
\end{algorithm}

The above process is deterministic; that is, the algorithm outputs a unique DAG for a given witness tree. 
Observe that \Cref{lem:witness-tree-no-same-level} implies the $\beta$ found in \Cref{Alg:glauber-dependency-line-find}, if exists, is always unique. 
Moreover, such $\beta$ always exists if $\tau$ is full due to \Cref{lem:witness-tree-full}.
This observation can be formalised and generalised as follows. 

\begin{lemma} \label{lem:dependency-reconstruction}
Suppose $\tau$ is a tree returned by \Cref{Alg:glauber-witness} with the set $A$. 
Then with input $\tau$, the set $B$ returned by \Cref{Alg:glauber-dependency} is the same as $A$. 
\end{lemma}

\begin{proof}
Proof by induction. 
If $\tau$ is of size $1$, then the lemma holds trivially. 
Assume this lemma holds for all witness trees of size at most $L-1$. 
Let $\tau$ be a witness tree of size $L$. 
Consider the last vertex $\gamma$ that is added into $\tau$ during \Cref{Alg:glauber-witness}, whose label is $(v,r)$. 
If $r=\perp$, then the addition of $\gamma$ in \Cref{Alg:glauber-witness} updates $A\gets A\cup N(v)-\{v\}$. 
The rule of adding $\gamma$ into $\tau$ ensures that $\gamma$ is below any vertex with vertex label from $\gamma$. 
Therefore, in \Cref{Alg:glauber-dependency} any $\alpha$ that causes $v$ to join $B$ in the $L-1$ case now finds such $\gamma$, and hence $v$ is absent in $B$. 
However, any vertex below $\gamma$ does not have any vertex label in $N(v)$, causing \Cref{Alg:glauber-dependency-line-find} fails for all vertices in $N(v)$. 
This introduces $N(v)$ into $B$. 
Hence, $A$ and $B$ agree with each other. 
The case that $r\neq\perp$ is identical to the argument above but without introducing $N(v)$. 

This lemma then follows by the principle of induction. 
\end{proof}

A similar induction also shows that the DAG constructed captures the correct dependency introduced during the recursion of the original algorithm. 
Moreover, it can be shown that
\begin{lemma}\label{lem:full-dependency}
Suppose $\tau$ is a full witness tree, and $W$ is the DAG constructed by \Cref{Alg:glauber-dependency} with input $\tau$. 
Then any vertex in $W$ with in-degree $0$ has a non-$\perp$ random choice label. 
\end{lemma}

\begin{proof}
By $\tau$ is a full witness tree, the set $A$ returned by \Cref{Alg:glauber-witness} is empty. By \Cref{lem:dependency-reconstruction}, \Cref{Alg:glauber-dependency-line-find} always succeeds for vertices with random choice labels $\perp$. 
Thus, they have in-degree at least $1$. 
\end{proof}

\subsubsection{Derandomisation}

Harnessed with the witness tree and dependency DAG, we are now able to carry out the derandomisation efficiently. 
Assume each vertex in the infinite scan sequence $\mathcal{S}$ is independently drawn uniformly at random, and the random choices are drawn independently subject to the lower bound distribution. 
Rather than enumerating over all possible scan sequences and the random choices, we turn to enumerating the witness tree defined above in order to derandomise the random scan Glauber dynamics. 
However, different witness trees may emerge from the process with different probabilities, and this should be coped with first. 

Denote the witness tree generated by the process starting at $v$ as $\tau^v$. 
We provide an algorithm that, given a witness tree $\tau$, computes the probability of a witness tree $\tau$ emerging, namely $\Pr{\tau^v=\tau}$. 
This can be done by the following dynamic programming. 
\begin{itemize}
    \item Boundary case: If $\tau$ consists merely the root, then enumerate over all possible states $r$ of the lower bound distribution (including $\perp$), and set $f(\tau)=\Pr{\tau^v=\tau}=\mu^{\mathrm{LB}}(r)$. 
    \item Recursion: For each tree $\tau$ of size $K$, suppose the probabilities of all size-$(K-1)$ and degree-$\Delta$ trees have been computed. 
    Set $f(\tau)\gets 0$. 
    Enumerate over all its leaves and do the following: 
    \begin{itemize}
        \item Suppose the current leaf is labelled $(u,r)$. Remove the current leaf from $\tau$ to obtain $\tau'$. 
        \item Run \Cref{Alg:glauber-dependency} on $\tau'$ to obtain $B$. 
        \item If $v\in B$, then accumulate the current probability by
        \[
        f(\tau)\gets f(\tau)+\Pr{\tau^v=\tau'}\cdot|B|^{-1}\cdot\mu^{\mathrm{LB}}(r).
        \]
    \end{itemize}
    Then let $\Pr{\tau^v=\tau}=f(\tau)$. 
\end{itemize}

Now fix a \emph{full} witness tree $\tau$. 
Conditional independence together with \Cref{lem:full-dependency} implies that 
conditioned on such tree generated, the marginal of the spin 
conditioned on such tree generated, the marginal of the spin then depends solely on the DAG generated based on $\tau$.
Thus it is left for us to compute $\Pr{\sigma(v)=i\mid\tau^v=\tau}$. 

Recall that the algorithm needs to sample from the padding distribution for every internal vertex of $\tau$, and the dependency is captured by the DAG generated by \Cref{Alg:glauber-dependency}.  
We then enumerate over all possible configurations of the internal nodes other than the root, which receives a fixed spin $i$, and compute the probability of one such configuration by a simple traversal of the DAG in topological order. 

\sloppy
Assembling both parts, the following algorithm computes the output distribution of $\appresolverandom(v,K)$. 
\begin{itemize}
\item Enumerate, in the order of size, \emph{all} possible trees $\tau$ of size $\leq K$ and degree $\Delta$ with root labeled $v$, and for each of them, do the followings:
\begin{itemize}
    \item Compute $\Pr{\tau^v=\tau}$. 
    \item If $\tau$ is full, then for each spin $i$, compute $\Pr{\sigma(v)=i\mid\tau^v=\tau}$. 
        Accumulate the marginal of spin $i$ at $v$ by $\Pr{\tau^v=\tau}\Pr{\sigma(v)=i\mid\tau^v=\tau}$. 
\end{itemize}
\end{itemize}

\subsubsection{Running-time analysis}

We finally analyse the time efficiency of the above algorithm, and conclude \Cref{thm:derandomise-random-scan}. 
The key is to bound the number of possible subtrees of a certain size in a given graph. 
We need the following lemma. 
\begin{lemma}[\text{\cite[Lemma 2.1]{borgs2013left}}]\label{lem:subtree-bound}
  Let $G = (V, E)$ be a graph with maximum degree $\Delta$ and $v \in V$ be a vertex.  Then the number of subtrees of $G$ with $\ell$ nodes containing $v$ is at most $\frac{(\mathrm{e}\Delta)^{\ell-1}}{2}$.
\end{lemma}

The outer loop of the algorithm, namely enumerating over all possible witness trees of size $\leq K$, is done by noticing that all possible $\tau$ satisfy $\abs{V(\tau)}\leq K$. 
Note that considering only the vertex labels, all possible witness trees exist as a subtree of an infinite rooted tree with root labelled $v$, 
and all $\abs{N(u)}$ children of a node labelled $u$ receive distinct labels from the set $N(u)$. 
Therefore by \Cref{lem:subtree-bound}, there are at most $(\mathrm{e}\Delta)^{K}$ possible witness trees (without distinguishing the random choice labels).
We can then enumerate all possible random choice labels, resulting in at most $((q+1)\mathrm{e}\Delta)^{K}$ possibilities. 
Hence, there are $\poly(q^K,\Delta^K)$ possibilities to enumerate. 

For each of them, the first step where we compute $\Pr{\tau^v=\tau}$, which takes $\poly(q^K,\Delta^K)$ time due to dynamic programming that ranges over all possible smaller witness trees than the one being processed. 
The second step, computing $\Pr{\sigma(v)=i\mid\tau^v=\tau}$, requires us to enumerate all possible spins from the support of padding distribution, and each enumeration requires one topological traversal. 
This, in total, is also $\poly(q^K,\Delta^K)$. 
\Cref{thm:derandomise-random-scan} then follows.

\section{Concluding remarks}

In this paper, we propose a new framework for derandomising MCMC algorithms.
We introduce a method called {coupling towards the past} (CTTP) for evaluating a small amount of variables in their stationary states without simulating the entire chain,
which gives light-weight samplers that can draw from marginal distributions. 
Under strong enough marginal lower bound guarantees, 
CTTP terminates in logarithmic steps with high probability.
This provides a direct-sum style decomposition of the randomness used in MCMC sampling: 
a marginal sampler that can draw one random variable using only $O(\log n)$ random bits 
is extracted from an existing machinery for generating $n$ jointly distributed random variables using $O(n\log n)$ random bits.
A direct consequence to this is that, 
derandomising such marginal sampler becomes easy in polynomial time by a brute-force enumeration of all possible random choices,
which gives efficient deterministic counting algorithms via standard self-reductions.

As concrete applications, we obtain efficient deterministic approximate counting algorithms for hypergraph independent sets and hypergraph colourings, in regimes matching the state-of-the-art achieved by randomised counting/sampling algorithms. 

The current work makes the first step towards the goal of derandomising general Markov chain Monte Carlo algorithms. 
We summarize some challenges to the current framework which may lead to interesting future directions:

\begin{itemize}
    \item 
%

      Our current CTTP method relies crucially on the marginal lower bound being non-trivial, namely \Cref{condition:sufficient-correctness}.
      This condition seems necessary for our current implementation (see \Cref{rem:CTTP}),
      and it restricts problems our method can be applied to.
      For example, for graph colourings, no such lower bound holds,
      and yet efficient deterministic algorithms exist \cite{LSS22}.
      Not only that, the algorithm in \cite{LSS22} requires $q>2\Delta$,
      where $q$ is the number of colours and $\Delta$ is the maximum degree of the graph.
      This is close to the condition $q>(11/6-\epsilon)\Delta$, where $\epsilon\approx 10^{-5}$ is a small constant, for the best randomised algorithms \cite{chen2019improved}.
      It would be very interesting to find alternative ways to implement CTTP without marginal lower bounds to match the bounds from other methods.
      One idea is to expand $\perp$ to a set of possible values, similar to bounding chains \cite{huber2004perfect},
      and yet for this modified method to be efficient it apparently requires a condition of the order $q=\Omega(\Delta^2)$.

      In fact, our CTTP method implies perfect sampling algorithms,
      and the bound $q=\Omega(\Delta^2)$ (for the modified implementation) matches other general purpose perfect sampling algorithms \cite{huber2004perfect,FGY22,AJ22} applied to graph colourings.
      With refined techniques specific to this problem,
      the bounding chain approach can be made efficient under the condition $q \ge (8/3+o(1))\Delta$ \cite{BC20,JSS21}.
      It would be interesting to explore if these refined techniques can help our method as well.
    \item 
      Even if marginal lower bounds exist,
      our method still has a constant or lower order term gap in the conditions comparing to randomised algorithms.
      This is true for the two applications we consider in this paper,
      as is for problems such as estimating the partition function of the hardcore model.
      In the latter problem, efficient deterministic algorithms \cite{Wei06} work all the way up to the computational complexity transition threshold,
      and actually predate randomised counterparts with matching bounds \cite{ALO20,CLV21,chen2021rapid,anari2022entropic,CE22,CFYZ22}.
      However, none of these results use coupling techniques, which our method crucially relies on.
      An interesting direction is to find CTTP matching these results,
      and one potential lead is through finding more refined couplings.
    \item
    Our current direct-sum style proof for the MCMC sampling, 
    effectively decomposes an $O(n\log n)$-step Markov chain to an $O(\log n)$-step marginal sampler, 
    whose random choices are enumerable in polynomial time.
    It is then a challenge to derandomise those MCMC algorithms with significantly higher mixing time bounds to have, 
    for example, low-cost marginal samplers for graph matchings or bipartite perfect matchings.
    \item 
    So far, all known deterministic approximate counting algorithms suffer from running time bounds whose exponents depend on the parameters of the problem and/or the instance.
    It is still wide open to give a deterministic approximate counting algorithm with a polynomial running time where the exponent of the polynomial is an absolute constant, 
    especially when the instance is close enough to the critical threshold. 
    We hope our framework of derandomising MCMC can be combined with some classical derandomisation techniques, e.g.~$k$-wise independence, or their variants, 
    to make progress towards such a breakthrough.
\end{itemize}

\addtocontents{toc}{\protect\setcounter{tocdepth}{0}}

\ifdoubleblind
\else{
\section*{Acknowledgement}
This project has received funding from the European Research Council (ERC) under the European Union's Horizon 2020 research and innovation programme (grant agreement No. 947778). We would like to thank Hongyang Liu for the discussion in the beginning of the project.
}

\addtocontents{toc}{\protect\setcounter{tocdepth}{1}}









\bibliographystyle{alpha}
\bibliography{refs.bib}

\appendix

\section{Construction of 2-block-tree}\label{app-A}
\begin{proof}[Proof of \Cref{lemma:existence-2-block-tree}]
The lemma is proved by a greedy algorithm similar to Algorithm~4 in~\cite{FGW22a}.
We can assume $|C| \geq \theta$ as otherwise the lemma holds trivially by setting $\ell = 0$.
Let $G_C=G[C]$ be the connected subgraph of $G$ induced by $C$.
For any $u  \in C$, we use $\Gamma_{G_C}(u)$ to denote the neighbourhood of $u$ in $G_C$.
For any $\Lambda \subseteq C$, define $\Gamma_{G_C}(\Lambda)=\{u \in C \setminus \Lambda \mid \exists w \in \Lambda \text{ s.t. } u \in \Gamma_{G_C}(w) \}$.

Let $\ell = 0$ and $R = C$. The algorithm repeats the following process until $R = \emptyset$:
\begin{enumerate}
	\item $\ell \gets \ell + 1$, if $\ell = 1$, let $u = v$, if $\ell > 1$, let $u$ be an arbitrary vertex in $\Gamma_{G_C}(C \setminus R)$;
	\item find an arbitrary connected component $C_{\ell}$ in $G[R]$ satisfying $|C_\ell| = \theta$ and $u \in C_{\ell}$;
	\item $R \gets R \setminus (C_{\ell} \cup \Gamma_{G_C}(C_{\ell}))$; 
	\item for all connected components $G'=(V',E')$ in $G[R]$ with $|V'| < \theta$, let $R \gets R\setminus V'$.
\end{enumerate}  
Output the set $\{C_1,C_2,\ldots,C_{\ell}\}$.


We first show that the above algorithm is valid. The vertex $u$ in Line (1) can be found because if $\ell = 1$, $u = v$; if $\ell > 1$, we know that $R \neq \emptyset$ and $R \neq C$ (some vertices are deleted in previous steps), since $G_C$ is connected, $\Gamma_{G_C}(C \setminus R)$ is not empty as otherwise $R \neq \emptyset$ and $C \setminus R \neq \emptyset$ are disconnected, which contradicts with the fact $G_C$ is connected.
We then prove that the component $C_\ell$ can be found in Line (2). Note that $u \in R$. If $\ell = 1$, $C_{\ell}$ exists because $|C| \geq \theta $ and $G_C$ is connected. If $\ell > 1$, by Line (4), we know that $R$ is a set of connected components, and each component has a size at least $\theta$, and thus $C_{\ell}$ exists.

We next show that the output is a $2$-block-tree in graph $G$. It is straightforward to see that each $C_i$ has size $\theta$ and is connected in graph $G_C$, and thus is connected in graph $G$.
Note that after we found $C_i$, we remove all vertices in $\Gamma_{G_C}(C_i)$ in Line (3). Hence, for any $C_i$ and $C_j$ with $i \neq j$, $\dist_{G_C}(C_i,C_j) \geq 2$. Since $C_i,C_j \subseteq C$, it holds that $\dist_{G}(C_i,C_j) \geq 2$.
Finally, fix $2 \leq i \leq \ell$. Consider the $i$-th iteration. Let $u_{i}$ be the vertex picked in Line (1), we show that there exists $w \in C_1\cup C_2\cup\ldots\cup C_{i-1}$ such that $\dist_G(w,u_{i}) = 2$, which implies that $\bigcup_{j=1}^{i}C_i$ is connected on $G^2$. 
Let $R$ denote the set $R$ at the beginning of the $i$-th iteration.
Since $u_i \in \Gamma_{G_C}(C \setminus R)$, we know that in graph $G_C$, $u_\ell$ has a neighbour in $v' \in C \setminus R$, which is one of the vertices deleted in Line (3) or Line (4). 
By Line (3), we know that for any $j < i$, when $C_j$ is removed from $R$, $\Gamma_{G_C}(C_j)$ is also removed, which implies for any $j < i$, $v' \notin C_j$. Next, we show that $v'$ cannot in any component removed in Line (4), which implies $v'$ must in $\Gamma_{G_C}(C_j)$ for some $j < i$. Hence, there is a vertex $w \in C_j$ such that $\dist_{G_C}(w,u_i) = 2$. 
Note that both $w,u_i \in C$. It holds that $\dist_{G}(w,u_i) = 2$.
Suppose $v' \in V'$ for some $|V'| < \theta$ in Line (4). Note that $u_i$ and $v'$ are adjacent, and $u_i \in R$ belongs to a component with size at least $\theta$, which implies $|V'|$ also belongs to a component with size at least $\theta$, this implies a contradiction.
We remark that we proved a stronger result:  any prefix $\{C_1,C_2,\ldots,C_{i}\}$ for $i \in [\ell]$ is a 2-block-tree in $G$.

Finally, we bound the size of the output 2-block-tree. In Line (3), we remove at most $\theta (\Delta + 1)$ vertices. For each $V'$ in Line (4), we claim that there exists $w \in V'$ such that $w$ is adjacent to $\Gamma_{G_C}(C_\ell)$ in Line (3). Before the removal of $C_\ell \cup \Gamma_{G_C}(C_\ell)$, $V'$ belongs to another component $V'' \supset V'$ such that $|V''| \geq \theta$. 
However, after the removal, $V' \subseteq V''$ becomes a component of size $<\theta$, which implies some vertices adjacent to $V'$ must be removed. This proves the claim.
Note that $|V'| \leq \theta - 1$.
Hence the number of vertices removed in Line (4) is at most $(\theta - 1) \cdot  \theta\Delta^2$.
The total number of vertices removed in each iteration is at most $\theta (\Delta + 1) + (\theta - 1) \cdot  \theta\Delta^2 \leq \theta^2\Delta^2$, where the inequality holds because $\theta \geq 1$ and $\Delta \geq 2$.
Hence, the algorithm outputs a 2-block-tree with tree size $\ell \geq \lfloor |C|/(\theta^2\Delta^2) \rfloor$.
If $\ell > \lfloor |C|/(\theta^2\Delta^2) \rfloor$, we just take the prefix of length $\lfloor |C|/(\theta^2\Delta^2) \rfloor$.
\end{proof}






\section{Deterministic counting via derandomising the AJ algorithm} \label{sec:AJ}

In this appendix, we apply the generic derandomising argument in \Cref{sec-reduction} to the Anand-Jerrum (AJ) algorithm \cite{AJ22}, and hence provides deterministic approximate counting algorithms for spin systems on sub-exponential neighbourhood growth graphs. 
We start with some basic definitions. 

A spin system is specified by a tuple $\mathcal{S}=(G=(V,E),[q],{\bm h},{\bm A})$, 
where $G$ is a graph, $[q]$ is the spins that each vertex may take, ${\bm h}\in\mathbb{R}_{\geq 0}^{q}$ is the \emph{external field} of the system, and ${\bm A}\in\mathbb{R}_{\geq 0}^{q\times q}$ is the \emph{interaction matrix}. 
A \emph{configuration} $\sigma\in[q]^V$ assigns each vertex a spin amongst $[q]$. 
The \emph{weight} of each configuration $\sigma$ is defined by
\[
w(\sigma):=\prod_{v\in V}{\bm h}(\sigma_v)\prod_{(u,v)\in E}{\bm A}(\sigma_u,\sigma_v). 
\]
The \emph{partition function} of the system $Z(\mathcal{S})$ is the sum of weight over all possible configurations, 
and the \emph{Gibbs distribution} $\mu$ is defined where the probability that each configuration is drawn is proportional to its weight. 
Namely, 
\[
Z(\mathcal{S}):=\sum_{\sigma\in [q]^V}w(\sigma)
\qquad\text{ and }\qquad
\mu(\sigma):=\frac{w(\sigma)}{Z(\mathcal{S})}.
\]
The spin system captures a lot of other counting problems. 
For example, when the external field is an all-$1$ vector and the interaction matrix has $0$'s on the diagonal and $1$'s off the diagonal, the partition function counts the number of proper $q$-colourings of the underlying graph. 

It is useful to consider the same spin system but with some vertices pinned to take some certain spins. 
A configuration $\sigma \in [q]^V$ is called feasible if $w(\sigma) > 0$.
For $\Lambda \subset V$, we say a \emph{partial configuration} $\tau_{\Lambda}$ over the subset $\Lambda$ is feasible, if it can be extended to a feasible configuration by assigning all other vertices with the spins.
Denote by $\mu^{\sigma_{\Lambda}}$ the distribution over $[q]^V$ obtained from the Gibbs distribution $\mu$ conditional on the partial configuration $\sigma_\Lambda$. 
For any $S \subseteq V$, denote by $\mu^{\sigma_{\Lambda}}_S$ the marginal distribution on $S$ projected from $\mu^{\sigma_{\Lambda}}$. 
If $S = \{v\}$ is a singleton, we abbreviate the notation a bit as $\mu^{\sigma_{\Lambda}}_v$ for simplicity.
\begin{definition}[marginal lower bound]
A distribution is said to take a marginal lower bound $b$, if for any $\Lambda \subset V$, $v\notin \Lambda$, any feasible partial configuration $\sigma_\Lambda$ on $\Lambda$, and any spin $i\in[q]$, it holds that
\[
\text{either }\qquad
\mu_{v}^{\sigma_{\Lambda}}(i)\geq b
\qquad\text{ or }\qquad
\mu_{v}^{\sigma_{\Lambda}}(i)=0.
\]
\end{definition}

To give a precise characterisation where the deterministic counting algorithm works, we need the definition of the \emph{strong spatial mixing}, a notion that is of great interest in the study of spin systems. 
\begin{definition}[strong spatial mixing]
Let $\delta: \mathbb{N} \to \mathbb{R}_{\geq 0}$ be a non-increasing function.
A spin system $\+S =(G = (V,E),[q],{\bm h},{\bm A})$ is said to exhibits \emph{strong spatial mixing} with rate $\delta$, 
if for any vertex subset $\Lambda \subseteq V$, any other vertex $v \notin \Lambda$, and any two feasible partial configurations $\sigma,\tau \in [q]^{\Lambda}$ with $\ell = \min \{ \dist_G(u,v) \mid u \in \Lambda \text{ and } \sigma_u \neq \tau_u \}$,
\begin{align*}
	\DTV{\mu^\sigma_v}{\mu^\tau_v} \leq \delta(\ell).
\end{align*}
In particular, if $\delta(\ell)= \alpha \exp (-\beta \ell)$ for some constants $\alpha,\beta > 0$, the spin system $\+S$ then is said to exhibit the strong spatial mixing with exponential decay. 
\end{definition}
In the context of this paper, the decay rate is always exponential. 
We shall abbreviate the term ``strong spatial mixing with exponential decay'' just as SSM. 

We also need some assumptions on the graph. 
Let $G=(V,E)$ be a graph. 
For any vertex $v \in V$, any integer $\ell > 0$, let 
\begin{align*}
S_{\ell}(v) = \{u \in V \mid \dist_G(u,v) = \ell\}
\end{align*}
be the \emph{sphere} of radius $\ell$ centred at $v$, where $\dist_G(u,v)$ denotes the length of the shortest path between $u$ and $v$ in the graph $G$.
\begin{definition}[sub-exponential neighbourhood growth]
Let $s: \mathbb{N} \to \mathbb{N}$ be a sub-exponential growth function, namely, satisfying $s(\ell) = \exp(o(\ell))$.\footnote{The notation $s(\ell) = \exp(o(\ell))$ means for any $c > 1$, there exists an integer $N$ such that for all $\ell \geq N$, $s(\ell) < c^\ell$.}
A graph $G=(V,E)$ is said to have \emph{sub-exponential neighbourhood growth} (SENG) within rate $s$, if $|S_{\ell}(v)| \leq s(\ell)$ for any vertex $v$ and any integer $\ell > 0$.
\end{definition}
We remark that the maximum degree of a SENG graph $G$ is $\Delta \leq s(1) = O(1)$. 
In statistical mechanism, SENG graphs are ubiquitous; for example, the $d$-dimensional integer lattice $\mathbb{Z}^d$ is heavily studied. 

We now state the result for spin systems. 
\begin{condition}\label{cond-spin-requirement}
A tuple $(q,\delta,s)$, where $\delta(\ell)$ and $s(\ell)$ are two functions, is said to satisfy this condition with constant $L=L(q,\delta,s)$, if
\begin{equation}\label{equ-spin-requirement}
    2eq(1+s(\ell))\delta(\ell)\leq 1\text{\qquad holds for all }\ell\geq L.  
\end{equation} 
\end{condition}
    
\begin{theorem}\label{thm-spin-main}
There exists a deterministic algorithm such that,
for any $q\geq 2$, $\delta$ and $s$ satisfying \Cref{cond-spin-requirement} with constant $L$, 
and for any $q$-spin system on SENG graph with growth rate $s$ exhibiting SSM with decay rate $\delta$,
the algorithm outputs an estimate to the partition function of the spin system with multiplicative error $(1\pm\varepsilon)$ in time $O((\frac{n}{b\varepsilon})^{O(s(L)\log q)})$, where $b$ is the marginal lower bound. 
\end{theorem}
In practice, the interaction matrix and the external field are fixed and considered as constants, so that the actual input only contains the graph itself. 
With the assumption of SENG graphs, the marginal lower bound is a constant $b=b({\bm h},{\bm A},s(1))$ independent of the input graph. 

\subsection{Implications of \texorpdfstring{\Cref{thm-spin-main}}{Theorem B.5}}

As one of the important corollaries of \Cref{thm-spin-main}, we provide an FPTAS for the number of (proper) colourings in SENG graphs. 
\begin{theorem} \label{thm-spin-colouring}
Let $q\geq 3$ and $\Delta \geq 3$ be two constants satisfying  $q \geq (\frac{11}{6}-\epsilon_0)\Delta$, where $\epsilon_0 > 10^{-5}$ is a universal constant.
There is an FPTAS for the number of proper $q$-colourings in SENG graphs with maximum degree at most $\Delta$.
\end{theorem}
Prior to this work, an FPTAS for the number of proper $q$-colourings on a graph of maximum degree $\Delta$ is known only when $q\geq 2\Delta$ \cite{LSS22}. 
This result pushes forward the condition where FPTAS exists to that for FPRAS \cite{chen2019improved}, albeit only on SENG graphs. 


Specialised to colouring lattices graphs, one can prove SSM in a more refined regime than applying general theorems for colourings. 
The known SSM results, together with our main theorem for spin systems \Cref{thm-spin-main}, imply the following theorem. 
\begin{theorem}
There is an FPTAS for the number of proper $q$-vertex-colourings for any finite subgraph of the following lattice graphs: 
\begin{itemize}
\item $\mathbb{Z}^2$ lattice, if $q\geq 6$; \hfill (cf. \cite{AMMVB05})
\item $\mathbb{Z}^3$ lattice, if $q\geq 10$; \hfill (cf. \cite{GMP05})
\item Triangular lattice, if $q\geq 10$; \hfill (cf. \cite{GMP05})
\item Honeycomb lattice, if $q\geq 5$; \hfill (cf. \cite{GMP05})
\item Kagome lattice, if $q\geq 5$. \hfill (cf. \cite{Jal09})
\end{itemize}
\end{theorem}

\subsection{Truncated AJ algorithm: proof of \texorpdfstring{\Cref{thm-spin-main}}{}}

\begin{algorithm}  
  \caption{$\ssmstr{}_T(\+S,(\Sigma,\sigma),v,\ell)$} \label{Alg:bounded-ssms}
  \KwIn{a spin system $\+S=(G,[q],{\bm h},{\bm A})$, a set of vertices $\Sigma\subseteq V$ with a configuration $\sigma\in \Omega_{\Sigma}$, a vertex to sample $v\notin \Sigma$, and a distance $\ell\in \mathbb{N}$}
  \KwOut{the partial configuration passed in with a spin at
$v$: $(\Sigma, \sigma) \oplus (v, i)$ for some $i \in [q]$.}
\SetKwIF{Try}{Catch}{Exception}{try}{:}{catch}{exception}{}
\uTry{}{
Decrease the global timer $T\gets T-1$\;
%
\For{$i\in [q]$}{
    $p_{v}^i\gets \min_{\tau\in \Omega_{S_{\ell}\setminus \Sigma}}\mu^{\sigma\oplus \tau}(i)$\;
}
$p_{v}^0\gets 1-\sum_{i\in [q]}p_{v}^i$\;
Sample a random value $X\in \{0,1,\dots,q\}$ with $\Pr{X=i}=p_v^i$ for each $0\leq i\leq q$\;
\If{$X=0$}{
\If{$T>0$}{
  $(\rho_1,\rho_2,\dots,\rho_q)\gets \bdsplittr{}(\+S,(\Sigma,\sigma),v,\ell,(p_v^{0},p_v^{1},p_{v}^2,\dots,p_{v}^q))$\;
  Sample a random value $Y\in [q]$ with $\Pr{Y=i}=\rho_i$ for each $1\leq i\leq q$\;  
}
\Return $((\Sigma,\sigma)\oplus(v,Y))$\;
}
\Else
{
    \Return $((\Sigma,\sigma)\oplus(v,X))$\;
}
}
\Catch{$T<0$}{
  Terminate all instances of $\ssmstr{}$ and $\bdsplittr{}$ routines and set the vertex the first call is on to a random feasible colour\; \label{alg-line-truncate} 
}
\end{algorithm} 

\begin{algorithm}  
  \caption{$\bdsplittr{}(\+S,(\Sigma,\sigma),v,\ell,(p_v^0,p_v^1,p_v^2,\dots,p_v^q))$} \label{Alg:bdsplit}
  \KwIn{a spin system $\+S=(G,[q],{\bm h},{\bm A})$, a set of vertices $\Sigma\subseteq V$ with a configuration $\sigma\in \Omega_{\Sigma}$, a vertex to sample $v\notin \Sigma$, a distance $\ell\in \mathbb{N}$, and a probability distribution $(p_v^0,p_v^1,p_v^2,\dots,p_v^q)$}
  \KwOut{a probability distribution $(\rho_1,\rho_2,\dots,\rho_q)$}
    Give $S_{\ell}(v)\setminus \Sigma$ an arbitrary ordering $S_{\ell}(v)\setminus \Sigma=\{w_1,w_2,\dots,w_m\}$\;
    $(\Sigma',\sigma')\gets (\Sigma,\sigma)$\;
    \For{$1\leq j\leq m$}{
        $(\Sigma',\sigma')\gets \ssmstr{}(\+S,(\Sigma',\sigma'),w_j,\ell)$\; \label{alg-line-subroutine}
    }
    \For{$i\in [q]$}{
    $\rho_i\gets (\mu^{\sigma'}_{v}(i)-p_v^i)/p_{v}^0$\;
    }
    \Return $(\rho_1,\rho_2,\dots,\rho_q)$\;
\end{algorithm} 

To derandomise the AJ algorithm, we truncate it \`{a} la \Cref{sec:approx-resolve}. 
This results in \Cref{Alg:bounded-ssms} and \Cref{Alg:bdsplit}. 
The algorithm maintains a global timer $T$ shared by all copies of the process. 
We use a subscript to represent what the timer is initially set to, e.g., $\ssmstr_{T}$ or $\ssmstr_{\infty}$. 
For simplicity in notation, as the input spin system $\+S$ and the parameter $\ell$ are never changed throughout the algorithm, we drop it from the list of parameters. 

We remark that this is slightly different from the ``bounded'' variant appeared in \cite{AJ22} as that truncates the algorithm by a certain depth instead of the number of calls. 
The correctness of the untruncated algorithm is summarised as follows. 
\begin{theorem}[{\cite[Theorem 5.3]{AJ22}}] \label{thm-aj-correctness}
Let $\mathcal{S}$ be a spin system exhibiting SSM, $\mu$ be its Gibbs distribution, and $\ell\geq 1$ be an integer. 
If the untruncated AJ algorithm, i.e., $\ssmstr{}_{\infty}((\Sigma,\sigma),v)$ with the global timer set to $\infty$, terminates with probability $1$, 
then it generates a spin of $v$ subject to the correct conditional distribution upon terminating, i.e., 
\[
\Pr{\ssmstr{}_{\infty}((\Sigma,\sigma),v)=i}=\mu_{v}^{\sigma}(i), 
\]
providing the partial configuration $(\Sigma,\sigma)$ is feasible.
\end{theorem}

As a corollary of the above theorem, the truncated variant outputs a spin close to the correct distribution, the distance of which depends on the probability that \Cref{alg-line-truncate} is executed. 
\begin{corollary} \label{cor-aj-correctness}
Let $\mathcal{S}$ be a spin system exhibiting SSM, $\mu$ be its Gibbs distribution, and $T,\ell\geq 1$ be integers. If the untruncated AJ algorithm $\ssmstr{}_{\infty}((\Sigma,\sigma),v)$ terminates with probability $1$. 
Then, suppose $Y$ is the distribution of the output of $\ssmstr{}_{T}((\Sigma,\sigma),v)$ with the global timer initially set to $T$, it holds that
\[
  \DTV{Y}{\mu_{v}^{\sigma}}\leq\Pr{\text{\Cref{alg-line-truncate} is executed}}.
\]
\end{corollary}

\begin{proof}
For simplicity, let $\mathcal{A}$ be the truncated algorithm $\ssmstr{}_{T}((\Sigma,\sigma),v)$ and $\mathcal{A}'$ be the untruncated algorithm $\ssmstr{}_{\infty}((\Sigma,\sigma),v)$. 
Denote the bad event that \Cref{alg-line-truncate} is executed by $\mathcal{B}$.
For any spin $j\in[q]$, by \Cref{thm-aj-correctness}, 
\[
\mu_{v}^{\sigma}(j)=\Pr{\mathcal{A}'=j\mid\neg\mathcal{B}}\Pr{\neg\mathcal{B}}+\Pr{\mathcal{A}'=j\mid\mathcal{B}}\Pr{\mathcal{B}}, 
\]
and by definition,  
\[
Y(j)=\Pr{\mathcal{A}=j\mid\neg\mathcal{B}}\Pr{\neg\mathcal{B}}+\Pr{\mathcal{A}=j\mid\mathcal{B}}\Pr{\mathcal{B}}. 
\]
Naturally, $\Pr{\mathcal{A}'=j\mid\neg\mathcal{B}}=\Pr{\mathcal{A}=j\mid\neg\mathcal{B}}$, which gives
\begin{align*}
\left|\mu_{v}^{\sigma}(j)-Y(j)\right|&=\Pr{\mathcal{B}}\cdot\left|\Pr{\mathcal{A}'=j\mid\mathcal{B}}-\Pr{\mathcal{A}=j\mid\mathcal{B}}\right|\\
&\leq \Pr{\mathcal{B}}\cdot\left(\Pr{\mathcal{A}'=j\mid\mathcal{B}}+\Pr{\mathcal{A}=j\mid\mathcal{B}}\right)    
\end{align*}
Therefore, 
\[
\DTV{Y}{\mu_{v}^{\sigma}}=\frac{1}{2}\sum_{j\in[q]}\left|\mu_{v}^{\sigma}(j)-Y(j)\right|\leq\Pr{\mathcal{B}}. \qedhere
\]
\end{proof}

To analyse the probability of truncation, we treat the algorithm as a branching process. 
There are two possible scenarios each time when an iteration of $\ssmstr{}$ is invoked: either this iteration vanishes without invoking any iteration, or the algorithm calls $\bdsplittr{}$ which leads to the creation of at most $\Delta^\ell$ new copies of $\ssmstr{}$. 
With the presence of strong spatial mixing, the branching process is likely to terminate in finite time, which is captured by the following lemma. 
\begin{lemma} \label{aj-running-time-whp}
Suppose $(q,\delta,s)$ is a tuple satisfying \Cref{cond-spin-requirement} with constant $L$. 
Let $0<\varepsilon'<1$ be a real, and $T=\frac{1+s(L)}{\log 2}\cdot\log\frac{2}{\varepsilon'}$.  
Given a $q$-spin system $\mathcal{S}=(G,[q],{\bm h},{\bm A})$ as an input, where $\mathcal{S}$ exhibits SSM with decay rate $\delta$ and $G$ is a SENG graph with rate $s$, 
the probability that \Cref{alg-line-truncate} in $\ssmstr{}_{T}((\Sigma,\sigma),v)$ is executed is bounded from above by $\varepsilon'$ if the global timer is initially set to $T$. 
\end{lemma}

One of the main ingredients for the above argument is to bound the length of ZOI, which is also the probability for branching. 
This is done by the following lemma \cite[Lemma 4.2]{AJ22}. 
\sloppy
\begin{lemma} \label{lem-aj-zoi-length}
Let $\mathcal{S}$ be a $q$-spin system exhibiting SSM with rate $\delta$. 
In the execution of $\ssmstr{}_T((\Sigma,\sigma),v)$ where the partial assignment $\sigma$ over $\Sigma$ is feasible, it holds that
\[
p_{v}^0\leq q\cdot\delta(\ell). 
\]
\end{lemma}

Given the exponential tail bound, \Cref{thm-spin-main} can be proved. 

\begin{proof}[Proof of \Cref{thm-spin-main}]
Each random number drawn in the algorithm can take $\{0,1,\cdots,q\}$, and hence the domain is of size $q+1$. 
The worst case running time is $O(qT)$ which is straightforward to show; 
note that the cost of \Cref{alg-line-subroutine} in $\bdsplittr{}$ is amortised into each subroutine. 
The number of random numbers drawn throughout the process is at most $2T$. 
The conditions in \Cref{cor-reduction} get fulfilled if we set $\varepsilon'=\frac{b\varepsilon}{10n}$ where $b$ is the marginal lower bound, which gives $T=\frac{1+s(L)}{\log 2}\cdot\log\frac{20n}{b\varepsilon}$. 
The theorem then follows by applying \Cref{cor-reduction}. 
\end{proof}

\subsection{Exponential tail on running time: proof of \texorpdfstring{\Cref{aj-running-time-whp}}{}}

In this section we give a proof of \Cref{aj-running-time-whp}. 
Each time the algorithm recurses into $\bdsplittr{}$, it creates at most $s(\ell)$ new copies of the routine $\ssmstr{}$. 
Such branching happens with probability $p_{v}^0$ which is at most $q\cdot\delta(\ell)$ due to \Cref{lem-aj-zoi-length}. 
This leads us to analysing a Markov process that stochastically dominates the actual branching process. 

Consider the following discrete Markov chain $(X_t)$ where $X_t\in\mathbb{Z}_{\geq 0}$. 
At the beginning, $X_0=1$. 
The chain has an absorbing barrier at $0$, and for other $X_t>0$, the transition is given by
\begin{equation} \label{equ-gwb-def}
X_{t+1}\leftarrow\begin{cases}
X_t+D & \text{with probability } p; \\
X_t-1 & \text{with probability } 1-p. 
\end{cases}
\end{equation}
Intuitively, if $pD<0.99$, in expectation the random walk moves towards the absorbing barrier, and the process terminates in constant time. 
But for analysing the truncated algorithm, a tail bound on the event that the process does not terminate for a long time is required. 
This is shown as the next lemma. 
\begin{lemma} \label{lem-gwb-terminate}
Suppose $2\mathrm{e}(1+D)p\leq 1$ and $0<\varepsilon'<1$. 
Let $T=\frac{1+D}{\log 2}\cdot\log \frac{2}{\varepsilon'}$. 
With probability at most $\varepsilon'$, the process $(X_t)$ defined by (\ref{equ-gwb-def}) does not terminate in $T$ rounds.  
\end{lemma}

We need to use the following notion of the generalised Dyck path. 
\begin{definition}[$(D+1)$-Dyck path]
A sequence $a_1,a_2,\cdots,a_{k}\in\{+D,-1\}$ forms a $(D+1)$-Dyck path of length $k$, 
if
\[
\forall j, \sum_{i=1}^{j}a_i\geq 0\text{\qquad and \qquad}\sum_{i=1}^{k}a_i=0. 
\]
\end{definition}
The number of $(D+1)$ Dyck path of given length is known to be the Fuss-Catalan number (see, for example, \cite{Ava08}). 
\begin{lemma}
The number of $(D+1)$-Dyck paths of length $(D+1)N$ is $\frac{1}{DN+1}\binom{(D+1)N}{N}$. 
\end{lemma}

\begin{proof}[Proof of \Cref{lem-gwb-terminate}]
If the process $(X_t)$ terminates exactly at $(T+1)$-th round, namely $X_{T}=1$ and $X_{T+1}=0$, then it must hold that the sequence $\{X_{t+1}-X_{t}\}_{t=0}^{T-1}$ forms a dyke path, and $X_{T+1}=0$. 
Obviously $T$ is a multiple of $D+1$. 
Let $N:=T/(D+1)$.
Then there are $N$ ``move-ups'', each with probability $p$, and $DN+1$ ``move-downs'', each with probability $1-p$. 
The extra plus one is owing to the move $X_T=1$ to $X_{T+1}=0$. 
Using the count on $(D+1)$-Dyck paths of length $T=(D+1)N$, the probability that $(X_t)$ terminates at $(T+1)$-th round is exactly
\[
w_N=\frac{1}{DN+1}\binom{(D+1)N}{N}p^{N}(1-p)^{DN+1}. 
\]
Using the inequality $\binom{n}{k}\leq(\mathrm{e}n/k)^k$, we have
\begin{align*}
w_N\leq\frac{1}{DN+1}\cdot \left(\mathrm{e}(1+D)\right)^{N} \cdot p^{N}(1-p)^{DN+1}\leq\left(\mathrm{e}(1+D)p\right)^{N}. 
\end{align*}

Let $\gamma:=\mathrm{e}(1+D)p$ which is at most $1/2$. 
The probability that the process does not terminate in $T=(D+1)N$ rounds is
\begin{equation*}
\Pr{(X_t)\text{ does not terminate in $(D+1)N$ rounds}}=\sum_{i=N}^{\infty}w_i\leq
\gamma^{N}\sum_{j=0}^{\infty}\gamma^{j}<2\gamma^{N}<2\left(\frac{1}{\varepsilon'}\right)^{\frac{\log(\gamma)}{\log 2}}<\varepsilon'. \qedhere
\end{equation*}
\end{proof}

\begin{proof}[Proof of \Cref{aj-running-time-whp}]
Let $p=q\cdot\delta(\ell)$ and $D=s(\ell)$
Suppose we are now running $\ssmstr{}_{\infty}$ with the timer set to infinity. 
We might alter \Cref{alg-line-subroutine} in $\bdsplittr{}$ a bit, by registering first the $m$ instances of $\ssmstr{}$ to run, and then setting the appropriate parameter and invoking the routine. 
We call an instance of $\ssmstr{}$ \emph{active}, if it is running or has been registered but not yet run. 
Let $Y_t$ be the number of active instances upon the invocation of the $(t+1)$-th instance. 
If there is no such $(t+1)$-th instance, set $Y_t=0$. 
For each instance, it dies out and becomes no more active if it does not recurse into $\bdsplittr{}$, and upon the invocation of the next instance, the count goes down by one. 
Or otherwise, it registers up to $S_{\ell}(v)$ new instances, so that prior to the next invocation (the current instance has not died out yet) the count goes up by the number of newly-created instances. This happens with probability $p_v^0\leq p$, the length of ZOI. 

Despite the fact that the $i$-th invocation may change the boundary condition for later invocations, \Cref{lem-aj-zoi-length} holds for all boundary conditions, and because each invocation use fresh random numbers, the probability that $x$ invocations create copies and $y$ invocations do not is bounded by $p^x(1-p)^y$. 
Then it is a simple observation that
\[
\Pr{(Y_t)\text{ does not terminate in $T$ rounds}}\leq\Pr{(X_t)\text{ does not terminate in $T$ rounds}}. 
\]
This is because each time it branches out $D'\leq D$ new instances. 
We can just treat this as if it created $D$ copies, but the last $D-D'$ ones were bound to die out. 
The lemma then follows by applying \Cref{lem-gwb-terminate} with the aforementioned choices of $p$, $D$ and $T$. 
\end{proof}

\subsection{FPTAS from optimal temporal mixing}

As a well-known result, the notion of SSM is equivalent to \emph{optimal temporal mixing} of the Glauber dynamics on SENG graphs \cite{DSVW04}. 
The latter notion receives a glaring attention in recent study of Markov chains. 

Fix a subset of vertices $\Lambda \subseteq V$. Define its boundary by
\begin{align*}
	\partial \Lambda \defeq \{ u \in V \setminus \Lambda \mid \exists w \in \Lambda \text{ s.t. } \{u,w\} \in E \}.
\end{align*}
Let $(X_t)_{t \geq 0}$ and $(Y_t)_{t \geq 0}$ be two instance of the Glauber dynamics, where $(X_t)_{t \geq 0}$ and $(Y_t)_{t \geq 0}$ may start from different initial configurations $X_0$ and $Y_0$. 
The optimal temporal mixing is defined as follows. 
\begin{definition}[optimal temporal mixing] \label{def-optimal-temporal-mixing}
The Glauber dynamics is said to have the optimal temporal mixing under arbitrary pinning, if there exists $\gamma,\zeta > 0$ such that for any vertex set $\Lambda \subseteq V$, any feasible boundary condition $\sigma \in [q]^{\partial \Lambda}$, and any two instances $(X_t)_{t \geq 0}$ and $(Y_t)_{t \geq 0}$ of the Glauber dynamics on $\Lambda$ with boundary configuration $\sigma$, it holds that 
\begin{align*}
	\forall k \in \mathbb{N}, \quad \DTV{X_{kn}}{Y_{kn}} \leq |\Lambda| \gamma \exp(-\zeta k). 
\end{align*}
\end{definition}
The above definition also implies $O(n \log n)$ mixing time of the Glauber dynamics, by observing that starting from an arbitrary $X_0$, $\DTV{X_T}{\mu^\sigma_\Lambda} \leq 1/4$ for $T = O(n \log n)$. 

One of the main results in \cite{DSVW04} states: 
\begin{theorem}[{\cite[Theorem 2.3]{DSVW04}}] \label{thm-dsvw04}
If the spin system poses optimal temporal mixing on SENG graphs, then the system exhibits SSM. 
\end{theorem}
We remark that the definition of SSM varies slightly across the literature. 
However, the version in this paper is equivalent to that in \cite{DSVW04} on SENG graphs, and the above theorem still applies in our context. 
Combining \Cref{thm-spin-main} and \Cref{thm-dsvw04} yields: 
\begin{theorem}\label{thm-time}
Let ${\bm A}\in\mathbb{R}^{q\times q}_{\geq 0}$ and ${\bm h}^{q}_{\geq 0}$ be an interaction matrix and an external field vector. 
There is an FPTAS for the partition functions of the spin system defined by ${\bm A}$ and ${\bm h}$ on SENG graphs, if the Glauber dynamics on the Gibbs distribution has the optimal temporal mixing under arbitrary pinning.
\end{theorem}

\subsection{FPTAS from spectral independence}

To derive optimal temporal mixing and hence apply \Cref{thm-time}, 
we utilise a powerful tool called the \emph{spectral independence}, first defined by Anari, Liu and Oveis Gharan \cite{ALO20} to obtain rapid mixing, extended to general $[q]$ domains by various authors \cite{CGSV21,FGYZ21}, and refined by Chen, Liu and Vigoda \cite{CLV21} for optimal mixing of the Glauber dynamics.

The formal definition of spectral independence is given below. 
For any subset $\Lambda \subset V$, any $\sigma \in \Omega(\mu_\Lambda)$, where $\Omega(\mu_\Lambda)$ denotes the support of the distribution $\mu_{\Lambda}$, define 
\begin{align*}
	\widetilde{V}_{\sigma} = \left\{ (u,c) \mid u \in V \setminus \Lambda \text{ and } c \in \Omega(\mu^\sigma_v) \right\}.
\end{align*}
For every pair $(u,i),(v,j) \in \widetilde{V}_{\sigma}$ with $u \neq v$, define the (signed) influence from $(u,i)$ to $(v,j)$ with respect to the conditional $\sigma$ by
\begin{align*}
	\gamma^\sigma_\mu((u,i),(v,j)) = \mu_{v}^{\sigma \land (u \gets i)}(j) - \mu_v^\sigma(j),
\end{align*}
and define $\gamma^\sigma_\mu((v,i),(v,j)) = 0$ for all $(v,i),(v,j) \in \widetilde{V}_{\sigma}$, where $ \mu_{v}^{\sigma \land (u \gets i)}$ denotes the marginal distribution on $v$ conditional on $\sigma$ and the event that $u$ takes the value $i$. 
\begin{definition}[spectral independence]
Let $\eta > 0$ be a constant.
A distribution $\mu$ is $\eta$-spectrally independent if for any $\Lambda \subset V$, any $\sigma \in \Omega(\mu_{\Lambda})$, the maximum eigenvalue of $\gamma^\sigma_\mu$ satisfies $\lambda_{\max}(\gamma^\sigma_\mu) \leq \eta$.
\end{definition}

The main result of this subsection is the following theorem. 
\begin{theorem} \label{thm-si-to-fptas}
Let ${\bm A}\in\mathbb{R}^{q\times q}_{\geq 0}$ and ${\bm h}^{q}_{\geq 0}$ be an interaction matrix and an external field vector. 
There is an FPTAS for the partition functions of the spin system defined by ${\bm A}$ and ${\bm h}$ on SENG graphs, if the Gibbs distribution of the spin system is $\eta$-spectrally independent for some constant $\eta$. 
\end{theorem}

The work \cite{CLV21} establishes that, under some conditions, spectral independence implies approximate tensorisation of entropy, a notion that is used to establish the decay of relative entropy and hence optimal temporal mixing (see, for example, \cite{CMT15}). 
Connecting the main theorem of \cite[Theorem 2.8]{CLV21} (see also \cite[Theorem 1.7]{BCCPSV22}) with \cite{CMT15} yields the following. 
\begin{theorem} \label{thm-clv-si-at}
Let $\Delta\geq 3$ be an integer and $b,\eta>0$ be reals. 
Let $\mu$ be the distribution of a $q$-spin system on a graph $G=(V,E)$ of maximum degree at most $\Delta$. 
If $\mu$ is $\eta$-spectrally independent and $b$-marginally bounded, then the relative entropy of $P_{\mathrm{GL}}$ decays with rate $\frac{1}{C|V|}$, i.e.,  
\[
\KL{\nu P_{\mathrm{GL}}}{\mu}\leq\left(1-\frac{1}{C|V|}\right)\KL{\nu}{\mu}
\]
holds for any distribution $\nu$ over $[q]^V$.
Here the constant $C=\left(\frac{\Delta}{b}\right)^{1+2\left\lceil\frac{\eta}{b}\right\rceil}$. 
\end{theorem}
\Cref{thm-si-to-fptas} can be proved by iterating the above theorem. 
We also remark that spectral independence holds under arbitrary pinning. 

\begin{proof}[Proof of \Cref{thm-si-to-fptas}]
Fix the vertex set $\Lambda$. 
Under arbitrary feasible pinning, we apply \Cref{thm-clv-si-at} iteratively. 
The distribution $X_{kn}$ of the Glauber dynamics on $\Lambda$ after $kn$ rounds then satisfies
\begin{align}
\KL{X_{kn}}{\mu^{\sigma}_{\Lambda}}&\leq\left(1-\frac{1}{C|\Lambda|}\right)^{kn}\KL{X_0}{\mu^{\sigma}_{\Lambda}}\nonumber\\
&\leq \exp\left\{-\frac{k}{C}\right\}\KL{X_0}{\mu^{\sigma}_{\Lambda}} \tag{Using $|\Lambda|\leq n$}\nonumber\\
&\leq \exp\left\{-\frac{k}{C}\right\}\log\left(\frac{1}{\mu^{\sigma}_{\Lambda,\min}}\right) \label{equ-ent-decay-iterate}
\end{align}
where $\mu^{\sigma}_{\Lambda,\min}=\min_{\tau\in[q]^{\Lambda}}\mu^{\sigma}_{\Lambda}(\tau)$ is the minimum non-zero probability of the distribution $\mu^{\sigma}_{\Lambda}$.

We verify optimal temporal mixing as follows. 
\begin{align*}
\DTV{X_{kn}}{Y_{kn}}&\leq\DTV{X_{kn}}{\mu^{\sigma}_{\Lambda}}+\DTV{Y_{kn}}{\mu^{\sigma}_{\Lambda}}\\
&\leq 2\DTV{X_{kn}}{\mu^{\sigma}_{\Lambda}} \tag{w.l.o.g.}\\
&\leq\sqrt{2\KL{X_{kn}}{\mu^{\sigma}_{\Lambda}}} \tag{Pinsker's Inequality}\\
&\leq\sqrt{2}\sqrt{\log\left(\frac{1}{\mu^{\sigma}_{\Lambda,\min}}\right)}\exp\left\{-\frac{k}{2C}\right\}\tag{By (\ref{equ-ent-decay-iterate})}\\
&\leq\sqrt{2|\Lambda|\log(1/b)}\exp\left\{-\frac{k}{2C}\right\}. 
\end{align*}

The theorem then follows by invoking \Cref{thm-time}. 
\end{proof}

Finally, \Cref{thm-spin-colouring} can be proved. 
\begin{proof}[Proof of \Cref{thm-spin-colouring}]
The uniform distribution over proper $q$-colourings in the same regime as \Cref{thm-spin-colouring} is shown to be $\eta$-spectrally independent for some constant $\eta=\eta(q,\Delta)$ \cite{Liu21}. 
This theorem then follows after \Cref{thm-si-to-fptas}. 
\end{proof}

\end{document}